\documentclass[parskip=full, titlepage=false, twoside=false]{scrartcl} 
\usepackage[utf8]{inputenc} 
\usepackage[T1]{fontenc} 
\usepackage[english]{babel}
\usepackage[autostyle=true]{csquotes}
\usepackage{lmodern}
\usepackage[backend=biber, maxnames=50, backref=true, backrefstyle=three]{biblatex}
\usepackage[colorlinks=true, bookmarks=true, pdfstartview=FitH, pdfauthor={Romain Gicquaud and Jonathan Glöckle}, pdftitle={On the topology of the space of vacuum initial data sets}]{hyperref} 
\usepackage{amsmath, amsthm, amssymb, bm, bbm, mathrsfs, mathtools, mathabx}
\usepackage{extarrows}
\usepackage{tabulary,tabularx}
\usepackage[usenames,dvipsnames]{xcolor}
\usepackage{graphicx}
\usepackage{tikz-cd}
\usepackage{enumerate}
\usepackage[capitalise]{cleveref}
\usepackage{bbold}
\usepackage{wasysym}
\usepackage{enumitem}

\theoremstyle{plain}

\newtheorem{thm}{Theorem}[section]
\newtheorem{theorem}[thm]{Theorem}

\newtheorem{corollary}[thm]{Corollary}

\newtheorem{lemma}[thm]{Lemma}

\newtheorem{proposition}[thm]{Proposition}

\newtheorem{claim}{Claim}

\theoremstyle{remark}
\newtheorem{example}[thm]{Example}
\newtheorem{remark}[thm]{Remark}
\newtheorem{observation}[thm]{Observation}

\theoremstyle{definition}

\newtheorem{definition}[thm]{Definition}

\newcounter{mnotecount}[section]

\renewcommand{\phi}{\varphi}
\renewcommand{\epsilon}{\varepsilon}

\newcommand{\bR}{\mathbb{R}}
\newcommand{\bL}{\mathbb{L}}

\newcommand{\bZ}{\mathbb{Z}}
\newcommand{\bN}{\mathbb{N}}

\newcommand{\cA}{\mathcal{A}}

\newcommand{\cC}{\mathcal{C}}

\newcommand{\cO}{\mathcal{O}}

\newcommand{\cW}{\mathcal{W}}

\newcommand{\cY}{\mathcal{Y}}

\newcommand{\ghat}{{\widehat{g}}}

\newcommand{\gbar}{\overline{g}}

\renewcommand{\hbar}{\overline{h}}

\newcommand{\cL}{\mathcal{L}}

\newcommand{\sigmatil}{\widetilde{\sigma}}

\newcommand{\Hring}{\mathring{H}}

\newcommand{\Sring}{\mathring{S}}
\newcommand{\kring}{\mathring{k}}

\let\<\langle
\let\>\rangle

\newcommand{\kulkoverlay}[2]{\ooalign{$#1\bigcirc$\cr$#1\mkern2mu\wedge$\cr}}
\newcommand{\kulk}{\mathbin{\mathpalette\kulkoverlay{}}}

\DeclareMathOperator{\tr}{tr}
\DeclareMathOperator{\divg}{div}

\DeclareMathOperator{\id}{Id}

\newcommand{\lie}{\mathcal{L}}

\DeclareMathOperator{\supp}{supp}

\newcommand{\riem}{\mathrm{R}}

\newcommand{\riemuddd}[4]{\riem^{#1}_{\phantom{#1} #2 #3 #4}}

\newcommand{\weyl}{\mathrm{W}}

\newcommand{\schouten}{\mathrm{P}}

\newcommand{\cotton}{\mathrm{C}}

\newcommand{\ric}{\mathrm{Ric}}

\newcommand{\scal}{\mathrm{Scal}}

\DeclareMathOperator{\hess}{Hess}

\def\XXint#1#2#3{{\setbox0=\hbox{$#1{#2#3}{\int}$}
			\vcenter{\hbox{$#2#3$}}\kern-.5\wd0}}

\newcommand{\arxiv}[1]{\textsc{arxiv}: \href{https://arxiv.org/abs/#1}{\nolinkurl{#1}}}
\bibliography{refs}

\let\div\relax

\newcommand{\del}{\partial}
\newcommand{\lto}{\longrightarrow}
\newcommand{\lmapsto}{\longmapsto}
\newcommand{\blank}{\,\cdot\,}

\newcommand{\upd}{\mathrm{d}}
\newcommand{\dotcup}{\mathbin{\dot\cup}}
\newcommand{\pt}{\{*\}}
\newcommand{\Ztwo}{\bZ/2\bZ}

\renewcommand{\emptyset}{\varnothing}
\renewcommand{\epsilon}{\varepsilon}

\renewcommand{\phi}{\varphi}
\renewcommand{\hat}{\widehat}

\DeclareMathOperator{\div}{div}

\DeclareMathOperator{\Fred}{Fred}
\DeclareMathOperator{\ind}{ind}

\DeclareMathOperator{\Cl}{Cl}

\DeclareMathOperator{\SO}{SO}

\DeclareMathOperator{\adiff}{\alpha--diff}
\DeclareMathOperator{\oladiff}{\overline{\alpha}--diff}
\DeclareMathOperator{\KO}{KO}

\newcommand{\Ini}{\mathcal{I}}
\newcommand{\sDEC}{\Ini^{>}}
\newcommand{\DEC}{\Ini^{\geq}}
\newcommand{\VacC}{\Ini^{\equiv\,0}}
\newcommand{\InvI}{\Ini^{\mathrm{inv}}}
\newcommand{\Met}{\mathcal{R}}
\newcommand{\PSC}{\Met^{>}}
\newcommand{\NCK}{\Met_0}
\newcommand{\PNCK}{\PSC_0}
\newcommand{\TT}{\mathrm{TT}}

\newcommand{\Susp}{\mathrm{Susp}}

\newcommand{\CMC}{\mathrm{CMC}}

\newcommand{\co}{\mathfrak{co}}

\newcommand{\cf}{cf.~}
\newcommand{\ie}{i.\,e.~}
\newcommand{\eg}{e.\,g.~}
\renewcommand{\geq}{\geqslant}
\renewcommand{\leq}{\leqslant}
\let\Im\relax
\DeclareMathOperator{\Im}{Im}

\begin{document}

\author{Romain Gicquaud\thanks{Université de Tours, Parc Grandmont, 37200 Tours, France, E-mail address: \href{mailto:romain.gicquaud@univ-tours.fr}{romain.gicquaud@univ-tours.fr}} \and Jonathan Glöckle\thanks{Institutionen för Matematik, KTH Stockholm, Lindstedtsvägen 25, 100 44 Stockholm, Sweden, E-mail address: \href{mailto:jonathan.gloeckle.math@outlook.de}{jonathan.gloeckle.math@outlook.de}}}
\title{On the topology of the space of vacuum initial data sets}
\date{}
\maketitle

\begin{abstract}
	We show that the space of vacuum initial data sets on a closed manifold often
	has many non-trivial homotopy groups. The starting point is a result of the
	second named author, which constructs non-trivial elements in the homotopy
	groups of the space of initial data sets satisfying the strict dominant energy
	condition, together with a later result in joint work with Bernd Ammann showing
	that these elements often persist when the strictness assumption is dropped. In
	this work, we use a parametrized version of the conformal method to show that
	these elements may also be represented by maps into the space of vacuum initial
	data sets. This requires two results that may be of independent interest:
	metrics admitting conformal Killing vectors can be removed from the space of
	metrics without changing its weak homotopy type, and over the remaining metrics
	the York decomposition can be carried out in families, the TT-tensors forming a
	trivial Hilbert bundle. To our knowledge, this is the first result on the
	global topology of this space.
\end{abstract}



\pagebreak

\tableofcontents

\section{Introduction}\label{secIntro}
In general relativity, spacetimes are modeled by Lorentzian manifolds
satisfying the Einstein field equations, and Choquet-Bruhat's initial value
formulation describes them through their Cauchy data on a spacelike
hypersurface, \cf\cite{ChoquetBruhat,Ringstroem}. Such an initial data set on a
closed $n$-dimensional manifold $M$ consists of a Riemannian metric $g$
together with a symmetric $2$-tensor $k$, to be thought of respectively as the
induced metric and the second fundamental form of the hypersurface. The pair
$(g,k)$ cannot be prescribed freely. It must satisfy the Einstein constraint
equations, which relate it to the energy density $\rho$ and the momentum
density $j$ of the ambient matter fields:
\begin{equation}\label{eqDensities}
	\left\lbrace
	\begin{aligned}
		\rho & = \tfrac{1}{2}\big(\scal^g + (\tr^g k)^2 - |k|_g^2\big), \\
		j    & = \divg^g(k) - \upd(\tr^g k),
	\end{aligned}
	\right.
\end{equation}
\cf\cite{BartnikIsenberg}. The data are called vacuum if $\rho$ and $j$ vanish
identically. They satisfy the dominant energy condition (DEC) if
$\rho \geq |j|_g$ everywhere, which is the condition met by physically
reasonable matter fields, \cf\cite[Ch.~9]{Wald}.

Building on work of the second named author~\cite{Gloeckle2024a}, we study the
global topology of the spaces of initial data sets that these conditions
define. We fix the smooth manifold $M$, equip tensor fields with the
$C^\infty$-topology, write $\Ini(M)$ for the space of all initial data sets,
and single out the subspaces
\[
	\VacC(M) \;\subset\; \DEC(M) \;\supset\; \sDEC(M)
\]
of vacuum initial data sets, of initial data sets satisfying the DEC, and of
those satisfying the strict DEC $\rho > |j|_g$, respectively.

The structure of $\VacC(M)$ has been studied since the 1970s, but almost
exclusively from a local point of view. The linearization stability program of
Fischer and Marsden~\cite{FischerMarsden} and Moncrief~\cite{Moncrief} showed
that $\VacC(M)$ is a smooth infinite-dimensional manifold near every initial
data set whose vacuum development carries no Killing vector, and Arms, Marsden
and Moncrief~\cite{ArmsMarsdenMoncrief} proved that, at these points, it has a
quadratic cone singularity. Bartnik~\cite{Bartnik} constructed a phase space
formulation adapted to this picture, and Fischer and
Moncrief~\cite{FischerMoncrief} studied the quotient by the diffeomorphism
group and its reduced dynamics. All these results describe $\VacC(M)$ locally.
To the best of our knowledge, nothing was known so far about its topology in
the large, and the present paper is a first step in this direction.

A closely related and notoriously difficult question is the parametrization
problem: describing the whole of $\VacC(M)$ by freely specifiable data. In the
constant mean curvature (CMC) regime, \ie in the case where $\tau = \tr^g k$ is
constant, the answer is complete. The conformal method, introduced by
Lichnerowicz~\cite{Lichnerowicz1944} and brought to its modern form by
Choquet-Bruhat and York~\cite{ChoquetBruhatYork}, parametrizes the CMC initial
data by a conformal class, the constant $\tau$ and a transverse and traceless
tensor $\sigma$, and Isenberg's classification~\cite{Isenberg} turns this into
an exact description of the CMC part of $\VacC(M)$, \cf the
surveys~\cite{BartnikIsenberg,CarlottoReview}. Away from the CMC regime,
despite decades of efforts, only partial existence results are available,
\cf\cite{MaxwellNonCMC,DahlGicquaudHumbert} and the references therein, and
there is strong evidence that the conformal parametrization itself degenerates.
Maxwell exhibited model problems in which it folds and
bifurcates~\cite{MaxwellModel}, and proposed drift
reformulations~\cite{MaxwellDrift} in which one gives up on a single global
parametrization and keeps instead enough freedom to change charts.

The results of this paper show that such difficulties are not defects of one
particular method, but topological features of $\VacC(M)$ itself. Every natural
space of freely specifiable data, such as sections of tensor bundles or
conformal classes, is contractible. Assume that $P \colon \mathcal{D} \to
	\VacC(M)$ is a continuous parametrization defined on such a space $\mathcal{D}$
and that it admits a homotopy section, \ie a continuous map $s \colon \VacC(M)
	\to \mathcal{D}$ with $P \circ s$ homotopic to the identity. Then all homotopy
groups of $\VacC(M)$ vanish. More generally, if $P$ is only defined on a subset
$\mathcal{D}' \subset \mathcal{D}$, for instance on the data for which a given
method has a unique solution depending continuously on the data, and admits a
homotopy section, then $s$ induces injective maps $\pi_\ell(\VacC(M)) \to
	\pi_\ell(\mathcal{D}')$. The topology of $\VacC(M)$ then has to be present in
$\mathcal{D}'$. Our main theorem below produces non-zero homotopy groups
$\pi_\ell(\VacC(M))$ in every degree $\ell \geq 1$ with $n + \ell \equiv 0, 1,
	2, 4 \mod 8$, for a large class of closed spin manifolds of dimension $n \geq
	6$. On such manifolds, a parametrization of the above kind can only be defined
on a set of data that is not contractible.

\begin{table}
	\begin{tabular}{c|cccc}
		               & $\tau = 0,\, \sigma \equiv 0$ & $\tau = 0,\, \sigma \not\equiv 0$ & $\tau \neq 0,\, \sigma \equiv 0$ & $\tau \neq 0,\, \sigma \not\equiv 0$ \\
		\hline
		$\cY(M,g) > 0$ & No                            & Yes                               & No                               & Yes                                  \\
		$\cY(M,g) = 0$ & Yes, but not unique           & No                                & No                               & Yes                                  \\
		$\cY(M,g) < 0$ & No                            & No                                & Yes                              & Yes                                  \\
	\end{tabular}
	\caption{Existence of a unique solution of the CMC conformal method, adapted from~\cite{Isenberg}}
	\label{tab:SolLich}
\end{table}

In the CMC case, this is what Isenberg's classification shows (\cf
\cref{tab:SolLich}). The conformal method requires $\sigma \not\equiv 0$ when
the Yamabe invariant of the conformal class is non-negative, and $\tau \neq 0$
when it is non-positive. When the Yamabe invariant vanishes, $\tau = 0$ and
$\sigma \equiv 0$, solutions exist but are only unique up to multiplication by
a positive constant. Since metrics with positive Yamabe invariant allow any
value of $\tau$ and the other ones only $\tau \neq 0$, the set of pairs $(g,
	\tau)$ allowed by these conditions is, up to homotopy, the model $\PSC(M)
	\times I \cup \Met(M) \times \del I$ of the suspension of $\PSC(M)$ used
in~\eqref{eqSuspDEC}. The homotopy classes of \cref{thmMain} are obtained by
applying the CMC conformal method to data of this kind (\cref{secHomotopy}).
Away from the CMC case, no analogue of Isenberg's classification is known. For
instance, the results of~\cite{GicquaudSmallTT,GicquaudConformal} show that, in
the far-from-CMC regime and for small~$\sigma$, the conformal method admits
either zero, one or two solutions whose size is controlled by~$\sigma$. It is
determined by an equation involving both $\sigma$ and $\tau$,
whereas \cref{tab:SolLich} only involves whether $\sigma$ vanishes. Our results
show that phenomena of this kind cannot be avoided entirely. Whatever form an
analogue of Isenberg's classification takes, it cannot be a continuous
parametrization, with a homotopy section, by a contractible set of data. Some
data have to be excluded, in such a way that the remaining set is not
contractible, or uniqueness or continuity have to fail. Both phenomena occur in
the CMC case (\cref{tab:SolLich}) as well as in the far-from-CMC results just
mentioned. This supports Maxwell's point of view that what one can hope for is
an atlas for $\VacC(M)$, not a global chart.

The tools we use to detect this topology come from the older and well-developed
study of the space $\PSC(M)$ of positive scalar curvature metrics. Beyond the
existence problem studied by Gromov and Lawson, Schoen and Yau, and
Stolz~\cite{GromovLawson1980,GromovLawson1980b,SchoenYau,Stolz}, a great deal
is known about its homotopy type, \cf the survey~\cite{Rosenberg}. The first
results go back to Hitchin~\cite{Hitchin}, who used the Dirac operator to
produce non-trivial elements in $\pi_0$ and $\pi_1$ of $\PSC(M)$. The modern
form of his idea is the \emph{index difference}, a map from the homotopy groups
of $\PSC(M)$ to the $\KO$-theory of a point. By the work of Hanke, Schick and
Steimle~\cite{HankeSchickSteimle}, of Crowley, Schick and
Steimle~\cite{CrowleySchickSteimle} and of Botvinnik, Ebert and
Randal-Williams~\cite{BotvinnikEbertRandalWilliams}, this index difference is
non-trivial in every degree in which its target does not vanish, so that
$\PSC(M)$ has many non-trivial higher homotopy groups whenever $M$ is spin of
dimension $n \geq 6$.

An analogue of this picture for initial data sets was developed
in~\cite{Gloeckle2024a}. There the Dirac operator is replaced by the
Dirac-Witten operator, which underlies Witten's spinorial proof of the positive
mass theorem~\cite{Witten,ParkerTaubes}, and one obtains a $\KO$-valued index
difference on $\sDEC(M)$ that detects non-trivial homotopy groups, again
produced from the corresponding classes on $\PSC(M)$. In later joint work with
Ammann~\cite{AmmannGloeckle}, these classes were shown to persist in the larger
space $\DEC(M)$ for every manifold whose topology forbids initial data sets
carrying a parallel spinor. In the present article we extend this program to
the space of vacuum initial data $\VacC(M)$. Concretely, we prove that, under
the same hypotheses, the homotopy classes
of~\cite{Gloeckle2024a,AmmannGloeckle} can already be represented by maps into
$\VacC(M)$, which is the content of \cref{thmMain} below.

Our main tool is a parametrized version of the conformal method discussed
above. In its constant mean curvature form, it builds a vacuum initial data set
out of a metric, a constant and a transverse and traceless (TT) ``seed''
tensor, by solving the Lichnerowicz equation for a conformal factor. To carry
out this construction in families, the seed has to be non-zero and chosen
continuously in the metric. The TT-tensors of a metric are obtained from York's
decomposition, which degenerates as the metric approaches one admitting a
conformal Killing vector (CKV). We deal with this in two steps that may be of
independent interest. We first show, using ideas from tractor calculus,
\cf\cite{CurryGover}, that metrics admitting CKVs are confined to subspaces of
arbitrarily high codimension, and that discarding them does not change the
(weak) homotopy type of the relevant spaces. The strategy goes back to work of
Beig, Chru\'sciel and Schoen~\cite{BeigChruscielSchoen}, to which we give a
more modern point of view based on the prolongation equation for conformal
Killing vectors. Second, we carry out the York decomposition in families over
the space $\NCK(M)$ of metrics without CKVs. The TT-tensors then form an
infinite-dimensional Hilbert bundle over $\NCK(M)$, which is trivial by
Kuiper's theorem and hence admits a nowhere-vanishing continuous section. In
order to establish the Hilbert bundle structure, we need a positive lower bound
on the coercivity constant of the conformal Killing operator that is locally
uniform in the metric.

The following is the main result of this article. Because its general strategy
is quite simple, we also provide its proof straight away, making use of the
more technical results that we are going to establish later in the article.
Although it may not seem so at first sight, there are plenty of manifolds $M$
satisfying all the assumptions of the theorem, \cf\cref{exMfds}.

\begin{theorem}\label{thmMain}
	Let $M$ be a closed connected spin manifold of dimension $n \geq 6$ admitting a
	metric of positive scalar curvature. Assume furthermore that $M$ does not admit
	initial data sets with timelike or lightlike parallel spinors, \cf Theorems
	\ref{thmParallelSpinorTop} and \ref{thmParallelSpinorObstructions}. Then there
	exists a CMC vacuum initial data set $(g,k) \in \VacC(M)$ such that
	\begin{gather*}
		\pi_\ell(\VacC(M), (g,k)) \neq 0
	\end{gather*}
	for all $\ell \geq 1$ with $n+\ell \equiv 0,1,2,4 \mod 8$. Moreover $(g,k)$ may
	be chosen independently of $\ell$, and in each of these degrees there is a
	non-zero class in $\pi_\ell(\VacC(M),(g,k))$ whose image in
	$\pi_\ell(\DEC(M),(g,k))$ is again non-zero, being detected by the index
	difference.
\end{theorem}
\begin{proof}
	The proof builds on the main result of~\cite{Gloeckle2024a} that we recall in
	\cref{thmCompIndDiff}: for any $\ell \geq 1$ and $g_0 \in \PSC(M)$ there is a
	commutative diagram
	\begin{equation*}
		\begin{tikzcd}
			\pi_{\ell-1}(\PSC(M),g_0) \rar{\Susp} \arrow[dr,"{\adiff}"'] & \pi_\ell(\Susp(\PSC(M)),[(g_0,0)]) \arrow[r,"\Phi_*"] & \pi_\ell(\sDEC(M),(g_0,0)) \arrow[dl,"{\oladiff}"]\\
			& \KO^{-n-\ell}(\pt). &
		\end{tikzcd}
	\end{equation*}
	The index difference $\adiff \colon \pi_{\ell-1}(\PSC(M), g_0) \to
		\KO^{-n-\ell}(\pt)$, whose target is recalled in~\eqref{eqKOGroups}, has been
	studied in numerous
	works~\cite{HankeSchickSteimle,CrowleySchickSteimle,BotvinnikEbertRandalWilliams}
	on the homotopy groups of $\PSC(M)$. For instance, it is known to be a
	non-trivial map for $\ell \geq 1$ and $n \geq 6$ whenever the target is
	non-zero, \ie if $n+\ell \equiv 0,1,2,4 \mod 8$, \cf\cref{thmNonTrivPSC}. So in
	all these cases, there exist non-trivial elements in
	$\pi_\ell(\sDEC(M),(g_0,0))$. They can be constructed from elements outside the
	kernel of $\adiff$ via $\Phi_* \circ \Susp$ and detected by their non-trivial
	image under $\oladiff$.

	We now make use of two observations about the maps in this diagram. Firstly,
	for the definition of $\oladiff$ it is actually not the strict DEC that matters
	but only the invertibility of the associated Dirac-Witten operator
	(\cf\cref{defOladiff}). This implies that $\oladiff$ factors over the subspace
	$\InvI(M)$ of initial data sets with invertible Dirac-Witten operator. Under
	the assumption that $M$ does not admit initial data sets with timelike or
	lightlike parallel spinors, \cref{thmParallelSpinorTop} implies that $\DEC(M)
		\subset \InvI(M)$ so that $\oladiff$ factors as
	\begin{align*}
		\oladiff \colon \pi_\ell(\sDEC(M),(g_0,0)) \to \pi_\ell(\DEC(M),(g_0,0)) \to \KO^{-n-\ell}(\pt),
	\end{align*}
	where the first map is induced by the inclusion, \cf also
	\cref{corExtensionToDEC}.

	Secondly, by \cref{obsValuesInCMC}, the map $\Phi$ takes values in CMC initial
	data sets. In other words, $\Phi_*$ factors as
	\begin{align*}
		\Phi_* \colon \pi_\ell(\Susp(\PSC(M)),[(g_0,0)]) \longrightarrow \pi_\ell(\sDEC_{\CMC}(M),(g_0,0)) \longrightarrow \pi_\ell(\sDEC(M),(g_0,0)),
	\end{align*}
	where the second map is induced by the inclusion. Taking these observations
	together, we see that for all $\ell$ under consideration there exist non-zero
	elements in $\pi_\ell(\sDEC_{\CMC}(M),(g_0,0))$ that stay non-zero after
	mapping to $\pi_\ell(\DEC(M),(g_0,0))$ with the map induced by the inclusion,
	since $\oladiff$ factors through this group. Let us fix such a $\gamma \in
		\pi_\ell(\sDEC_{\CMC}(M),(g_0,0))$.

	We now wish to use the conformal method to deform the inclusion
	$\sDEC_{\CMC}(M) \hookrightarrow \DEC_{\CMC}(M)$ into a map with values in
	$\VacC_{\CMC}(M)$. This does not work directly -- the potential presence of
	CKVs poses problems -- but by \cref{thmHomotopy}, we can achieve the following:
	there is a map $\Psi \colon \sDEC_{0, \CMC}(M) \to \VacC_{0,\CMC}(M)$ such that
	the left triangle in
	\begin{equation} \label{eqDefVacC}
		\begin{tikzcd}
			& & \VacC_{0, \CMC}(M) \arrow[d,hook] \arrow[r,"\iota",hook] & \VacC(M) \arrow[d,hook]\\
			\sDEC_{0, \CMC}(M) \arrow[urr,"{\Psi}"] \arrow[rr,hook] & & \DEC_{0, \CMC}(M) \arrow[r,hook] & \DEC(M)
		\end{tikzcd}
	\end{equation}
	commutes up to homotopy. Here the subscript ``$0$'' indicates that the first
	component of the initial data sets is required to be a metric without CKVs.

	The final ingredient is the statement that the inclusion $\sDEC_{0, \CMC}(M)
		\hookrightarrow \sDEC_{\CMC}(M)$ is a weak homotopy equivalence
	(\cref{corReductionToNCK}). This firstly means that for $g_0 \in \PSC(M)$ there
	is some $(g_1,k_1) \in \sDEC_{0, \CMC}(M)$ that is in the same path-connected
	component of $\sDEC_{\CMC}(M)$ as $(g_0,0)$, and we fix once and for all a path
	between those. Neither $(g_1,k_1)$ nor this path depends on $\ell$, which will
	entail that the initial data set $(g,k)$ produced below is the same in all
	degrees. Secondly, for all $\ell \geq 1$ the inclusion $\sDEC_{0, \CMC}(M)
		\hookrightarrow \sDEC_{\CMC}(M)$ induces an isomorphism
	\begin{gather*}
		\pi_\ell(\sDEC_{0, \CMC}(M),(g_1,k_1)) \overset{\cong}{\longrightarrow} \pi_\ell(\sDEC_{\CMC}(M),(g_1,k_1)) \cong \pi_\ell(\sDEC_{\CMC}(M),(g_0,0)),
	\end{gather*}
	the second isomorphism being constructed from the fixed path connecting the
	base points. By virtue of this isomorphism, the non-zero element $\gamma \in
		\pi_\ell(\sDEC_{\CMC}(M),(g_0,0))$ from above gives rise to a non-zero element
	of $\pi_\ell(\sDEC_{0, \CMC}(M),(g_1,k_1))$. Note that $\Psi$ moves the base
	point, so we set
	\begin{gather*}
		(g,k) \coloneqq \Psi(g_1,k_1) \in \VacC_{0, \CMC}(M),
	\end{gather*}
	which is the CMC vacuum initial data set of the statement, and we push our
	element forward by $(\iota \circ \Psi)_*$ to obtain $\widetilde{\gamma} \in
		\pi_\ell(\VacC(M),(g,k))$.

	It remains to see that the outcome $\widetilde{\gamma}$ is non-zero. The left
	triangle in~\eqref{eqDefVacC} commutes up to a homotopy $H$ from the inclusion
	$\sDEC_{0, \CMC}(M) \hookrightarrow \DEC_{0, \CMC}(M)$ to the composition of
	$\Psi$ with $\VacC_{0, \CMC}(M) \hookrightarrow \DEC_{0, \CMC}(M)$. In
	particular $t \mapsto H((g_1,k_1),t)$ is a path in $\DEC(M)$ from $(g_1,k_1)$
	to $(g,k)$, and concatenating it with the path fixed above yields an
	isomorphism
	\begin{gather*}
		\pi_\ell(\DEC(M),(g,k)) \cong \pi_\ell(\DEC(M),(g_0,0)).
	\end{gather*}
	Under this isomorphism, the image of $\widetilde{\gamma}$ in
	$\pi_\ell(\DEC(M),(g,k))$ corresponds to the image in
	$\pi_\ell(\DEC(M),(g_0,0))$ of the original element $\gamma \in
		\pi_\ell(\sDEC_{\CMC}(M),(g_0,0))$. The latter is non-zero with
	$\oladiff(\gamma) \neq 0$, so the same is true for $\widetilde{\gamma} \in
		\pi_\ell(\VacC(M),(g,k))$ as well.
\end{proof}

Let us point out that, under the assumptions of the theorem, we do not know
whether the space of CMC vacuum initial data sets, or $\VacC(M)$ itself, is
path-connected, see however~\cite{Gloeckle2021,Gloeckle2024a}. The conclusion
of the theorem is therefore understood relative to the base point $(g,k)$
produced by the proof.

We view this article as a first step in the study of the topology of spaces of
initial data sets, and several natural questions arise after this work. The
restriction $n \geq 6$ in \cref{thmMain} comes from the results on $\PSC(M)$
that it relies on, and in low dimensions the topology of $\PSC(M)$ is quite
different. In dimension $3$, Marques proved that the moduli space
$\PSC(M)/\mathrm{Diff}(M)$ is path-connected~\cite{Marques}, and Bamler and
Kleiner proved that $\PSC(M)$ is contractible whenever it is
non-empty~\cite{BamlerKleiner}. It is therefore natural to ask whether this
contractibility carries over to $\DEC(M)$ and $\VacC(M)$ when $M$ is a closed
$3$-manifold admitting a metric of positive scalar curvature. This is the
physically relevant dimension. Note that the results of
\cref{secKilling,secTTtensors} hold for all $n \geq 3$, so that the
parametrized conformal method of \cref{secHomotopy} is also available in
dimension $3$.

Much less is known about the topology of $\PSC(M)$ in dimensions $4$ and $5$,
where non-trivial topology of $\PSC(M)$ has been detected by Seiberg-Witten
theory~\cite{Ruberman} and by eta invariants~\cite{BotvinnikGilkey}, and the
corresponding questions for initial data sets are open.

Finally, all the classes constructed in this article consist of CMC initial
data sets. This raises the question of how the CMC initial data sets sit inside
the space of all initial data sets. Is $\VacC_{\CMC}(M)$ a deformation retract
of $\VacC(M)$, or at least is the inclusion a weak homotopy equivalence, and
similarly for $\DEC_{\CMC}(M)$ in $\DEC(M)$? Or is it possible to construct
homotopy classes of non-CMC initial data sets that do not come from CMC ones,
that is, which are not in the image of the map induced by the inclusion? In the
first case, the topology of $\VacC(M)$ would be accessible through the CMC
conformal method, which describes the CMC part of the constraint set
completely. In the second case, it would contain topology that this method
cannot detect. We intend to study this question in future work.

The rest of the article establishes the results used above. First, we recall in
\cref{secPrelim} previous results on the homotopy groups of $\PSC(M)$ and
$\sDEC(M)$. This includes the extension to $\DEC(M)$
following~\cite{AmmannGloeckle} as well as the observation that the considered
non-zero homotopy classes may be represented by maps to $\sDEC_{\CMC}(M)$. In
\cref{secKilling} we reduce to metrics without CKVs. In \cref{secTTtensors} we
carry out the York decomposition in families and construct the continuous
family $\NCK(M) \ni g \mapsto \sigma(g)$ of non-zero TT-tensors, and in
\cref{secHomotopy} we use it to deform the inclusion $\sDEC_{0, \CMC}(M)
	\hookrightarrow \DEC_{0, \CMC}(M)$ into a map with values in $\VacC_{0,
		\CMC}(M)$. The appendix recalls the definition of the index difference
$\oladiff$.

\subsection*{Acknowledgments}
R. G. is partially supported by the French National Research Agency (ANR) under
grants ANR-23-CE40-0010-02 (Einstein constraints: past, present, and future,
EINSTEIN-PPF) and ANR-25-CE40-4883 (Scattering, Holography and General
Relativity). J. G. is thankful to Matthias Uschold for sharing his
group-theoretic knowledge and acknowledges support through the Walter Benjamin
grant No.~556505019 by the German Science Foundation DFG.

\subsection*{Use of AI assistance}
In preparing this article the authors used Claude (Anthropic) to perform part
of the symbolic computations of \cref{secKilling}, in particular in proving
that explicit metrics do not carry any non-zero conformal Killing vector, and
to help revise parts of the exposition. All definitions, statements and proofs
are the authors' own, and the authors have checked every part of the article,
including these calculations.

\section{Preliminaries}\label{secPrelim}
In the following, $M$ will always be a closed manifold of dimension $n$. We use
the geometric sign convention for the Laplacian, $\Delta = -\div \nabla = -\tr
	\nabla^2$, so that $\Delta f(x_0) \geq 0$ at any local maximum $x_0$ of $f$ and
$\Delta$ has non-negative spectrum. We define $\Met(M)$ to be the space of
Riemannian metrics on $M$, equipped with $C^\infty$-topology. We denote by
$\PSC(M) = \{g \in \Met(M) \mid \scal^g > 0\}$ its subspace of positive scalar
curvature (PSC) metrics. Likewise, $\Ini(M) = \{(g,k) \mid g \in \Met(M), k \in
	\Gamma(S_2M)\}$ is the space of initial data sets on $M$. We shall consider
various subspaces:
\begin{align*}
	\sDEC(M) & \coloneqq \{(g,k) \in \Ini(M) \mid \rho > |j|_g \}     \\
	\DEC(M)  & \coloneqq \{(g,k) \in \Ini(M) \mid \rho \geq |j|_g \}  \\
	\VacC(M) & \coloneqq \{(g,k) \in \Ini(M) \mid \rho = 0,\, j=0 \},
\end{align*}
where $\rho$ and $j$ are the energy and momentum densities defined in
\cref{eqDensities}.

Since our goal is to study the topology of $\VacC(M)$, we briefly review what
is known about $\sDEC(M)$ (and $\DEC(M)$). To date, the only information about
non-trivial higher homotopy groups of $\sDEC(M)$ is obtained by Dirac-operator
methods. More precisely, it is obtained from the index difference
\begin{gather*}
	\oladiff \colon \pi_\ell(\sDEC(M), (g_0,k_0)) \lto \KO^{-n-\ell}(\pt),
\end{gather*}
of~\cite{Gloeckle2024a}, defined for a base point $(g_0,k_0) \in \sDEC(M)$. We
recall its definition in \cref{secOladiff}.

Its definition requires $M$ to be a spin manifold with a chosen spin structure,
which we assume from now on. We recall that the $\KO$-groups of a point are
simply given by
\begin{align} \label{eqKOGroups}
	\KO^{-i}(\pt) \cong \begin{cases}
		                    \bZ   & i \equiv 0,4 \mod 8      \\
		                    \Ztwo & i \equiv 1,2 \mod 8      \\
		                    0     & i \equiv 3,5,6,7 \mod 8.
	                    \end{cases}
\end{align}

The motivation to consider $\oladiff$ originates from an analogous construction
in the case of positive scalar curvature metrics that dates back to
Hitchin~\cite{Hitchin}. In this case, the index difference is a map
\begin{gather*}
	\adiff \colon \pi_\ell(\PSC(M),g_0) \to \KO^{-n-\ell-1}(\pt),
\end{gather*}
defined for $g_0 \in \PSC(M)$, for which several non-triviality results have
been established. Probably the most powerful of these is the following:

\begin{theorem}[Botvinnik, Ebert and Randal-Williams~\cite{BotvinnikEbertRandalWilliams}]\label{thmNonTrivPSC}
	Let $M$ be a closed spin manifold of dimension $n \geq 6$ admitting a metric of positive scalar curvature.
	Then
	\begin{align*}
		\adiff \colon \pi_\ell(\PSC(M),g_0) \to \KO^{-n-\ell-1}(\pt)
	\end{align*}
	is non-trivial for all $\ell \geq 0$ for which the target is non-zero, \ie when
	$n+\ell+1 \equiv 0,1,2,4 \mod 8$, and all base points $g_0 \in \PSC(M)$.
\end{theorem}

\begin{remark}
	\begin{enumerate}
		\item In degrees $n+\ell+1 \equiv 1,2 \mod 8$, where $\KO^{-n-\ell-1}(\pt) \cong
			      \Ztwo$, Crowley, Schick and Steimle \cite{CrowleySchickSteimle} showed that
		      families $(g_x)_{x \in S^\ell}$ with $\adiff([g_x]) \neq 0$ can be constructed
		      by pulling back a single PSC metric with a suitable family of diffeomorphisms.
		\item On the other hand, Hanke, Schick and Steimle \cite{HankeSchickSteimle}
		      constructed ``geometrically significant'' families of positive scalar curvature
		      metrics on spheres $S^n$ for large~$n$. They have non-zero index difference in
		      degrees $n+\ell+1 \equiv 0,4 \mod 8$, where $\KO^{-n-\ell-1}(\pt) \cong \bZ$,
		      and the additional property that they remain non-trivial under the homomorphism
		      $\pi_\ell(\PSC(S^n),g_0) \to \pi_\ell(\PSC(S^n)/\mathrm{Diff}_p(S^n),[g_0])$,
		      where $\mathrm{Diff}_p(S^n)$ denotes the group of diffeomorphisms of $S^n$
		      fixing a point $p \in S^n$.
	\end{enumerate}
\end{remark}

In order to conclude from \cref{thmNonTrivPSC} non-triviality of $\oladiff$, we
need a suitable comparison map between the space of PSC metrics and the one of
DEC initial data sets. Setting $I \coloneqq [-1,1]$, we consider the following
map of pairs:
\begin{equation} \label{eqSuspDEC}
	\begin{aligned}
		(\Met(M) \times I, \PSC(M) \times I \cup \Met(M) \times \del I) & \lto (\Ini(M), \sDEC(M))                  \\
		(g,t)                                                           & \lmapsto \bigl(g,\; t\,\tau(g)\, g\bigr),
	\end{aligned}
\end{equation}
where $C > 0$ is a fixed constant and
\begin{gather*}
	\tau(g) = \sqrt{\frac{1}{n(n-1)} \max \bigl\{0, -\min_{x \in M} \scal^g(x)\bigr\}} + C.
\end{gather*}
The precise form of $\tau \colon \Met(M) \to \bR_{\geq 0}$ does not matter. All
we use is that it is a continuous function for which $\bigl(g,\tau(g)\,g\bigr)$
satisfies the strict dominant energy condition $\rho > |j|$. The subspace
$\PSC(M) \times I \cup \Met(M) \times \del I$ may be viewed as a model for the
suspension of $\PSC(M)$, which is defined by $\Susp(\PSC(M)) \coloneqq \PSC(M)
	\times I / \sim$ if $\PSC(M) \neq \emptyset$, where $(g,1) \sim (g^\prime, 1)$
and $(g,-1) \sim (g^\prime, -1)$ for all $g, g^\prime \in \PSC(M)$, and by
$\Susp(\emptyset) = \del I$. Because $\Met(M)$ is contractible, the canonical
projection map $\PSC(M) \times I \cup \Met(M) \times \del I \to \Susp(\PSC(M))$
is a homotopy equivalence, \ie there is a continuous map $\Susp(\PSC(M)) \to
	\PSC(M) \times I \cup \Met(M) \times \del I$ such that both compositions are
homotopic to the respective identities. An explicit homotopy inverse is given
in~\cite[before Cor.~6.2]{AmmannGloeckle} and we use it together with
\eqref{eqSuspDEC} to obtain a continuous map
\begin{align*}
	\Phi \colon \Susp(\PSC(M)) \overset{\simeq}{\lto} \PSC(M) \times I \cup \Met(M) \times \del I & \lto \sDEC(M).
\end{align*}
We recall furthermore that there is a suspension homomorphism
\begin{align*}
	\Susp \colon \pi_\ell(\PSC(M),g_0) & \lto \pi_{\ell+1}(\Susp(\PSC(M)),[(g_0,0)]),
\end{align*}
which takes a representative $S^\ell \to \PSC(M)$ and maps it to its suspension
$S^{\ell+1} \cong \Susp(S^\ell) \to \Susp(\PSC(M))$. The main theorem
of~\cite{Gloeckle2024a} shows that the suspension homomorphism and the map
$\Phi$ produce, out of families in $\PSC(M)$ with non-trivial index difference,
families in $\sDEC(M)$ with non-trivial $\oladiff$.
\begin{theorem}[{\cite{Gloeckle2024a}}]\label{thmCompIndDiff}
	For any base point $g_0 \in \PSC(M)$ (and compatible base points in the other
	spaces) and all $\ell \geq 0$, the following diagram commutes:
	\begin{equation*}
		\begin{tikzcd}
			\pi_\ell(\PSC(M),g_0) \rar{\Susp} \arrow[dr,"{\adiff}"'] & \pi_{\ell+1}(\Susp(\PSC(M)),[(g_0,0)]) \arrow[r,"\Phi_*"] & \pi_{\ell+1}(\sDEC(M),(g_0,0)) \arrow[dl,"{\oladiff}"]\\
			& \KO^{-n-\ell-1}(\pt). &
		\end{tikzcd}
	\end{equation*}
\end{theorem}
Combining Theorems \ref{thmNonTrivPSC} and \ref{thmCompIndDiff}, we can
conclude the following. Note that \cref{thmCompIndDiff} also gives us an
explicit way to construct non-trivial elements using suspension and the map
$\Phi$ out of elements in $\pi_\ell(\PSC(M),g_0)$ with non-trivial index
difference.

\begin{corollary}
	\label{corNonTrivSDEC}
	Let $M$ be a closed spin manifold of dimension $n \geq 6$ admitting a metric of positive scalar curvature.
	Then
	\begin{align*}
		\oladiff \colon \pi_\ell(\sDEC(M),(g_0,0)) \to \KO^{-n-\ell}(\pt)
	\end{align*}
	is non-trivial for all $\ell \geq 1$ for which the target is non-zero, \ie when
	$n+\ell \equiv 0,1,2,4 \mod 8$, and all base points $g_0 \in \PSC(M)$.
\end{corollary}

All the initial data sets in the image of~\eqref{eqSuspDEC} have constant
$\tr^g(k)$, that is, they are CMC (constant mean curvature) initial data sets.
This is what will let us run the conformal method later. From now on, the
subscript \emph{CMC} indicates the restriction to initial data sets satisfying
the CMC condition.

\begin{observation}\label{obsValuesInCMC}
	The map $\Phi$ in \cref{thmCompIndDiff} takes values in $\sDEC_{\CMC}(M)
		\subset \sDEC(M)$.
\end{observation}

Next, we explain how to extend the results to the (non-strict) dominant energy
condition. The ideas for this were developed in~\cite{AmmannGloeckle}. We
recall from \cref{defOladiff} that $\oladiff$ factors through the space
$\InvI(M)$ of initial data sets with invertible Dirac-Witten operator defined
in \cref{defInvI}:
\begin{align*}
	\oladiff \colon \pi_\ell(\sDEC(M), (g_0,k_0)) & \lto \pi_\ell(\InvI(M), (g_0,k_0)) \lto \KO^{-n-\ell}(\pt),
\end{align*}
where the first map is induced by the inclusion $\sDEC(M) \subset \InvI(M)$.
This implies that elements in $\pi_\ell(\sDEC(M))$ whose non-triviality is
detected by $\oladiff$ actually remain non-trivial in the bigger space
$\InvI(M)$.

The central idea of~\cite{AmmannGloeckle} is now that often all DEC initial
data sets have an invertible Dirac-Witten operator, for otherwise the initial
data set has to carry some kind of parallel spinor, which is a very special
property. For many manifolds, the existence of such spinors may be excluded for
all DEC initial data sets for topological reasons. Consequently, $\DEC(M)
	\subset \InvI(M)$ and $\oladiff$ may serve to also detect non-triviality of
$\pi_\ell(\DEC(M),(g_0,k_0))$.

\begin{theorem}[{\cite[Thm.~5.1]{AmmannGloeckle}}]\label{thmParallelSpinorTop}
	Given an $n$-dimensional closed connected spin manifold $M$, an initial data
	set $(g,k) \in \DEC(M) \setminus \InvI(M)$ admits a non-trivial parallel
	spinor. This spinor is of one of two types that are called \emph{timelike} and
	\emph{lightlike}.
\end{theorem}

Each of the two types imposes strong restrictions on the topology of $M$, which
is what makes the criterion usable.

\begin{theorem}[{\cite[Cor.~5.5 and Thm.~5.19]{AmmannGloeckle}}]\label{thmParallelSpinorObstructions}
	Let $M$ be an $n$-dimensional closed connected spin manifold.

	If $M$ admits an initial data set with timelike parallel spinor, then it also
	admits a Ricci-flat Riemannian metric with parallel spinor. In particular,
	\begin{enumerate}[nosep,itemsep=2pt,topsep=-5pt]
		\item $\pi_1(M)$ is virtually abelian and has polynomial growth of degree at most $n$,
		\item unless $M$ is finitely covered by a torus, $b_i(\widetilde{M}) \neq 0$ for at
		      least one $i \in \{2,3,4\}$, where $\widetilde{M}$ denotes the universal
		      covering of $M$, and
		\item $|\hat{A}(M)| \leq 2^{n/2 -1}$.
	\end{enumerate}

	If $M$ admits an initial data set with lightlike parallel spinor, then
	\begin{enumerate}[nosep,itemsep=2pt,topsep=-5pt]
		\item $b_1(M) \neq 0$, and
		\item a finite index subgroup $\Gamma < \pi_1(M)$ of its fundamental group fits into
		      a short exact sequence $1 \to \bZ^{s-r} \to \Gamma \to \bZ^{r+1} \to 1$ for
		      some $0 \leq r \leq s \leq n$.
	\end{enumerate}
\end{theorem}

Thus any manifold $M$ for which \cref{thmParallelSpinorObstructions} rules out
both types has $\DEC(M) \subset \InvI(M)$. In~\cite[Ex.~5.6~(3)--(5),
	Ex.~5.25]{AmmannGloeckle} numerous examples of such manifolds are listed. In
view of \cref{thmMain}, we are interested in examples that are furthermore
closed connected manifolds of dimension $n \geq 6$ and admit metrics of
positive scalar curvature.

\begin{example} \label{exMfds}
	\begin{itemize}
		\item The sphere $S^n$, $n \geq 6$, is simply connected and can therefore not carry
		      lightlike parallel spinors ($b_1(M) = 0$). Since its Betti numbers in degrees
		      $2,3,4$ vanish as well, it cannot carry timelike parallel spinors. The same
		      argument applies to every simply connected closed spin manifold with $b_2(M) =
			      b_3(M) = b_4(M) = 0$, and such a manifold carries a positive scalar curvature
		      metric if and only if $\alpha(M) = 0$ by Stolz's theorem~\cite{Stolz}. Examples
		      of this kind are homotopy spheres $\Sigma^n$ with $\alpha(\Sigma^n) = 0$,
		      products $S^p \times S^q$ with $p, q \geq 5$, and connected sums of such
		      manifolds. If a finite quotient of such a manifold is spin and carries a
		      positive scalar curvature metric, which is the case for some of the lens
		      spaces, it is also an example.
		\item 	Examples with large fundamental group exist as well. If $\pi_1(M)$ is not
		      virtually solvable, then $M$ does not admit initial data sets with timelike or
		      lightlike parallel spinors. Such manifolds are easily constructed by connected
		      sums. The connected sum of closed spin manifolds carrying positive scalar
		      curvature metrics is spin and, if $n \geq 3$, admits a positive scalar
		      curvature metric by the surgery theorem of Gromov and Lawson and of Schoen and
		      Yau \cite{GromovLawson1980b,SchoenYau}. In these dimensions, the fundamental
		      group of a connected sum is obtained as the free product of the fundamental
		      groups of the summands. In particular, its first $\ell^2$-Betti number is
		      non-zero if at least two summands have a non-trivial fundamental group
		      \cite[Thm.~4.5 (ii)]{Kammeyer} (unless both of them are $\Ztwo$, in which case
		      a third non-trivial group is needed). On the other hand, by a theorem of
		      Cheeger and Gromov \cite[Thm.~0.2]{CheegerGromov}, the $\ell^2$-Betti numbers
		      of an amenable group vanish in all positive degrees. Hence, the fundamental
		      group of such a connected sum is not amenable and in particular not virtually
		      solvable. The simplest examples obtained in this way are the connected sums
		      $\#_N(S^{n-1}\times S^1)$ with $N \geq 2$ and $n \geq 6$.
	\end{itemize}
\end{example}

\begin{corollary} \label{corExtensionToDEC}
	If $M$ is such that $\DEC(M) \subset \InvI(M)$, then $\oladiff$ factors as
	\begin{align*}
		\pi_\ell(\sDEC(M),(g_0,k_0)) \to \pi_\ell(\DEC(M),(g_0,k_0)) \to \KO^{-n-\ell}(\pt),
	\end{align*}
	where the first map is induced by the inclusion, and \cref{thmCompIndDiff}
	holds with $\sDEC(M)$ replaced by $\DEC(M)$.
\end{corollary}

\section{Metrics without conformal Killing vectors}\label{secKilling}
Let $M$ be an $n$-dimensional closed connected manifold. A vector field $X$ is
called a conformal Killing vector (CKV) for the metric $g$ if $\lie_X g =
	\frac{1}{n} \tr^g(\lie_X g) g$, \ie $X$ is a conformal Killing vector if and
only if the traceless part of $\lie_X g$ vanishes:
\begin{equation}\label{eqConfKillingOp}
	\bL X \coloneqq \lie_X g - \frac{2}{n} \divg^g(X) g \equiv 0.
\end{equation}

Just as Killing vectors are the infinitesimal generators of isometries,
CKVs are the infinitesimal generators of conformal diffeomorphisms. We denote
by $\NCK(M) \subset \Met(M)$ the space of metrics without non-zero CKVs and by
$\PNCK(M) \coloneqq \NCK(M) \cap \PSC(M)$ its subspace of metrics that are
additionally of positive scalar curvature. Moreover, we let $\Ini_{0,\CMC}(M)
	\coloneqq \{(g,k) \in \Ini_{\CMC}(M) \mid g \in \NCK(M)\}$ and
$\sDEC_{0,\CMC}(M) \coloneqq \Ini_{0,\CMC}(M) \cap \sDEC(M)$, and similarly
$\DEC_{0,\CMC}(M)$ and $\VacC_{0,\CMC}(M)$. The main technical result of this
section is the following theorem, from which we will deduce that several
natural inclusions are weak homotopy equivalences.

\begin{theorem}\label{thmReductionToNCK}
	Assume $n \geq 3$. Let $\ell \geq 0$ and $\Omega \subset S^\ell$ be a closed subset.
	Moreover, let $A \subset \Omega$ be a closed subset such that $A \supset \del \Omega = \Omega
		\setminus \ring{\Omega}$.
	Let $g \colon (\Omega,A) \to (\Met(M), \NCK(M))$ be a map of pairs and fix a
	$C^\infty$-neighborhood $\mathcal{U}$ of $g$. Then there exists a homotopy $H
		\colon \Omega \times [0,1] \to \Met(M)$ such that
	\begin{enumerate}
		\item $H(-,0) = g$,
		\item\label{itmRedNCKFixA} $H(a,-)$ is constant for all $a \in A$,
		\item $H(-,1)$ takes values in $\NCK(M)$, and
		\item\label{itmRedNCKCinfty} $H(-,t)$ is $C^\infty$-close to $g$ for all $t \in [0,1]$ in the sense that it lies in~$\mathcal{U}$.
	\end{enumerate}
\end{theorem}

Here, what is meant by a $C^\infty$-neighborhood of $g$ is a neighborhood of
$g$ in the compact-open topology with respect to the induced topology on $\Omega$
and the $C^\infty$-topology on $\Met(M)$.

We remind the reader that a map $f \colon Y \to Z$ is called a weak homotopy
equivalence if the pushforward map
\[
	f_* \colon \pi_0(Y) \to \pi_0(Z)
\]
is a bijection and if for any base point $y_0 \in Y$, the pushforward map
\[
	f_* \colon \pi_{\ell}(Y, y_0) \to \pi_{\ell}(Z, f(y_0))
\]
is a group isomorphism for every $\ell > 0$. From \cref{thmReductionToNCK} we
deduce the following corollary, which is what will be used in the rest of the
paper.

\begin{corollary}\label{corReductionToNCK}
	Assume that $n \geq 3$. All the following inclusions are weak homotopy
	equivalences:
	\begin{align*}
		\NCK(M)           & \hookrightarrow \Met(M)          \\
		\PNCK(M)          & \hookrightarrow \PSC(M)          \\
		\Ini_{0,\CMC}(M)  & \hookrightarrow \Ini_{\CMC}(M)   \\
		\sDEC_{0,\CMC}(M) & \hookrightarrow \sDEC_{\CMC}(M).
	\end{align*}
\end{corollary}

\begin{proof}[Proof of \cref{corReductionToNCK} using \cref{thmReductionToNCK}]
	All the assertions are proved according to the same scheme. In each case, the
	map to be deformed is defined on a closed subset $\Omega$ of a sphere, the closed
	subset $A \supset \del \Omega$ collects the part of the domain where the map
	already takes values in the smaller space and must not move, and the homotopy
	$H$ provided by \cref{thmReductionToNCK} deforms the map into the smaller
	space relative to $A$. All that has to be done in each case is to identify $\Omega$
	and $A$.

	We spell this out for the inclusion $\NCK(M) \hookrightarrow \Met(M)$. For
	surjectivity on $\pi_0$, let $g \in \Met(M)$ be arbitrary. Applying the theorem
	with $\Omega = \pt \subset S^0$ and $A = \emptyset$ to the constant map $g$, the
	homotopy $H$ is a path from $g$ to a metric in $\NCK(M)$, so every path
	component of $\Met(M)$ contains a point of $\NCK(M)$. For $\pi_0$-injectivity,
	let $g \colon [0,1] \to \Met(M)$ be a path whose endpoints lie in $\NCK(M)$. We
	regard $[0,1]$ as a closed arc $\Omega$ in the circle $S^1$. The boundary of $\Omega$ in
	$S^1$ is $\{0, 1\}$, so the pair $(\Omega, A) = ([0,1], \{0,1\})$ satisfies $A
		\supset \del \Omega$, and the theorem yields a homotopy, relative to the endpoints,
	from $g$ to a path $H(-,1)$ with values in $\NCK(M)$. Hence the endpoints of
	$g$ lie in the same path component of $\NCK(M)$.

	For surjectivity on $\pi_\ell$, $\ell \geq 1$, we apply the theorem with $\Omega =
		S^\ell$ and $A = \pt$ the base point of $S^\ell$ (here $\del \Omega = \emptyset$, so
	any closed $A$ is allowed) to a representative of an element of
	$\pi_\ell(\Met(M),g_0)$ with $g_0 \in \NCK(M)$. The resulting $H(-,1)$
	represents the same class, and it is again a based map because the base point
	is fixed by \cref{itmRedNCKFixA}. Lastly, for $\pi_\ell$-injectivity, let $f_0,
		f_1 \colon S^\ell \to \NCK(M)$ be maps based at $g_0$ that are homotopic as
	maps to $\Met(M)$, and let $g \colon S^\ell \times [0,1] \to \Met(M)$ be a
	based homotopy between them. To fit this into the setting of the theorem, we
	regard the cylinder $S^\ell \times [0,1]$ as the complement of two open polar
	caps in $S^{\ell+1}$. In this way $\Omega \coloneqq S^\ell \times [0,1]$ becomes a
	closed subset of $S^{\ell+1}$ with $\del \Omega = S^\ell \times \{0,1\}$. We set $A
		\coloneqq S^\ell \times \{0,1\} \cup \pt \times [0,1]$, which contains $\del
		\Omega$, and on which $g$ takes values in $\NCK(M)$: on $S^\ell \times \{0,1\}$
	because $f_0$ and $f_1$ map to $\NCK(M)$, and on $\pt \times [0,1]$ because the
	homotopy is based. The theorem yields $H$ such that $H(-,1) \colon S^\ell
		\times [0,1] \to \NCK(M)$ agrees with $g$ on $A$. In particular $H(-,1)$ is a
	based homotopy from $f_0$ to $f_1$ within $\NCK(M)$, so $f_0$ and $f_1$
	represent the same class in $\pi_\ell(\NCK(M), g_0)$.

	The inclusion $\PNCK(M) \hookrightarrow \PSC(M)$ is a weak homotopy equivalence
	by exactly the same proof except that in every application of
	\cref{thmReductionToNCK}, we have to make use of \cref{itmRedNCKCinfty} to
	ensure that the homotopies nowhere violate the PSC condition (which is an open
	condition in the $C^\infty$-topology).

	For the statements about initial data sets, we note that
	\cref{thmReductionToNCK} continues to hold true if at all occasions $\Met(M)$
	is replaced by $\Ini_{\CMC}(M)$ and $\NCK(M)$ by $\Ini_{0,\CMC}(M)$. In fact,
	this follows from \cref{thmReductionToNCK}: we consider the decomposition map
	\begin{align*}
		\mathrm{Dec} \colon \Ini_{\CMC}(M)                          & \longrightarrow \Met(M) \times \bR \times \Gamma(S_2 M)                                    \\
		(g,k)                                                       & \longmapsto \left(g, \tr^g(k), k-{\textstyle\frac{1}{n}} \tr^g(k)g \right)
		\intertext{and its left inverse}
		\mathrm{Rec} \colon \Met(M) \times \bR \times \Gamma(S_2 M) & \longrightarrow \Ini_{\CMC}(M)                                                             \\
		(g,\tau,\tilde{k})                                          & \longmapsto \left(g, \tilde{k} + {\textstyle\frac{1}{n}} (\tau -\tr^g(\tilde{k}))g \right)
	\end{align*}

	and notice that $\mathrm{Dec}$ maps $\Ini_{0,\CMC}(M)$ into $\NCK(M) \times \bR
		\times \Gamma(S_2 M)$, while $\mathrm{Rec}$ maps $\NCK(M) \times \bR \times
		\Gamma(S_2 M)$ into $\Ini_{0,\CMC}(M)$. Now given a map $(\Omega,A) \to
		(\Ini_{\CMC}(M), \Ini_{0,\CMC}(M))$, we compose with $\mathrm{Dec}$ and denote
	the component functions by $(g,\tau,\tilde{k})$, so that $\mathrm{Rec} \circ
		(g,\tau,\tilde{k})$ is the original map. Then we apply the theorem to its first
	component $g \colon (\Omega,A) \to (\Met(M), \NCK(M))$ and consider $\mathrm{Rec}
		\circ (H,\tau,\tilde{k})$, where $H$ is the homotopy from
	\cref{thmReductionToNCK} and $\tau$ and $\tilde{k}$ are extended constantly in
	$t \in [0,1]$. Note that by continuity of $\mathrm{Rec}$ for every
	$C^\infty$-neighborhood $\mathcal{U}$ of $\mathrm{Rec} \circ
		(g,\tau,\tilde{k})$ there is a $C^\infty$-neighborhood $\mathcal{V}$ of $g$
	such that $\mathrm{Rec} \circ (\tilde{g}, \tau, \tilde{k}) \in \mathcal{U}$ for
	every $\tilde{g} \in \mathcal{V}$. It is now straightforward to check that
	$\mathrm{Rec} \circ (H,\tau,\tilde{k})$ indeed satisfies all the properties.
	With this version of \cref{thmReductionToNCK} for initial data sets at hand,
	the proofs that $\Ini_{0,\CMC}(M) \hookrightarrow \Ini_{\CMC}(M)$ and
	$\sDEC_{0,\CMC}(M) \hookrightarrow \sDEC_{\CMC}(M)$ are weak homotopy
	equivalences are literally the same as the proofs for the other two inclusions,
	respectively.
\end{proof}

The remainder of this section is devoted to the proof of
\cref{thmReductionToNCK}. The strategy goes back to Beig, Chru\'sciel and
Schoen~\cite{BeigChruscielSchoen}, who proved that metrics with local conformal
Killing vectors are non-generic, with an explicit pointwise obstruction built
from the Cotton tensor in dimension three, and with a non-explicit one in
higher dimensions. Their results would in fact be sufficient for our purposes,
\cf \cite[Thm.~2.1 and Thm.~7.4]{BeigChruscielSchoen}, at the cost of replacing
the explicit value of $K$ in \cref{propLocalDefNCK} below by a finite but
non-explicit $K = K(n)$. What we propose instead is a version of their argument
that is hopefully easier to follow. The obstruction is written down explicitly,
and it involves no more than the $4$-jet of the metric. It does, however, rest
in part on the Weyl tensor, so we construct it for $n \geq 4$ and treat
dimension three separately in \cref{propLocalDefNCK}, where their obstruction
takes over.

Following~\cite[App.~A]{CurryGover}, we encode the $2$-jet of a conformal
Killing vector as a parallel section of the adjoint tractor bundle, so that a
conformal Killing vector is determined by its $2$-jet at a single point. Such a
section lies, at every point, in the kernel of the curvature of the
corresponding connection. This curvature is a pointwise object built from
the Weyl and Cotton tensors and their first covariant derivatives, and for a
generic metric its kernel is trivial.

We first recall how the $2$-jet of a conformal Killing vector at any given
point $p_0 \in M$ is constrained. If $X$ is an arbitrary vector field,
splitting $\nabla^g X^\flat$ into its antisymmetric part, its trace-free
symmetric part and its trace part gives
\[
	\nabla^{g}_a X_ b = \frac{1}{2}(\upd X^\flat)_{ab} + \frac{1}{2} (\bL X)_{ab} + \frac{1}{n} \divg^g(X) g_{ab}.
\]
In particular, if $X$ is a conformal Killing vector, the component $\bL X$
vanishes and we are left with
\begin{equation}\label{eqRestricted1stJet}
	\nabla^{g}_a X_ b = \frac{1}{2}(\upd X^\flat)_{ab} + \frac{1}{n} \divg^g(X) g_{ab}.
\end{equation}

Taking the covariant derivative of the previous equation and summing over
circular permutations, we get
\begin{align*}
	 & \nabla^g_a \nabla^g_b X_c + \nabla^g_b \nabla^g_c X_a + \nabla^g_c \nabla^g_a X_b                                                                                                      \\
	 & \qquad = \frac{1}{2}(\upd (\upd X^\flat))_{abc} + \frac{1}{n} \nabla^g_a (\divg^g(X) g_{bc}) + \frac{1}{n} \nabla^g_b (\divg^g(X) g_{ca}) + \frac{1}{n} \nabla^g_c (\divg^g(X) g_{ab}) \\
	 & \qquad = \kappa_a g_{bc} + \kappa_b g_{ac} + \kappa_c g_{ab},
\end{align*}
where we used the fact that $\upd \circ \upd = 0$ and where $\kappa \in
	\Gamma(T^* M)$ is defined as
\[
	\kappa \coloneqq \frac1n\, \upd(\divg^g X).
\]
The last two terms on the left hand side combine to
\begin{align*}
	\nabla^g_b \nabla^g_c X_a + \nabla^g_c \nabla^g_a X_b
	 & = \nabla^g_c \nabla^g_b X_a - \riemuddd{d}{a}{b}{c} X_d + \nabla^g_c \nabla^g_a X_b \\
	 & = \nabla^g_c \big(\nabla^g_b X_a + \nabla^g_a X_b\big) - \riemuddd{d}{a}{b}{c} X_d  \\
	 & = \frac{2}{n} \nabla^g_c (\divg^g(X) g_{ab}) - \riemuddd{d}{a}{b}{c} X_d            \\
	 & = 2 \kappa_c g_{ab} - \riemuddd{d}{a}{b}{c} X_d,
\end{align*}
so we are left with
\begin{equation}\label{eqRestricted2ndJet}
	\nabla^g_a \nabla^g_b X_ c = \riemuddd{d}{a}{b}{c} X_d + \kappa_a g_{bc} + \kappa_b g_{ac} - \kappa_c g_{ab}.
\end{equation}
In the rest of this section, $X$ denotes a conformal Killing vector defined on
a connected open subset $U$ of $M$. By~\eqref{eqRestricted1stJet}, the
covariant $2$-tensor $\nabla^g X^\flat$ takes values in the bundle
\[
	\co(TM) \coloneqq \Lambda^2 T^*M \oplus \bR \cdot g.
\]
Raising the second index with $g$, we identify a $2$-tensor $\beta_{ab}$ with
the endomorphism $\beta_a{}^b$ of $TM$. In this way $\co(TM)$ becomes the Lie
algebra of infinitesimal rotations and scalings, and $\nabla^g X^\flat$ becomes
the endomorphism $Y \mapsto \nabla^g_Y X$. Setting
\[
	\cA \coloneqq TM \oplus \co(TM) \oplus T^*M,
\]
the conformal Killing vector $X$ therefore defines a section
\[
	\mathscr{X} \coloneqq \bigl(X,\, \nabla^g X^\flat,\, \kappa\bigr)
\]
of $\cA$ over $U$, with $\kappa$ as above, which records the $2$-jet of $X$. We
denote a general section of $\cA$ by $(\xi, \beta, \kappa)$, and we use the
same letters for the value of $\mathscr{X}$ at a point $p_0 \in U$, so that, in
particular, $\xi = X(p_0)$.

Besides the Weyl tensor $\weyl$, seen as a tensor of type $(1,3)$, the
construction involves the Schouten tensor and the Cotton tensor
\[
	\schouten_{ab} \coloneqq \frac{1}{n-2}\left(\ric_{ab} - \frac{\scal}{2(n-1)}\, g_{ab}\right),
	\qquad
	\cotton_{abc} \coloneqq \nabla^g_a \schouten_{bc} - \nabla^g_b \schouten_{ac},
\]
the latter being a $2$-form in $(a,b)$ with values in $T^*M$. Through this
identification, the bundle $\co(TM)$ acts on every tensor bundle over $M$. On
the three tensors above, this action reads
\begin{equation}\label{eqCoAction}
	\left\lbrace
	\begin{aligned}
		(\beta \cdot \schouten)_{ab}  & = -\schouten_{cb}\, \beta_a{}^c - \schouten_{ac}\, \beta_b{}^c,                                                                  \\
		(\beta \cdot \weyl)^d{}_{abc} & = \beta_e{}^d\, \weyl^e{}_{abc} - \weyl^d{}_{ebc}\, \beta_a{}^e - \weyl^d{}_{aec}\, \beta_b{}^e - \weyl^d{}_{abe}\, \beta_c{}^e, \\
		(\beta \cdot \cotton)_{abc}   & = - \cotton_{dbc}\, \beta_a{}^d - \cotton_{adc}\, \beta_b{}^d - \cotton_{abd}\, \beta_c{}^d.
	\end{aligned}
	\right.
\end{equation}
Since the Levi-Civita connection is torsion-free, the Lie derivative along $X$
of a tensor field $T$ is\footnote{Here $X$ is not necessarily a conformal
	Killing vector. The definition of the action is to be understood as extended to
	the whole of $\mathfrak{gl}(TM)$.}
\begin{equation}\label{eqLieRewrite}
	\lie_X T = \nabla^g_X T - (\nabla^g X) \cdot T.
\end{equation}
For instance, for a symmetric $2$-tensor $T$ and vector fields $Y, Z$, using
$[X, Y] = \nabla^g_X Y - \nabla^g_Y X$, we have
\begin{align*}
	(\lie_X T)(Y,Z)
	 & = X\bigl(T(Y,Z)\bigr) - T([X,Y], Z) - T(Y, [X,Z])                \\
	 & = (\nabla^g_X T)(Y,Z) + T(\nabla^g_Y X, Z) + T(Y, \nabla^g_Z X).
\end{align*}

\begin{definition}\label{defAdjointTractor}
	The bundle $\cA$ is the \emph{adjoint tractor bundle}, and the
	\emph{prolongation connection} of the conformal Killing equation is the linear
	connection $\nabla^{\cA}$ on $\cA$ defined by
	\[
		\begin{aligned}
			\nabla^{\cA}_a (\xi, \beta, \kappa) = \Bigl(\;
			 & \nabla^g_a \xi_b - \beta_{ab} \;;\;                                                                               \\
			 & \nabla^g_a \beta_{bc} - \riemuddd{d}{a}{b}{c}\, \xi_d - \kappa_a g_{bc} - \kappa_b g_{ac} + \kappa_c g_{ab} \;;\; \\
			 & \nabla^g_a \kappa_b + \bigl(\nabla^g_\xi \schouten - \beta \cdot \schouten\bigr)_{ab} \;\Bigr).
		\end{aligned}
	\]
\end{definition}

\begin{proposition}\label{propProlongation}
	Let $X$ be a conformal Killing vector on $U$. The section $\mathscr{X}$ is
	parallel for $\nabla^{\cA}$. Conversely, every $\nabla^{\cA}$-parallel section
	of $\cA$ over $U$ is of the form $\mathscr{X}$ for a unique conformal Killing
	vector $X$ on $U$.
\end{proposition}

\begin{proof}
	The first component of $\nabla^{\cA} \mathscr{X}$ vanishes by definition of
	$\mathscr{X}$, and the second one by~\eqref{eqRestricted2ndJet}. For the third
	one, note that $\lie_X g = \frac{2}{n} (\divg^g X)\, g$, so that the local flow
	$\phi_t$ of $X$ satisfies $\phi_t^* g = e^{2\omega_t} g$ for functions
	$\omega_t$ with $\omega_0 = 0$ and $\del_t \omega_t\vert_{t=0} = \frac{1}{n}
		\divg^g X$. By naturality of the Schouten tensor and its transformation law
	under conformal changes of the metric, \cf \cite[\S~2.2]{CurryGover}, we have
	\[
		\phi_t^* \schouten = \schouten - \nabla^g \upd \omega_t + \upd \omega_t \otimes \upd \omega_t - \frac{1}{2} |\upd \omega_t|_g^2\, g.
	\]
	Differentiating at $t = 0$ gives $\lie_X \schouten = -\nabla^g \kappa$, which
	is the vanishing of the third component by~\eqref{eqLieRewrite}. The same
	prolonged system appears in~\cite[App.~A.2]{CurryGover}, written in terms of
	$\kappa + \schouten(X, \cdot)$ instead of $\kappa$.

	Conversely, let $(\xi, \beta, \kappa)$ be a parallel section over $U$. The
	first component says that $\nabla^g \xi^\flat = \beta$ takes values in
	$\co(TM)$. The trace-free symmetric part $\frac{1}{2} \bL \xi$ of $\nabla^g
		\xi^\flat$ therefore vanishes, so that $X \coloneqq \xi$ is a conformal Killing
	vector with $\beta = \nabla^g X^\flat$. Subtracting~\eqref{eqRestricted2ndJet}
	from the second component, the $1$-form $\delta \coloneqq \kappa - \frac{1}{n}
		\upd(\divg^g X)$ satisfies $\delta_a g_{bc} + \delta_b g_{ac} - \delta_c g_{ab}
		= 0$, and taking the trace in $(b,c)$ gives $n\, \delta_a = 0$. Hence $(\xi,
		\beta, \kappa) = \mathscr{X}$, and $X = \xi$ is unique.
\end{proof}

A parallel section being determined by its value at a single point, and $U$
being connected, we obtain
\begin{equation}\label{eqJetDeterminesX}
	X \equiv 0 \quad\Longleftrightarrow\quad (\xi, \beta, \kappa) = 0
\end{equation}
for every point $p_0 \in U$, where $(\xi, \beta, \kappa)$ is the value of
$\mathscr{X}$ at $p_0$.

A parallel section of a bundle with connection lies, at every point, in the
kernel of the curvature of that connection. We now compute this curvature, with
the convention
\[
	\riem^{\nabla}(Y,Z) = \nabla_Y \nabla_Z - \nabla_Z \nabla_Y - \nabla_{[Y,Z]}.
\]

\begin{lemma}\label{lemTractorCurvature}
	In the splitting of $\cA$, the curvature of $\nabla^{\cA}$ is given by
	\[
		\riem^{\nabla^{\cA}}(Y,Z)\,(\xi, \beta, \kappa) = (0, \beta', \kappa'),
	\]
	where, with the second index of $\beta'$ raised as above,
	\begin{align*}
		\beta'_c{}^d & = \bigl(\nabla^g_\xi \weyl - \beta \cdot \weyl\bigr)^d{}_{c\,ab}\,Y^a Z^b,                                     \\
		\kappa'_c    & = \bigl(\nabla^g_\xi \cotton - \beta \cdot \cotton\bigr)_{abc}\,Y^a Z^b - \kappa_d\, \weyl^d{}_{cab}\,Y^a Z^b,
	\end{align*}
	for all $Y, Z \in T_p M$ and all $(\xi, \beta, \kappa) \in \cA_p$.
\end{lemma}

\begin{proof}
	Both sides are tensorial, so it is enough to compute $\nabla^{\cA}_a
		\nabla^{\cA}_b - \nabla^{\cA}_b \nabla^{\cA}_a$ from \cref{defAdjointTractor},
	one component at a time. We spell out the first component and indicate how the
	other two are obtained.

	The first component of $\nabla^{\cA}_b (\xi, \beta, \kappa)$ is $\nabla^g_b
		\xi_c - \beta_{bc}$, and its second component is the one displayed in
	\cref{defAdjointTractor}. Applying $\nabla^{\cA}_a$ to the resulting section,
	the first component of $\nabla^{\cA}_a \nabla^{\cA}_b (\xi, \beta, \kappa)$ is
	\[
		\nabla^g_a \nabla^g_b \xi_c - \nabla^g_a \beta_{bc}
		- \bigl(\nabla^g_b \beta_{ac} - \riemuddd{d}{b}{a}{c}\, \xi_d - \kappa_b g_{ac} - \kappa_a g_{bc} + \kappa_c g_{ba}\bigr).
	\]
	Antisymmetrizing in $(a,b)$, the terms in $\nabla^g \beta$ and the terms in
	$\kappa$ cancel, and the Ricci identity turns the first one into
	$-\riemuddd{d}{c}{a}{b}\, \xi_d$. What is left is
	\[
		\bigl(\riemuddd{d}{b}{a}{c} - \riemuddd{d}{a}{b}{c} - \riemuddd{d}{c}{a}{b}\bigr)\, \xi_d,
	\]
	which vanishes by the first Bianchi identity.

	The second component is computed in the same way. The terms carrying a
	derivative of $\xi$ cancel, and so do all the terms in $\nabla^g \kappa$. The
	terms in which the curvature is differentiated combine, by the second Bianchi
	identity, into $-\nabla^g_\xi \riem$, while the Ricci identity applied to
	$\beta$, together with the terms in $\beta$ coming from the second component of
	\cref{defAdjointTractor}, produces $\beta \cdot \riem$. Finally, the third
	component of \cref{defAdjointTractor} contributes $\bigl(\nabla^g_\xi \schouten
		- \beta \cdot \schouten\bigr) \kulk g$. Decomposing $\riem = \weyl + \schouten
		\kulk g$, the terms involving $\schouten$ cancel and the stated expression
	remains.

	For the third component, the Ricci identity applied to $\kappa$ produces
	$-\riemuddd{d}{c}{a}{b}\, \kappa_d$, whose $\schouten \kulk g$ part cancels
	against the terms in $\kappa$ coming from the second component of
	\cref{defAdjointTractor}, leaving $-\kappa_d\, \weyl^d{}_{cab}$. The
	derivatives of the Schouten terms give $\nabla^g_\xi \cotton$ by the definition
	of $\cotton$, and the terms in $\beta$ give $-\beta \cdot \cotton$.
\end{proof}

As a consequence, we obtain the following result.

\begin{corollary}\label{corObstruction}
	Let $X$ be a conformal Killing vector in a neighborhood of a point $p_0$, and
	let $(\xi, \beta, \kappa)$ be the value of $\mathscr{X}$ at $p_0$. Then, at
	$p_0$,
	\begin{align}
		\bigl(\nabla^g_\xi \weyl - \beta \cdot \weyl\bigr)^d{}_{abc}                              & = 0, \label{eqObstructionW} \\
		\bigl(\nabla^g_\xi \cotton - \beta \cdot \cotton\bigr)_{abc} - \kappa_d\, \weyl^d{}_{cab} & = 0. \label{eqObstructionC}
	\end{align}
	In other words, $(\xi, \beta, \kappa)$ lies in the kernel of the linear map
	\[
		\cO_{p_0} \colon T_{p_0}M \oplus \co(T_{p_0}M) \oplus T^*_{p_0}M
		\lto \cW_{p_0} \oplus \cC_{p_0}
	\]
	whose two components are the left hand sides of~\eqref{eqObstructionW}
	and~\eqref{eqObstructionC}. Here $\cW_{p_0}$ denotes the space of algebraic
	Weyl tensors at $p_0$, \cf \cite[Ch.~1, \S~G]{Besse}, and $\cC_{p_0}$ the space
	of algebraic Cotton tensors, that is, of the $3$-tensors $T_{abc}$ at $p_0$
	satisfying
	\[
		T_{abc} = -T_{bac},
		\qquad
		T_{abc} + T_{bca} + T_{cab} = 0,
		\qquad
		g^{bc}\, T_{abc} = 0,
	\]
	which are the symmetries of $\cotton$, the last one following from the
	contracted Bianchi identity.
\end{corollary}

\begin{proof}
	The section $\mathscr{X}$ is parallel by \cref{propProlongation}, so its value
	at $p_0$ lies in the kernel of $\riem^{\nabla^{\cA}}(Y,Z)$ for all $Y$ and $Z$,
	which by \cref{lemTractorCurvature} amounts to~\eqref{eqObstructionW}
	and~\eqref{eqObstructionC}.
\end{proof}

\begin{corollary}\label{corReductionToKernel}
	Suppose that the map $\cO_{p_0}$ of \cref{corObstruction} is injective. Then
	$g$ admits no non-zero conformal Killing vector on any connected neighborhood
	of $p_0$. In particular $(M,g)$ has no non-zero conformal Killing vector, that
	is, $g \in \NCK(M)$.
\end{corollary}

\begin{proof}
	Let $U$ be a connected neighborhood of $p_0$ (that can be the whole of $M$) and
	let $X$ be a conformal Killing vector on $U$. By \cref{corObstruction} its
	$2$-jet $(\xi, \beta, \kappa)$ at $p_0$ lies in the kernel of $\cO_{p_0}$, so
	it vanishes, and $X \equiv 0$ on $U$ by~\eqref{eqJetDeterminesX}.
\end{proof}

We now aim to produce metrics for which $\cO_{p_0}$ is injective, and to show
that this can be achieved by a small perturbation localized near $p_0$. Fix
local coordinates around $p_0$ and use these to identify the space of $4$-jets
of Riemannian metrics at $p_0$ with a subset of the finite-dimensional vector
space of jet coefficients, namely the open subset defined by positive
definiteness of the $0$-th order term.

\begin{lemma}\label{lemZariski}
	The set $Z$ of $4$-jets of Riemannian metrics at $p_0$ for which $\cO_{p_0}$ is
	not injective is Zariski-closed: it is the common zero locus, inside the open
	set of $4$-jets of Riemannian metrics, of a finite family of polynomials in the
	jet coefficients.
\end{lemma}

\begin{proof}
	In some fixed coordinates around $p_0$, the tensors $\weyl$, $\nabla^g\weyl$,
	$\cotton$ and $\nabla^g \cotton$ at $p_0$ are given by expressions that are
	polynomial in the derivatives of order at most $4$ of the coefficients of $g$
	at $p_0$ and in the coefficients of $g(p_0)^{-1}$. The map $\cO_{p_0}$ depends
	linearly on these tensors, and a linear map between two fixed vector spaces
	fails to be injective exactly when all its maximal minors vanish\footnote{This
		uses the fact that the source of $\cO_{p_0}$ has dimension less than or equal
		to that of its target, so that the maximal minors are the ones of the size of
		the source and their simultaneous vanishing is the failure of injectivity
		rather than of surjectivity. The source $T_{p_0}M \oplus \co(T_{p_0}M) \oplus
			T^*_{p_0}M$ has dimension $2n + \tfrac{n(n-1)}{2} + 1$, while the target
		$\cW_{p_0} \oplus \cC_{p_0}$ has dimension $\tfrac{n(n+1)(n+2)(n-3)}{12} +
			\tfrac{n(n^2-4)}{3}$, \cf \cite[\S~6.7]{Weinberg} for the first summand
		and~\cite{GarciaHehlHeinickeMacias} for the second. As the reader may check,
		the target exceeds the source for every $n \geq 4$. This fact will however also
		be a direct consequence of the fact that there exists at least one injective
		map.}. After multiplication by a suitable positive power of $\det g(p_0)$,
	these minors become polynomials in the jet coefficients, and $Z$ is their
	common zero locus.
\end{proof}

Suppose now that some $4$-jet of Riemannian metrics lies outside $Z$, so that
one of these polynomials, say $P$, is not identically zero. The set of $4$-jets
of Riemannian metrics is open, so $P$ does not vanish identically on any
non-empty open subset of it. Hence $\{P \neq 0\}$ is dense, and the complement
of $Z$, which is open and contains $\{P \neq 0\}$, is dense as well. In this
sense the metrics we are looking for are generic. What we shall actually need
is a quantitative version, in which the perturbation is supported in a
prescribed neighborhood of $p_0$ and is small in a prescribed $C^m$-norm. This
is \cref{propLocalDefNCK} below, and the last ingredient of its proof is one
explicit metric whose obstruction is injective.

\begin{lemma}\label{lemGoodJet}
	Assume $n \geq 4$ and let $\lambda_1, \ldots, \lambda_n$ be pairwise distinct
	non-zero real numbers. Let $f$ be the quadratic form on $\bR^n$ given by
	\[
		f(x^1, \ldots, x^n) \coloneqq \frac{1}{2} \sum_{i=1}^n \lambda_i\, (x^i)^2,
	\]
	and let $\Sigma \subset \bR^{n+1}$ be its graph. Denoting by
	$g_{\mathrm{eucl}}$ the Euclidean metric of $\bR^n$, consider the metric
	\[
		g \coloneqq g_{\mathrm{eucl}} + \upd f \otimes \upd f,
	\]
	which, equivalently, can be described as the pull-back under $x \mapsto (x,
		f(x))$ of the metric induced on $\Sigma$ by the Euclidean metric of
	$\bR^{n+1}$. Then the map $\cO_{p_0}$ of $g$ at the origin $p_0 = 0$ is
	injective.
\end{lemma}

\begin{proof}
	Throughout the proof, we perform calculations with respect to the standard
	coordinates of $\bR^n$. The indices $i, j, k, \ldots$ are used to denote tensor
	components in these coordinates, in contrast to the abstract indices $a, b, c,
		\ldots$ used before. We do not make use of the Einstein summation convention
	because of the specific form of the metric, so all summations will be indicated
	explicitly.

	We first observe that, for any $i=1, \ldots, n$, the metric is invariant under
	reflection of the $i$-th coordinate. As a consequence, any tensor $T$ that is
	naturally built from $g$, such as the various curvature tensors or their
	covariant derivatives, is invariant as well. Thus all of its component
	functions $T_{i_1 \ldots i_N}$ are even or odd in $x^i$ according to the parity
	of the number of indices among $i_1, \ldots, i_N$ that are equal to $i$. In
	particular, the component functions can only be non-zero at the origin if every
	index appears an even number of times. This fact will simplify many of the
	computations we are going to perform.

	The metric components are explicitly given by $g_{ij}(x) = \delta_{ij} +
		\lambda_i \lambda_j\, x^i x^j$ and thus its inverse has components $g^{ij}(x) =
		\delta_{ij} - \rho(x) \lambda_i \lambda_j x^i x^j$ with
	\[
		\rho(x) \coloneqq \Bigl(1 + \sum_{i=1}^n \lambda_i^2 (x^i)^2\Bigr)^{-1}.
	\]
	The curvature is most easily computed by viewing $(\bR^n, g)$ as the
	hypersurface $\Sigma$, as in the statement of the lemma. Since the second
	fundamental form of $\Sigma$ is given by
	\[
		S_{ij}(x) = \frac{\hess_{ij} f}{\sqrt{1 + |\nabla f|^2}} = \rho(x)^{1/2}\lambda_i \delta_{ij},
	\]
	the Gauss equation allows us to conclude
	\begin{gather} \label{eqRiemCurv}
		\riem_{ijk\ell}(x) = \rho(x) \lambda_i \lambda_j \bigl(\delta_{ik}\delta_{j\ell} - \delta_{i\ell}\delta_{jk}\bigr).
	\end{gather}
	Denoting by $\sigma_1 \coloneqq \sum_{i=1}^n \lambda_i$ and $\sigma_2 \coloneqq
		\sum_{1 \leq i < j \leq n} \lambda_i \lambda_j$ the first two symmetric
	polynomials evaluated at the numbers $\lambda_i$, we get the following
	expressions for the components of the Ricci and the scalar curvature at the
	origin:
	\[
		\ric_{j\ell}(0) = \lambda_j(\sigma_1 - \lambda_j)\delta_{j\ell},
		\qquad
		\scal(0) = 2 \sigma_2.
	\]

	The Schouten tensor at the origin is diagonal as well,
	\[
		\schouten_{j\ell}(0) = \mu_j\delta_{j\ell} \quad\text{with } \mu_j = \frac{1}{n-2}\Bigl(\lambda_j(\sigma_1 - \lambda_j) - \frac{\sigma_2}{n-1}\Bigr).
	\]
	Therefore, the Weyl tensor at the origin $\weyl_0 \coloneqq \weyl(0) = \riem(0)
		- \schouten(0) \kulk g$ is given by
	\begin{equation}\label{eqW0Diag}
		(\weyl_0)_{ijk\ell} = w_{ij}\,\bigl(\delta_{ik}\delta_{j\ell} - \delta_{i\ell}\delta_{jk}\bigr)
	\end{equation}
	with
	\[
		w_{ik} = (\weyl_0)_{ikik} = \lambda_i \lambda_k - \mu_i - \mu_k \quad \text{for } i \neq k, \quad \text{and } w_{ii} = 0.
	\]

	The differences of the $w$'s, which we shall use twice, factor as
	\[
		w_{jk} - w_{ik} = (\lambda_j - \lambda_i)\,\bigl(\lambda_k - q_{ij}\bigr),
		\qquad
		q_{ij} \coloneqq \frac{\sigma_1 - \lambda_i - \lambda_j}{n-2}
	\]
	for $k \notin \{i, j\}$. Now fix $i \neq j$. As the $\lambda$'s are pairwise
	distinct, $\lambda_j - \lambda_i \neq 0$, while $\lambda_k = q_{ij}$ can hold
	for at most one index $k$. Since $n \geq 4$ leaves $n - 2 \geq 2$ choices of $k
		\neq i,j$, there is such a $k$ with
	\begin{equation}\label{eqWDiffNonZero}
		w_{jk} - w_{ik} \neq 0,
	\end{equation}
	so that at least one of $w_{jk}$ and $w_{ik}$ is non-zero. In particular
	$\weyl_0 \neq 0$.

	We first look at the obstruction equation~\eqref{eqObstructionW} at the origin
	$p_0 = 0$. Since $\nabla\weyl$ is a tensor of odd rank built from the metric,
	it vanishes at the origin, $\nabla\weyl(0) = 0$, by the remark at the beginning
	of the proof. Thus the obstruction reduces to $\beta \cdot \weyl_0 = 0$. Write
	$\beta = \omega + \phi\, g$ with $\omega \in \Lambda^2 T^*_{p_0}M$ and $\phi
		\in \bR$. The metric corresponds to the identity endomorphism, so
	by~\eqref{eqCoAction} it acts by $g \cdot \weyl_0 = -2\weyl_0$, and the
	obstruction reads
	\begin{equation}\label{eqBetaSplit}
		\omega \cdot \weyl_0 - 2\phi\, \weyl_0 = 0.
	\end{equation}
	The infinitesimal rotation $\omega$ preserves the fiber metric on tensors, so
	its action is skew-adjoint and $\langle \omega \cdot \weyl_0, \weyl_0 \rangle =
		0$. Pairing~\eqref{eqBetaSplit} with $\weyl_0$ thus gives $-2\phi\, |\weyl_0|^2
		= 0$, and as $\weyl_0 \neq 0$, we get $\phi = 0$ and $\omega \cdot \weyl_0 =
		0$.

	We still need to prove that $\omega = 0$, for which we compute $\omega \cdot
		\weyl_0$ using~\eqref{eqW0Diag}. Since the action of $\SO(TM)$ preserves the
	metric, it is compatible with raising and lowering of indices, and the same is
	true for its derivative, \ie we have
	\[
		(\omega \cdot \weyl_0)_{ijk\ell} = -\sum_{m=1}^n \left((\weyl_0)_{mjk\ell}\,\omega_i{}^m + (\weyl_0)_{imk\ell}\, \omega_j{}^m + (\weyl_0)_{ijm\ell}\, \omega_k{}^m + (\weyl_0)_{ijkm}\, \omega_{\ell}{}^m \right).
	\]
	Together with~\eqref{eqW0Diag}, we get
	\begin{equation}\label{eqOmegaW0}
		\begin{split}
			(\omega \cdot \weyl_0)_{ijk\ell}
			 & = \omega_{ik}(w_{ij} - w_{jk})\, \delta_{j\ell} + \omega_{j\ell}(w_{ij} - w_{i\ell})\, \delta_{ik}         \\
			 & \qquad + \omega_{i\ell}(w_{j\ell} - w_{ij})\, \delta_{jk} + \omega_{jk}(w_{ik} - w_{ij})\, \delta_{i\ell}.
		\end{split}
	\end{equation}
	We now specialize the indices. For pairwise distinct $i,j,k$, in the expression
	for the component $(i,k,j,k)$ only the first summand of~\eqref{eqOmegaW0}
	survives, giving $(\omega \cdot \weyl_0)_{ikjk} = \omega_{ij}\,(w_{ik} -
		w_{jk})$. If $\omega \cdot \weyl_0 = 0$, this vanishes for all pairwise
	distinct $i,j,k$. Fixing $i \neq j$ and choosing $k$ as
	in~\eqref{eqWDiffNonZero} forces $\omega_{ij} = 0$. As $i \neq j$ was arbitrary
	and $\omega$ is antisymmetric, this implies $\omega = 0$ and therefore $\beta =
		0$.

	Now, we turn our attention to the other obstruction~\eqref{eqObstructionC}.
	Since the Cotton tensor $\cotton$ has odd order, the parity principle described
	in the beginning of the proof gives $\cotton(0) = 0$. For $\nabla \cotton(0)$,
	it yields that the only non-vanishing components are those where two indices
	$i, j$ appear twice. By the anti-symmetry of $\cotton$ in the first two
	entries, this implies that $\nabla \cotton(0)$ is completely determined by the
	coefficients $v_{ij} \coloneqq (\nabla_i\cotton)_{ijj}(0)$, namely through
	\begin{equation*}
		(\nabla_\ell\cotton)_{ijk}(0) = v_{ij}\,\delta_{i\ell}\delta_{jk} - v_{ji}\,\delta_{j\ell}\delta_{ik}.
	\end{equation*}

	In order to calculate these coefficients, we need to determine the Ricci and
	scalar curvature up to second order around the origin. Noting that $\rho(x) = 1
		- \sum_{i = 1}^n \lambda_i^2 (x^i)^2 + \mathcal{O}(x^4)$ and $g^{ij}(x) =
		\delta_{ij} - \lambda_i\lambda_j x^i x^j + \mathcal{O}(x^4)$, we obtain the
	following from~\eqref{eqRiemCurv}:
	\begin{align*}
		\ric_{j\ell}(x) & = \left(1-  \sum_{i = 1}^n \lambda_i^2 (x^i)^2 \right)\ric_{j\ell}(0) - \left(\sum_{i=1}^n \lambda_j\delta_{j\ell}\lambda_i^3 (x^i)^2 - \lambda_j^2 \lambda_\ell^2 x^j x^\ell \right) + \mathcal{O}(x^4) \\
		\scal(x)        & = \left(1-  \sum_{i = 1}^n \lambda_i^2 (x^i)^2 \right)\scal(0) - \left(\sum_{i=1}^n \sigma_1 \lambda_i^3 (x^i)^2 - \sum_{i=1}^n \lambda_i^4 (x^i)^2\right)                                               \\
		                & \phantom{=}\;-\sum_{j,\ell = 1}^n \ric_{j\ell}(0) \lambda_j \lambda_\ell x^j x^\ell + \mathcal{O}(x^4)                                                                                                   \\
		                & = \left(1-  \sum_{i = 1}^n \lambda_i^2 (x^i)^2 \right)\scal(0) - 2\sum_{i=1}^n (\sigma_1 - \lambda_i) \lambda_i^3 (x^i)^2 + \mathcal{O}(x^4).
	\end{align*}
	Moreover, we need to know the Christoffel symbols, which are computed to
	$\Gamma_{ij}^k(x) = \sum_{\ell = 1}^n g^{k\ell}(x) \lambda_i \delta_{ij}
		\lambda_\ell x^\ell$ using the Koszul formula, yielding
	\begin{gather*}
		\Gamma_{ij}^k(0) = 0 \qquad \text{and }\quad \del_m \Gamma_{ij}^k(0) = \lambda_i  \lambda_k \delta_{ij} \delta_{km}.
	\end{gather*}
	Thus, for $i \neq j$, the relevant second derivatives of the Schouten tensor
	$\schouten$ are given by
	\begin{align*}
		\nabla_i \nabla_i \schouten_{jj}(0) & = \del_i \del_i \schouten_{jj}(0) - 2\sum_{k=1}^n (\del_i \Gamma_{ij}^k(0)) \schouten_{kj}(0)                                                                                                     \\
		                                    & = \frac{2\lambda_i^2}{n-2} \left(-\lambda_j(\sigma_1 - \lambda_j) - \lambda_j \lambda_i + \frac{\sigma_2 + \lambda_i(\sigma_1 - \lambda_i)}{n-1}\right) -2 \cdot 0                                \\
		                                    & = \frac{\lambda_i^2}{n-2} \left(2\lambda_j(\lambda_j - \lambda_i -\sigma_1) + \frac{2\bigl(\sigma_2 + \lambda_i(\sigma_1 - \lambda_i)\bigr)}{n-1}\right)
		\intertext{and}
		\nabla_i \nabla_j \schouten_{ij}(0) & = \del_i \del_j \schouten_{ij}(0) - \sum_{k=1}^n (\del_i \Gamma_{ji}^k(0)) \schouten_{kj}(0)  - \sum_{k=1}^n (\del_i \Gamma_{jj}^k(0)) \schouten_{ik}(0)                                          \\
		                                    & = \frac{\lambda_i \lambda_j}{n-2} \left(\lambda_i\lambda_j - \frac{\sigma_2}{n-1}\right)  - 0 - \frac{\lambda_i \lambda_j}{n-2}\left(\lambda_i(\sigma_1 - \lambda_i)- \frac{\sigma_2}{n-1}\right) \\
		                                    & = \frac{\lambda_i^2}{n-2} \Bigl(\lambda_j(\lambda_j + \lambda_i - \sigma_1) \Bigr).
	\end{align*}
	Taking the difference, we obtain
	\[
		v_{ij} = \frac{\lambda_i^2}{n-2}\Bigl(\lambda_j(\lambda_j - 3\lambda_i - \sigma_1)
		+ \frac{2\bigl(\sigma_2 + \lambda_i(\sigma_1 - \lambda_i)\bigr)}{n-1}\Bigr).
	\]

	With this at hand, we return to the second obstruction~\eqref{eqObstructionC}
	and notice that the summand involving $\beta$ vanishes at the origin due to
	$\cotton(0) = 0$. Considering, for $i \neq j$, the $(i,j,j)$ component of the
	remaining equation and using that $\nabla_k \cotton_{ijj}(0) = v_{ij}
		\delta_{ik}$ and $\weyl^k{}_{jij}(0) = w_{ij} \delta_{ik}$, we arrive at
	\begin{equation}\label{eqQuadricSystem}
		\xi_i v_{ij} = \kappa_i w_{ij}.
	\end{equation}
	We now fix $i$ and let $j \neq i$ vary. The $\lambda$'s being fixed, so are
	$\sigma_1$, $\sigma_2$ and the $\mu$'s, and both sides
	of~\eqref{eqQuadricSystem} depend on $j$ only through the value $\lambda_j$.
	Introducing the two polynomials in one indeterminate $\lambda$
	\begin{align*}
		\hat{w}_i(\lambda) & \coloneqq \frac{\lambda^2}{n-2}
		+ \Bigl(\lambda_i - \frac{\sigma_1}{n-2}\Bigr)\lambda
		- \mu_i + \frac{\sigma_2}{(n-1)(n-2)},                                \\
		\hat{v}_i(\lambda) & \coloneqq \frac{\lambda_i^2}{n-2}\Bigl(\lambda^2
		- (3\lambda_i + \sigma_1)\,\lambda
		+ \frac{2\bigl(\sigma_2 + \lambda_i(\sigma_1 - \lambda_i)\bigr)}{n-1}\Bigr),
	\end{align*}
	we have $w_{ij} = \hat{w}_i(\lambda_j)$ and $v_{ij} = \hat{v}_i(\lambda_j)$ for
	every $j \neq i$.

	These two polynomials are linearly independent. Suppose that $s \hat{v}_i + t
		\hat{w}_i = 0$ for some $s, t \in \bR$. Comparing the coefficients of
	$\lambda^2$ gives $s \lambda_i^2 + t = 0$, so $s = 0$ forces $t = 0$. Assuming
	$s \neq 0$, we may therefore normalize the relation to $\hat{v}_i =
		\lambda_i^2\, \hat{w}_i$, and comparing in it the coefficients of $\lambda$
	gives
	\[
		-\frac{\lambda_i^2\,(3\lambda_i + \sigma_1)}{n-2}
		= \lambda_i^2 \Bigl(\lambda_i - \frac{\sigma_1}{n-2}\Bigr),
	\]
	that is $(n+1)\,\lambda_i^3 = 0$, which is impossible since $\lambda_i \neq 0$.
	Hence $s = t = 0$.

	We can now conclude. By~\eqref{eqQuadricSystem}, the polynomial $\xi_i\,
		\hat{v}_i - \kappa_i\, \hat{w}_i$ vanishes at $\lambda_j$ for every $j \neq i$.
	These are $n - 1 \geq 3$ pairwise distinct values, while a non-zero polynomial
	of degree at most $2$ has at most two roots. Hence $\xi_i\, \hat{v}_i -
		\kappa_i\, \hat{w}_i$ is the zero polynomial, and the linear independence just
	established gives $\xi_i = \kappa_i = 0$. As $i$ was arbitrary, $\xi$ and
	$\kappa$ vanish. Together with $\beta = 0$ this shows that $\cO_{p_0}$ is
	injective.
\end{proof}

We can now assemble what precedes into the statement that will actually be used
in the proof of \cref{thmReductionToNCK}.

\begin{proposition}\label{propLocalDefNCK}
	Let $(M,g_0)$ be a Riemannian manifold of dimension $n \geq 3$ and $U_0 \subset
		M$ an open neighborhood of a point $p_0 \in M$. Then for every $\epsilon > 0$
	and every $m \geq 0$ there exists a symmetric $2$-tensor $h$ supported in $U_0$
	with $\|h\|_{C^m} < \epsilon$ such that $g_0 + h$ is a Riemannian metric
	admitting no non-zero conformal Killing vector on any connected neighborhood of
	$p_0$. Moreover, this holds for all metrics in a sufficiently small
	$C^K$-neighborhood of $g_0 + h$, where one may take $K = 4$ if $n \geq 4$ and
	$K = 5$ if $n = 3$.
\end{proposition}

\begin{proof}
	Assume first $n \geq 4$. By \cref{lemGoodJet} there is a $4$-jet $j^*$ of
	Riemannian metrics at $p_0$ for which $\cO_{p_0}$ is injective. Let $P$ be a
	maximal minor of $\cO_{p_0}$, viewed as a polynomial in the jet coefficients as
	in the proof of \cref{lemZariski}, so that $P(j^*) \neq 0$.

	We first build a Riemannian metric $\gamma$ on $M$ whose $4$-jet at $p_0$ is
	$j^*$ and which coincides with $g_0$ outside $U_0$. Let $\gamma_0$ be the
	polynomial representative of $j^*$ in some fixed coordinates around $p_0$. Its
	value at $p_0$ is positive definite, so, after shrinking $U_0$, we may assume
	that $\gamma_0$ is positive definite on $U_0$. Choose a cutoff function $\chi$
	supported in $U_0$, equal to $1$ near $p_0$ and with values in $[0,1]$, and set
	$\gamma \coloneqq \chi \gamma_0 + (1-\chi) g_0$. At every point of $U_0$ this
	is a convex combination of two positive definite forms, hence positive
	definite, and outside $U_0$ it equals $g_0$, so $\gamma$ is a Riemannian
	metric. Since $\chi \equiv 1$ near $p_0$, we have $j^4_{p_0}(\gamma) = j^*$.

	Now consider, for $t \in [0,1]$,
	\[
		g_t \coloneqq (1-t)\, g_0 + t\, \gamma.
	\]
	Each $g_t$ is a Riemannian metric, being a convex combination of two of them,
	and $g_t - g_0 = t (\gamma - g_0)$ is supported in $U_0$. Its $4$-jet at $p_0$
	is $(1-t)\, j^4_{p_0}(g_0) + t\, j^*$, so $t \mapsto
		P\bigl(j^4_{p_0}(g_t)\bigr)$ is a polynomial in $t$ whose value at $t = 1$ is
	$P(j^*) \neq 0$. A non-zero polynomial has finitely many roots, so this
	function is non-zero for all but finitely many $t \in [0, 1]$, and in
	particular for small enough $t > 0$.

	Choosing such a $t$ with moreover $t\, \|\gamma - g_0\|_{C^m} < \epsilon$, the
	tensor $h \coloneqq t (\gamma - g_0)$ has all the required properties. Indeed,
	set $g \coloneqq g_0 + h$. Its map $\cO_{p_0}$ is injective, so by
	\cref{corReductionToKernel} it admits no non-zero conformal Killing vector on
	any connected neighborhood of $p_0$. Finally, injectivity of $\cO_{p_0}$ is the
	non-vanishing of one of its maximal minors, and the tensors $\weyl$,
	$\nabla\weyl$, $\cotton$ and $\nabla\cotton$ at $p_0$ depend continuously on
	the $C^4$-jet of the metric, so the same conclusion holds for every metric in a
	sufficiently small $C^4$-neighborhood of $g$.

	For $n = 3$ the argument is the same once $\cO_{p_0}$ is replaced by the
	pointwise obstruction of Beig, Chru\'sciel and Schoen. Their Theorem~2.1
	provides a non-trivial polynomial $Q$ in a finite jet of the metric at $p_0$,
	built from the Cotton tensor and its derivatives of order at most two, such
	that $Q \neq 0$ at $p_0$ forbids conformal Killing vectors on a neighborhood of
	$p_0$, \cf \cite[Thm.~2.1, Thm.~4.1 and Prop.~4.2]{BeigChruscielSchoen}.
	Neither \cref{lemZariski} nor the interpolation above uses anything else about
	$\cO_{p_0}$ than that its degeneracy implies the vanishing of a polynomial in a
	finite jet, so both go through with $Q$ in place of a maximal minor. Only the
	appeal to \cref{lemGoodJet}, which is what restricted us to $n \geq 4$, has to
	be replaced, by the non-triviality of $Q$. The one change in the conclusion is
	that the neighborhood becomes a $C^5$-neighborhood.
\end{proof}

We now show how to conclude \cref{thmReductionToNCK} from this. Let us first
explain the general idea. \Cref{propLocalDefNCK} allows us to slightly perturb
a single metric to one without non-zero CKVs. Suppose now we want to find a
slight perturbation of a continuous path of metrics $[0,1] \lto \Met(M),\, x
	\mapsto g_x$ so that the perturbed path avoids metrics with non-zero CKVs. In
this case, running \cref{propLocalDefNCK} at any $x_0 \in [0,1]$ provides us
with a symmetric $2$-tensor $h$ so that $g_x + h$ has no non-zero CKVs for all
$x \in (x_0 - \delta, x_0 + \delta)$, for some small $\delta > 0$. Due to
compactness of $[0,1]$, it is possible to cover the whole interval~$[0,1]$ by
finitely many small intervals of this kind, and let us assume without loss of
generality that they are given by
\[
	I_k \coloneqq \left(\frac{k-1.03}{N}, \frac{k+1.03}{N} \right) \cap [0,1],
	\qquad k = 0, \ldots, N.
\]
Note that the $I_k$ with $k$ even already cover $[0,1]$, and so do those with
$k$ odd, two consecutive intervals of the same parity overlapping only in an
interval of length $0.06/N$.

One would now like to use a suitable partition of unity of $[0,1]$ to produce
from the different perturbation tensors $h_k$ (associated to $I_k$) a single
$x$-dependent perturbation $h_x$ such that for all $x \in [0,1]$ the metric
$g_x + h_x$ has no non-zero CKVs. This, however, does not work. In general, if
for $x \in [0,1]$ there are two perturbations $h$ and $\tilde{h}$ so that $g_x
	+ h$ and $g_x + \tilde{h}$ admit no non-zero CKVs, there is no reason to assume
that an arbitrary interpolation between $h$ and $\tilde{h}$ has the same
property.

What helps us instead is that \cref{propLocalDefNCK} rules out non-zero
\emph{local CKVs} around a point in~$M$. If the perturbation $g_x + h$ has no
non-zero CKVs around $p \in M$, then this property is not destroyed by adding a
symmetric $2$-tensor whose support is disjoint from a neighborhood $U$ of $p$.
So what we do is the following: we choose two points $p_0 \neq p_1 \in M$ with
disjoint open neighborhoods $U_0, U_1$. When choosing the perturbation tensors
using \cref{propLocalDefNCK}, we work with $(p_0, U_0)$ for all intervals $I_k$
with $k$ even and $(p_1, U_1)$ for all intervals $I_k$ with $k$ odd. Then we
glue together the perturbation tensors for the even intervals using a partition
of unity subordinate to $(I_k)_{k\,\text{even}}$, yielding $h^{(0)}_x$, and
likewise for the odd intervals, yielding $h^{(1)}_x$. Now notice that $g_x +
	h^{(0)}_x + h^{(1)}_x$ admits no non-zero CKVs around $p_0$ if $x \in
	\bigcup_{k\,\text{even}} \left(\frac{k-0.97}{N},\frac{k+0.97}{N}\right) \cap
	[0,1]$ and no non-zero CKVs around $p_1$ if $x \in \bigcup_{k\,\text{odd}}
	\left(\frac{k-0.97}{N},\frac{k+0.97}{N}\right) \cap [0,1]$. These two unions
cover $[0,1]$, so there are no non-zero CKVs for any $x \in [0,1]$.

We now generalize this strategy to parameter spaces that are (suitable subsets
of) higher-dimensional spheres $S^\ell$. We will have to choose perturbations
$h^{(0)}_x, \ldots, h^{(\ell)}_x$ supported on pairwise disjoint open
neighborhoods $U_0, \ldots, U_\ell$ of $\ell + 1$ points $p_0, \ldots, p_\ell
	\in M$. Each of these perturbations $h^{(j)}_x$ will be glued together from
constant (in $x \in S^\ell$) symmetric $2$-tensors using a partition of unity
of $S^\ell$ that depends on $j$. Where $h^{(j)}_x$ coincides with one of the
symmetric $2$-tensors from which it is built (\ie where one of the functions
forming the $j$-th partition of unity is $1$), the metric $g_x + \sum_{j =
		0}^\ell h_x^{(j)}$ does not admit non-zero CKVs around $p_j$. For each $x \in
	S^\ell$, this should be the case for at least one $j$. A family of $\ell +1$
partitions of unity with this property is constructed in
\cref{lemGoodPartition} below.

The construction of these partitions of unity rests on the following refinement
lemma, which is where the dimension of the parameter space enters.

\begin{lemma}[{\cite[Lem.~2.7]{MunkresDiffTop}}]\label{lemDisjointRefinement}
	Let $X$ be a (smooth) $\ell$-dimensional manifold, $\ell \geq 0$. Then for any
	open cover $(V_i)_{i \in I}$ there exist $\ell + 1$ families $(U^0_i)_{i \in
				I_0}, \ldots, (U^\ell_i)_{i \in I_\ell}$ of open subsets such that $U^j_i \cap
		U^j_{i^\prime} = \emptyset$ for all $j = 0, \ldots, \ell$ and $i \neq i^\prime
		\in I_j$, and such that $(U^j_i)_{0 \leq j \leq \ell,\, i \in I_j}$ is a
	locally finite refinement of $(V_i)_{i \in I}$.
\end{lemma}

\begin{lemma} \label{lemGoodPartition}
	Let $\ell \geq 0$ and suppose $(V_i)_{i \in I}$ is an open cover of $S^\ell$.
	Then there are $\ell+1$ finite partitions of unity $(\eta^0_i)_{1 \leq i \leq
				N_0}, \ldots, (\eta^\ell_i)_{1 \leq i \leq N_\ell}$ such that
	$((\eta_i^j)^{-1}(1))_{0 \leq j \leq \ell, 1 \leq i \leq N_j}$ is a cover of
	$S^\ell$ and $(\supp(\eta_i^j))_{0 \leq j \leq \ell, 1 \leq i \leq N_j}$
	refines $(V_i)_{i \in I}$.
\end{lemma}

\begin{proof}
	Let $(U^j_i)_{0 \leq j \leq \ell,\, i \in I_j}$ be a refinement of $(V_i)_{i
				\in I}$ as in \cref{lemDisjointRefinement}. It is locally finite and $S^\ell$
	is compact, so we may assume that all the $I_j$ are finite. We let
	$(\chi^j_i)_{0 \leq j \leq \ell,\, i \in I_j}$ be a partition of unity
	subordinate to it and we set
	\[
		\chi_j \coloneqq \sum_{i \in I_j} \chi^j_i,
		\qquad
		B_j \coloneqq \chi_j^{-1}\Bigl(\Bigl[\tfrac{1}{\ell+1},1\Bigr]\Bigr),
	\]
	so that $B_j$ is closed and $\sum_{j = 0}^\ell \chi_j = 1$.

	Write $I^\prime \coloneqq \coprod_{j = 0}^\ell I_j$, a finite set, and for $i
		\in I^\prime$ let $j(i)$ be the unique index with $i \in I_{j(i)}$. For every
	$j$ we consider the open sets
	\begin{equation}\label{eqCovSets}
		W^j_i \coloneqq
		\begin{cases}
			U^j_i & \text{if } j(i) = j, \\ U^{j(i)}_i \setminus B_j\ \, &
			   \text{otherwise,}
		\end{cases}
		\qquad i \in I^\prime.
	\end{equation}
	They cover $S^\ell$. Indeed, if $x \notin B_j$, then any $U^{j(i)}_i$
	containing $x$ provides such a set, while if $x \in B_j$, then $\chi_j(x) > 0$,
	so that $\chi^j_i(x) > 0$ and $x \in U^j_i$ for some $i \in I_j$. We let
	$(\eta^j_i)_{i \in I^\prime}$ be a partition of unity subordinate to
	$(W^j_i)_{i \in I^\prime}$. Identifying $I^\prime$ with $\{1, \ldots,
		|I^\prime|\}$, we obtain partitions of unity indexed as in the statement, with
	$N_j = |I^\prime|$ for all $j$.

	The condition on the supports is satisfied, because $\supp(\eta^j_i) \subset
		W^j_i \subset U^{j(i)}_i$ and $(U^j_i)_{j,i}$ refines $(V_i)_{i \in I}$. There
	remains to see that every $x \in S^\ell$ lies in $(\eta^j_i)^{-1}(1)$ for some
	$i$ and $j$. As $\sum_{j = 0}^\ell \chi_j(x) = 1$, there is a $j$ with
	$\chi_j(x) \geq \frac{1}{\ell+1}$, that is, with $x \in B_j$. Exactly one of
	the sets \eqref{eqCovSets} then contains $x$: it is none of the $W^j_i$ with
	$j(i) \neq j$, since these avoid $B_j$, and the others are the $U^j_i$ with $i
		\in I_j$, which are pairwise disjoint. All the $\eta^j_{i^\prime}$ but one
	therefore vanish at $x$, and the remaining one equals $1$ there.
\end{proof}

\begin{remark}
	\begin{itemize}
		\item The proof uses nothing about $S^\ell$ beyond the fact that it is a compact
		      manifold of dimension $\ell$, and the statement holds in that generality.
		\item 	The open cover $((\eta_i^j)^{-1}((\frac{1}{2},1]))_{0 \leq j \leq \ell, 1 \leq
					      i \leq N_j}$ has the property that every point $x \in S^\ell$ is contained in
		      at most $\ell + 1$ of these sets. In fact, all of them must have different
		      $j$-index, because if $\eta^j_{i_0}(x) > \frac{1}{2}$ for some $i_0$ and $j$,
		      then $\eta^j_i(x) < \frac{1}{2}$ for all $i \neq i_0$. We may compare this with
		      a central result of dimension theory which states that a separable metrizable
		      topological space has dimension $\leq \ell$ if and only if every finite open
		      cover admits a refinement with the property that no more than $\ell + 1$ of its
		      open sets have a non-empty intersection \cite[Thm.~V 8]{HurewiczWallman}. For
		      manifolds, the notion of dimension used there coincides with the usual one
		      \cite[Cor.~1 of Thm.~IV 3]{HurewiczWallman}, so the sphere $S^\ell$ has
		      dimension $\ell$ and the $((\eta_i^j)^{-1}((\frac{1}{2},1]))_{0 \leq j \leq
					      \ell, 1 \leq i \leq N_j}$ form a refinement of $(V_i)_{i \in I}$ as guaranteed
		      by dimension theory. In this sense, \cref{lemGoodPartition} enhances the
		      conclusion in the case of compact manifolds.

	\end{itemize}
\end{remark}

\begin{proof}[Proof of \cref{thmReductionToNCK}]
	Let $A \subset \Omega \subset S^\ell$ be as in the statement, $g \colon (\Omega,A) \to (\Met(M), \NCK(M))$ be a map of pairs and $\mathcal{U}$ a $C^\infty$-neighborhood of $g$.
	Since $\Omega$ is compact, after possibly shrinking $\mathcal{U}$ we may assume that it consists of all $\tilde{g}$ with $\max_{x \in \Omega} \|\tilde{g}_x - g_x\|_{C^m} < C$ for a fixed $C^m$-norm on $\Gamma(S_2 M)$ with $m \geq 5$ (so that the $C^K$-stability of \cref{propLocalDefNCK} applies below in either dimension range).

	We first note that a compact Riemannian manifold $(M,g)$ has no non-zero CKV
	if and only if
	\begin{align*}
		\mu(g) \coloneqq \inf_{0 \neq Y \in H^1(M,TM)} \frac{\|\bL^{g} Y\|^2_{L^2(M,g)}}{\|Y\|^2_{H^1(M,g)}} > 0,
	\end{align*}
	see the proof of \cref{propSplitting}. The argument given in the proof of
	\cref{propTTBundle} for the lower semicontinuity of $\mu_V$ applies verbatim
	to $\mu$. Hence, if $g$ has no non-zero CKV, then $\mu(\tilde{g}) > 0$ for all
	$\tilde{g}$ in a $C^1$-neighborhood of $g$, so there are no CKVs for all those
	$\tilde{g}$. Since $A$ is
	compact, we may assume that $m$ and $C$ above are chosen such that for any $a
		\in A$ all metrics $\tilde{g}$ with $\|\tilde{g} - g_a\|_{C^m} < C$ admit no
	CKVs. For each $x_0 \in A$ we may now choose an open neighborhood $V_{x_0}$ of
	$x_0$ in $\Omega$ such that $\|g_x - g_{x_0}\|_{C^m} < \frac{C}{2}$ for all $x \in
		V_{x_0}$ (and in particular $g_x \in \NCK(M)$).

	We now define open neighborhoods $V_{x_0}$ in $\Omega \setminus A$ for all $x_0 \in
		\Omega \setminus A$. To do so, we first choose once and for all $\ell + 1$ different
	points $p_0, \ldots, p_\ell \in M$ with corresponding open neighborhoods $U_0,
		\ldots, U_\ell$ that we assume to be disjoint. Then we apply
	\cref{propLocalDefNCK} to the Riemannian manifold $(M,g_{x_0})$, $x_0 \in \Omega
		\setminus A$, for the point $p_j$ with open neighborhood $U_j$ and with
	$\epsilon = \frac{C}{2(\ell+1)}$. We obtain $h^j_{x_0}$ with $\supp(h^j_{x_0}) \subset U_j$
	and $\|h^j_{x_0}\|_{C^m} < \frac{C}{2(\ell+1)}$ such that $g_{x_0} + h^j_{x_0}$
	has no local conformal Killing vector around $p_j$. By the
	$C^K$-stability in \cref{propLocalDefNCK} and the continuity of $x \mapsto
		g_x$, the same holds for $g_x + h^j_{x_0}$ for all $x$ in an open neighborhood
	$V^j_{x_0} \subset \Omega \setminus A$ of $x_0$. We now define $V_{x_0} \coloneqq
		\bigcap_{j=0}^\ell V^j_{x_0}$.

	By construction, $(V_{x_0})_{x_0 \in \Omega}$ defines an open cover of $\Omega$. Since
	every open subset $V$ of $\Omega$ may be written as $V = \tilde{V} \cap \Omega$ for an
	open subset $\tilde{V}$ of $S^\ell$, this open cover gives rise to an open
	cover $(\tilde{V}_{x_0})_{x_0 \in \Omega \dotcup \pt}$ of $S^\ell$, where
	$\tilde{V}_\ast = S^\ell \setminus \Omega$. For this open cover, we obtain $\ell +
		1$ partitions of unity as in \cref{lemGoodPartition}. For every $i$ and $j$, we
	choose some $x_{ij} \in \Omega \dotcup \pt$ such that $\supp(\eta^j_i) \subset
		\tilde{V}_{x_{ij}}$. Further, if $x_{ij} \in \Omega \setminus A$, we set
	$h_{ij} \coloneqq h^j_{x_{ij}}$. If $x_{ij} \in A \dotcup
		\pt$, we set $h_{ij} = 0$. We now define for $x \in \Omega$ and $t \in [0,1]$
	\begin{gather*}
		H(x,t) \coloneqq g_x + t \sum_{j = 0}^\ell \sum_{i = 1}^{N_j} \eta_i^j(x)
		h_{ij}
	\end{gather*}
	and claim that this satisfies all the required properties.

	First of all, it is clear that $H(x,0) = g_x$ for all $x \in \Omega$. Secondly, let
	$a \in A$ and consider any $i, j$ such that $\eta_i^j(a) \neq 0$. Then
	$\supp(\eta_i^j) \cap \Omega \not\subset \Omega \setminus A$ and thus $x_{ij}$ cannot lie
	in $\Omega \setminus A$: by construction, $\tilde{V}_{x_0} \cap \Omega = V_{x_0} \subset
		\Omega \setminus A$ for all $x_0 \in \Omega \setminus A$. This implies that $h_{ij} = 0$
	and thus $H(a,t) = g_a$ for all $t \in [0,1]$, since $i, j$ were arbitrary with
	$\eta_i^j(a) \neq 0$. Third, let $x \in \Omega$ and $i,j$ such that $\eta^j_i(x) =
		1$. Note that such indices exist, because $((\eta_i^j)^{-1}(1))_{0 \leq j \leq
				\ell, 1 \leq i \leq N_j}$ covers $S^\ell$. Since $x \in \supp(\eta^j_i)$ and
	$\tilde{V}_\ast \cap \Omega = \emptyset$, we have $x_{ij} \neq \ast$ and $x \in
		\tilde{V}_{x_{ij}} \cap \Omega = V_{x_{ij}}$. Moreover, $\eta^j_{i^\prime}(x) = 0$
	for all $i^\prime \neq i$, since the $\eta^j_{i^\prime}$ form a partition of
	unity. If $x_{ij} \in A$, then $h_{ij} = 0$
	and we use the estimate
	\begin{align*}
		\|H(x,1) - g_{x_{ij}}\|_{C^m} & \leq \|g_x - g_{x_{ij}}\|_{C^m} + \|H(x,1) - g_x\|_{C^m}      \\
		                              & \leq \|g_x - g_{x_{ij}}\|_{C^m}+ \sum_{\substack{j^\prime = 0 \\ j^\prime \neq j}}^{\ell} \sum_{i^\prime=1}^{N_{j^\prime}} \eta_{i^\prime}^{j^\prime}(x) \|h_{i^\prime j^\prime}\|_{C^m} \\
		                              & < \frac{C}{2} + \ell \frac{C}{2(\ell+1)} \leq C
	\end{align*}
	to conclude that $H(x,1) \in \NCK(M)$. If $x_{ij} \in \Omega \setminus A$, we use
	the general observation that $H(x,1)_{|U_j} = (g_x + h_{ij})_{|U_j}$. This is
	because $\eta^j_{i^\prime}(x) = 0$ for $i^\prime \neq i$, $\supp(h_{i^\prime
				j^\prime}) \subset U_{j^\prime}$ for all $i^\prime,
		j^\prime$ by construction and $U_j \cap U_{j^\prime} = \emptyset$ if $j^\prime
		\neq j$. But $x \in V_{x_{ij}} \subset V^j_{x_{ij}}$, so $g_x + h_{ij}$ has no local
	CKVs around $p_j$, and this implies that $H(x,1) \in \NCK(M)$ as well. Thus we obtain $H(x,1) \in \NCK(M)$ for all
	$x \in \Omega$. Finally, we have that
	\begin{gather*}
		\|H(x,t) - g_x\|_{C^m} \leq t \sum_{j = 0}^\ell \sum_{i=1}^{N_j} \eta_i^j(x) \|h_{ij}\|_{C^m} \leq (\ell+1) \frac{C}{2(\ell+1)} < C
	\end{gather*}
	for all $x \in \Omega$, so $H(-,t) \in \mathcal{U}$ for all $t \in [0,1]$.
\end{proof}

\section{The bundle of TT-tensors}\label{secTTtensors}
We recall that, given a Riemannian manifold $(M, g)$, a TT-tensor $\sigma$ is a
symmetric traceless 2-tensor such that $\divg^g \sigma = 0$. These tensors have
become of great interest in the study of initial data in general relativity
because of the York decomposition of symmetric traceless 2-tensors that we
recall in \cref{propSplitting} below. The main goal of this section is to
choose, for every Riemannian metric $g$ without conformal Killing vectors, a
non-zero TT-tensor $\sigma(g)$ with respect to $g$, depending continuously on
$g$.

\begin{proposition}\label{propSplitting}
	Let $(M, g)$ be a closed Riemannian manifold. Then every smooth symmetric
	traceless 2-tensor $T$ admits a decomposition
	\[
		T = \sigma + \bL^g W,
	\]
	where $W$ is a smooth vector field, $\sigma$ is a smooth TT-tensor and the
	conformal Killing operator $\bL$ is defined as in~\eqref{eqConfKillingOp}. The
	summands $\sigma$ and $\bL^g W$ are unique, and $W$ is unique up to the
	addition of a conformal Killing vector. Further, this decomposition is
	$L^2$-orthogonal in the following sense: for any vector field $X$ and any
	TT-tensor $\sigma$,
	\begin{equation}\label{eqOrthogonality}
		\int_M \< \sigma, \bL^g X\>_g \upd\mu^g = 0.
	\end{equation}
\end{proposition}

As the construction of the splitting will play an important role in what
follows, we give an explicit proof of it:

\begin{proof}
	Let $\Hring^1(M, TM)$ denote the set of vector fields $X \in H^1(M, TM)$
	orthogonal to all the conformal Killing vectors:
	\[
		\Hring^1(M, TM) \coloneqq \left\{X \in H^1(M, TM)\;\middle|\; \forall Y \in \mathfrak{c}(M, g)\colon \int_M \left\<X, Y\right\>_g \upd\mu^g = 0\right\},
	\]
	where $\mathfrak{c}(M, g)$ denotes the subspace of conformal Killing vectors of
	$(M, g)$:
	\[
		\mathfrak{c}(M, g) = \left\{Y \in H^1(M, TM)\;\middle|\; \bL^g Y \equiv 0\right\}.
	\]
	We write $V \coloneqq \Hring^1(M, TM)$ for short, and we first claim that the
	coercivity constant of $\bL^g$ over $V$,
	\begin{equation}\label{eqDefMuV}
		\mu_V(g) \coloneqq \inf_{\substack{X \in V\\ X \not\equiv 0}} \frac{\int_M |\bL^g X|_g^2 \upd\mu^g}{\|X\|_{H^1}^2},
	\end{equation}
	is strictly positive. The argument goes by contradiction. Assume that $\mu_V(g)
		= 0$. Then there exists a sequence $(X_k)_{k \geq 1}$ of elements of
	$\Hring^1(M, TM)$ such that $\|X_k\|_{H^1} = 1$ for all $k \geq 1$ and such
	that
	\[
		\int_M \left|\bL^g X_k\right|_g^2 \upd\mu^g \to_{k \to \infty} 0.
	\]
	As the embedding $H^1(M, TM) \hookrightarrow L^2(M, TM)$ is compact, we can
	assume without loss of generality that the sequence $(X_k)_{k \geq 1}$
	converges to some $X_\infty$ for the $L^2$-norm.

	We now use the Bochner formula for the conformal Killing operator (see
	\eg\cite[Ch.~2, \S~6]{Yano}). For any $k, \ell \geq 1$, we have
	\[
		\begin{aligned}
			 & \frac{1}{2}\int_M \left|\bL^g (X_k - X_\ell)\right|_g^2 \upd\mu^g                                                                                            \\
			 & \qquad = \int_M \left[\left|\nabla (X_k - X_\ell)\right|_g^2 + \frac{n-2}{n} (\divg^g (X_k - X_\ell))^2 - \ric(X_k - X_\ell, X_k - X_\ell)\right] \upd\mu^g.
		\end{aligned}
	\]
	In particular, we get
	\begin{align*}
		\int_M \left|\nabla (X_k - X_\ell)\right|_g^2 \upd\mu^g
		 & \leq \int_M \left|\bL^g (X_k - X_\ell)\right|_g^2 \upd\mu^g + \int_M \ric(X_k - X_\ell, X_k - X_\ell) \upd\mu^g  \\
		 & \leq 2 \left(\int_M \left|\bL^g X_k\right|_g^2 \upd\mu^g + \int_M \left|\bL^g X_\ell\right|_g^2 \upd\mu^g\right) \\
		 & \qquad + \|\ric\|_{L^\infty} \|X_k - X_\ell\|_{L^2}^2.
	\end{align*}
	From this last inequality, we see that the sequence $(X_k)_{k \geq 1}$ is
	Cauchy in $H^1(M, TM)$. By the uniqueness of the limit of a sequence in $L^2(M,
		TM)$, we get that $X_\infty \in H^1(M, TM)$ and $X_k \to_{k \to \infty}
		X_\infty$ for the $H^1$-norm. In particular, $\|X_\infty\|_{H^1} = \lim_{k \to
			\infty} \|X_k\|_{H^1} = 1$ so $X_\infty \not\equiv 0$. The quadratic form
	\begin{equation}\label{eqQuadraticForm}
		X \mapsto \int_M \left|\bL^g X\right|_g^2 \upd\mu^g
	\end{equation}
	being continuous over $H^1(M, TM)$, we have that
	\[
		\int_M \left|\bL^g X_\infty\right|_g^2 \upd\mu^g = \lim_{k \to \infty} \int_M \left|\bL^g X_k\right|_g^2 \upd\mu^g = 0.
	\]
	As a consequence, $\bL^g X_\infty \equiv 0$ a.\,e., meaning that $X_\infty$ is
	a conformal Killing vector: $X_\infty \in \mathfrak{c}(M, g)$. But the subspace
	$\Hring^1(M, TM)$ is closed so $X_\infty \in \Hring^1(M, TM)$ as it is the
	limit of elements in $\Hring^1(M, TM)$. Hence,
	\[
		X_\infty \in \mathfrak{c}(M, g) \cap \Hring^1(M, TM) = \{0\}.
	\]
	This is the desired contradiction as $\|X_\infty\|_{H^1} = 1$. In particular,
	we have shown that the quadratic form given by~\eqref{eqQuadraticForm} is
	coercive over $\Hring^1(M, TM)$.

	The existence of the York decomposition is then obtained by the Lax-Milgram
	theorem: given $T \in L^2(M, \Sring_2 M)$, the linear form
	\[
		X \mapsto \int_M \left\<T, \bL^g X\right\>_g \upd\mu^g
	\]
	is continuous over $\Hring^1(M, TM)$ so there exists a unique $W \in
		\Hring^1(M, TM)$ such that, for all $X \in \Hring^1(M, TM)$, we have
	\[
		\int_M \left\<\bL^g X, \bL^g W\right\>_g \upd\mu^g = \int_M \left\<\bL^g X, T\right\>_g \upd\mu^g.
	\]
	In particular,
	\[
		\forall X \in \Hring^1(M, TM)\colon \int_M \left\<\bL^g X, T - \bL^g W\right\>_g \upd\mu^g = 0.
	\]
	Note that the previous equality is trivially true if $X \in \mathfrak{c}(M, g)$
	as, in this case, we have that $\bL^g X = 0$. So, by linearity, we can extend
	it to $X \in H^1(M, TM) = \mathfrak{c}(M, g) \oplus \Hring^1(M, TM)$. Hence,
	for any vector field $X \in H^1(M, TM)$, we have
	\[
		\int_M \left\<\bL^g X, T - \bL^g W\right\>_g \upd\mu^g = 0.
	\]
	Thus, integrating by parts, we see that, in the weak sense, $\divg^g(T - \bL^g
		W) = 0$. By elliptic regularity, we get that $W$ is a smooth vector field.
	Hence, $\sigma \coloneqq T - \bL^g W$ is divergence-free and traceless, \ie a
	TT-tensor. Finally, if $T = \sigma' + \bL^g W'$ is another such decomposition,
	then $\sigma - \sigma' = \bL^g(W' - W)$ is a TT-tensor that is $L^2$-orthogonal
	to itself by~\eqref{eqOrthogonality}. Hence $\sigma = \sigma'$, and $W' - W$ is
	a conformal Killing vector.
\end{proof}

The constant $\mu_V(g)$ defined by~\eqref{eqDefMuV} in the proof of
\cref{propSplitting} is not the same as the Sobolev-like one that is usually
defined in this context (see \eg \cite{DahlGicquaudHumbert}), but it is more
suited for our needs.

Note that the orthogonality condition~\eqref{eqOrthogonality} gives a weak
notion of a TT-tensor, which was also used in the course of the proof of
\cref{propSplitting}:
\begin{definition}\label{defWeakTT}
	Let $(M, g)$ be a given closed Riemannian manifold. A symmetric 2-tensor
	$\sigma \in L^2(M, S_2 M)$ is called \emph{(weakly) transverse and traceless},
	\ie a \emph{(weak) TT-tensor}, if $\tr^g(\sigma) \equiv 0$ and, for any vector
	field $X \in \Gamma(TM)$, we have
	\[
		\int_M \left\< \sigma, \bL^g X\right\>_g \upd\mu^g = 0.
	\]
\end{definition}

As the proof of \cref{propSplitting} indicates, metrics admitting non-zero
conformal Killing vectors deserve particular care for the construction of
York's splitting. However, it has been proven in the previous section that
these metrics are so rare that removing them has no influence on the homotopy
groups of the spaces of metrics and of CMC initial data sets that we consider
(\cf \cref{corReductionToNCK}). Recall that $\NCK(M)$ denotes the set of
metrics having no non-zero conformal Killing vector.

\begin{definition}\label{defTTFiber}
	Let $M$ be a closed manifold. For any metric $g \in \Met(M)$, we let $\TT(M,
		g)$ denote the set of (weak) TT-tensors of $(M, g)$:
	\[
		\TT(M,g) = \left\{\sigma \in L^2(M, S_2 M) \;\middle|\; \tr^g(\sigma) \equiv 0 ~\text{and}~ \forall X \in \Gamma(TM) \colon \int_M \left\<\sigma, \bL^g X\right\>_g \upd\mu^g = 0\right\}.
	\]
\end{definition}

The continuous assignment $\NCK(M) \ni g \mapsto \sigma(g)$ we are looking for
may be viewed as a non-vanishing section of the following bundle.
\begin{proposition}\label{propTTBundle}
	Let $M$ be a closed manifold of dimension $n \geq 3$. The space $\TT(M)$
	defined by
	\[
		\TT(M) \coloneqq \coprod_{g \in \NCK(M)} \TT(M, g)
	\]
	is a vector subbundle of the trivial Hilbert bundle $\NCK(M) \times L^2(M,
		S_2M)$ over $\NCK(M)$.
\end{proposition}

Before we can prove this proposition, we have to make clear what we mean by a
Hilbert bundle and establish some facts about their subbundles. Note that some
care is needed here due to the fact that the bundle of TT-tensors is infinite
dimensional and several topologies exist on the set $\cL(H)$ of continuous
linear maps if $H$ is an infinite-dimensional Hilbert space, see
\eg\cite[Ch.~IX, \S~5]{Conway}. The definition we take here is the strongest
possible one:

\begin{definition}\label{defHilbertBundle}
	Let $E$ and $B$ be topological spaces and $\pi_E \colon E \to B$ be a
	continuous map, let also $H$ be a (real) Hilbert space. We say that the triple
	$(E, B, \pi_E)$ is a \emph{Hilbert bundle with typical fiber $H$} if
	\begin{enumerate}
		\item For each $x \in B$, the set $E_x \coloneqq \pi_E^{-1}(x)$ is endowed with the
		      structure of a Hilbert space.
		\item There exists an open covering $(U_i)_{i \in I}$ of $B$ together with
		      homeomorphisms $\phi_i \colon \pi_E^{-1}(U_i) \to U_i \times H$ such that the
		      first component of $\phi_i$ is equal to $\pi_E$ and for each $x \in U_i$, the
		      second component of $\phi_i$ induces an isometric linear isomorphism from $E_x$
		      to $H$.
		\item For each $i, j \in I$, the transition functions $\phi_i \circ \phi_j^{-1}$
		      mapping $(U_i \cap U_j) \times H$ onto itself can be written as
		      \[
			      \phi_i \circ \phi_j^{-1} (x, v) = (x, g_{ij}(x, v))
		      \]
		      where, for each $x \in U_i \cap U_j$, $g_{ij}(x, \cdot) \in \mathrm{O}(H)$
		      where $\mathrm{O}(H) \subset \cL(H)$ denotes the orthogonal group with respect
		      to the scalar product $\<\cdot, \cdot\>$ and the mapping $x \mapsto g_{ij}(x,
			      \cdot)$ is continuous for the operator norm topology.
	\end{enumerate}
\end{definition}

We require the transition functions to be continuous for the operator norm
topology. The strong operator topology would be the most natural one if one
only required the trivializing functions $\phi_i \colon E\vert_{U_i} \to U_i
	\times H$ to be continuous. However, the proof of \cref{lemSubbundle2} below
uses the fact that the invertible operators form an open subset of $\cL(H)$,
and this fails for the strong operator topology. For instance, the orthogonal
projections onto an increasing sequence of finite-dimensional subspaces with
dense union converge strongly to the identity, and none of them is invertible.
A good replacement of a bundle notion for the strong operator topology is given
by Dixmier and Douady's \emph{continuous field of Hilbert spaces}
\cite{DixmierDouady}, which would probably have been sufficient to arrive at
the ultimate conclusion \cref{corTT}. However, we decided to stick with the
simpler and less technical notion of Hilbert bundle above, giving some
preliminary results as we were unable to find a good reference that suits our
needs.

The first result allows us to deal with the fact that the bundle of symmetric
$2$-tensors of regularity $L^2$ is naturally endowed with a scalar product
depending on the base metric $g \in \NCK(M)$:

\begin{lemma}\label{lemUniformisation}
	Let $(H, \<\cdot, \cdot\>_0)$ be a fixed real Hilbert space, and let $B$ be a
	topological space. Consider the product $E = B \times H$ endowed with its
	canonical projection $\pi_E \colon E \to B$. Suppose that for each $x \in B$,
	the fiber $\{x\} \times H$ is equipped with a scalar product
	$\langle\cdot,\cdot\rangle_x$ making it a Hilbert space with a norm equivalent
	to that of $\<\cdot, \cdot\>_0$, and that this scalar product depends
	continuously on the base point in the following sense:

	For every $x_0 \in B$ and every $\varepsilon > 0$, there exists a neighborhood
	$U$ of $x_0$ such that
	\[
		\forall x \in U,\quad \left\|\langle \cdot,\cdot \rangle_x - \langle \cdot,\cdot \rangle_{x_0}\right\|_{x_0} < \varepsilon,
	\]
	where $\|\cdot\|_{x_0}$ denotes the operator norm on bilinear forms defined by
	\[
		\|b\|_{x_0} \coloneqq \sup_{\substack{u,v \in H\\ \|u\|_{x_0} = \|v\|_{x_0} = 1}} |b(u,v)|.
	\]

	Then the bundle $(E, B, \pi_E)$ is isometrically equivalent to the trivial
	bundle $B \times H$: there exists a bundle isomorphism $\Phi \colon E \to B
		\times H, e \mapsto (\pi_E(e),\Phi_{\pi_E(e)}(e))$ such that $\Phi$ induces an
	isometry on each of the fibers.
\end{lemma}

\begin{proof}
	For each $x \in B$, the scalar product $\langle\cdot, \cdot\rangle_x$ defines a
	bounded, coercive, symmetric bilinear form on $(H, \<\cdot, \cdot\>_0)$, and
	hence induces a bounded, self-adjoint, positive-definite operator $A_x \in
		\mathcal{L}(H)$ via the Riesz representation theorem:
	\[
		\langle A_x u, v \rangle_0 = \langle u, v \rangle_x \quad \text{for all } u, v \in H.
	\]
	Note that $A_x$ is self-adjoint and strictly positive: there exist constants $0
		< m \leq M < \infty$ such that
	\[
		m \|u\|_0^2 \leq \langle A_x u, u \rangle_0 \leq M \|u\|_0^2.
	\]

	Now, define $A_x^{1/2}$, the positive square root of $A_x$. Since
	$A_x$ is self-adjoint and positive-definite, $A_x^{1/2}$ exists, is bounded and
	invertible, and self-adjoint (see~\cite[Ch.~VIII, Thm.~3.5]{Conway}).
	Furthermore, the map $x \mapsto A_x$ is continuous in operator norm, as a
	direct consequence of the continuity assumption on the scalar products, and the
	map
	\[
		A \mapsto A^{1/2}
	\]
	is continuous in the operator norm topology on the open cone of self-adjoint
	strictly positive operators, see \eg \cite{Bouldin}. It follows that $x \mapsto
		A_x^{1/2}$ is continuous in the operator norm topology. Each $A_x^{1/2}$ is a bijective
	isometry from $(H, \langle\cdot, \cdot\rangle_x)$ to $(H, \langle\cdot,
		\cdot\rangle_0)$.

	We now define the map
	\[
		\Phi \colon (x, v) \mapsto (x, A_x^{1/2} v),
	\]
	which is a homeomorphism from $B \times H$ (endowed with the topology where the
	norm on each fiber is $\|\cdot\|_x$) to $B \times H$ with the fixed Hilbert
	structure $\langle\cdot,\cdot\rangle_0$, and which is a fiberwise isometry by
	construction.
\end{proof}

We will have to deal with subbundles.
\begin{definition}
	Let $\pi_E \colon E \to B$ be a Hilbert bundle with typical fiber $H$. A subset
	$E' \subset E$ is a \emph{Hilbert subbundle} \label{defHilbertSubbundle} if for
	each $x \in B$ the fiber $E'_x \subset E_x$ is a closed subspace and if $E'$
	equipped with the induced fiberwise scalar product forms a Hilbert bundle.
\end{definition}

\begin{corollary}\label{corSubbundle1}
	Let $E$ and $F$ be two Hilbert bundles in the sense of \cref{defHilbertBundle}
	over the same base $B$. Assume that there exists a map $\iota \colon E \to F$
	that is linear on each fiber and such that, for any $x_0 \in B$, there exists a
	neighborhood $U$ of $x_0$ and constants $c, C > 0$ such that the following
	holds:

	\begin{enumerate}
		\item For all $x \in U$, the linear map $\iota_x \colon E_x \to F_x$ satisfies
		      \[
			      \forall v \in E_x,\ c \|v\|_{E_x} \leq \|\iota_x(v)\|_{F_x} \leq C \|v\|_{E_x}.
		      \]
		\item \label{itm2LmSubbundle1} Over any arbitrary trivialization $\pi_E^{-1}(U) \overset{\phi}{\cong} U
			      \times H_1$ and $\pi_F^{-1}(U) \overset{\psi}{\cong} U \times H_2$ of $E$ and
		      $F$, the family of mappings $x \mapsto \tilde{\iota}_x \colon H_1 \to H_2$ induced by $\iota$ is
		      continuous for the operator norm.
	\end{enumerate}

	Then the image of $\iota$, $\iota(E)$, is a subbundle of $F$.
\end{corollary}
\begin{proof}
	First of all, we note that for each $x \in B$ the lower bound implies that the
	map $\iota_x$ is injective and also that its image is complete and thus closed
	in $F_x$. Here, for the completeness just observe that a Cauchy sequence
	$(\iota_x(v_k))_{k \geq 0}$ in $\iota_x(E_x)$ has to come from a Cauchy
	sequence $(v_k)_{k \geq 0}$ in $E_x$ since $\|v_k - v_\ell\|_{E_x} \leq
		c^{-1}\|\iota_x(v_k) - \iota_x(v_\ell)\|_{F_x}$.

	It remains to check that $\iota(E)$ with the induced fiberwise scalar product
	forms a Hilbert bundle. Since $E$ is a Hilbert bundle, we may choose an open
	covering $(U_i)_{i \in I}$ and trivializations $\phi_i \colon \pi_E^{-1}(U_i)
		\to U_i \times H$ for $E$, where $(H,\langle\cdot,\cdot\rangle_0) \coloneqq
		(H_1,\langle\cdot,\cdot\rangle_{H_1})$ is its typical fiber. By virtue of the
	linear isomorphism $j_{i,x} \colon H \to E_x \to \iota_x(E_x),\,v \mapsto
		\iota_x(\phi_i^{-1}(x,v))$ we obtain a unique scalar product
	$\langle\cdot,\cdot\rangle_x$ on $H$ for which the isomorphism is an isometry.
	This scalar product depends continuously on $x$ as it is given by $\langle
		v,w\rangle_x = \langle \tilde{\iota}_x(v),\tilde{\iota}_x(w)\rangle_{H_2}$ in
	the notation of \cref{itm2LmSubbundle1}. The given bounds for $\iota_x$ show
	that $\langle\cdot,\cdot\rangle_x$ is norm equivalent to
	$\langle\cdot,\cdot\rangle_0$ for all $x \in U_i$.

	Thus we can apply \cref{lemUniformisation} and obtain an isometric bundle
	isomorphism $\Phi_i \colon U_i \times H \to U_i \times H, (x,v) \mapsto
		(x,\Phi_{i,x}(v))$, where the left hand side is equipped with the fiberwise
	scalar product $\langle\cdot,\cdot\rangle_x$ and the right hand side with
	$\langle\cdot,\cdot\rangle_0$. We can then define the trivializations for
	$\iota(E)$ by $\psi_i(f) \coloneqq \Phi_i (x, j_{i,x}^{-1}(f)) = (x, \Phi_{i,x}
		\circ j_{i,x}^{-1}(f))$ where $x = \pi_F(f)$. By definition, these are
	fiberwise isometric isomorphisms. Hence the second component of the transition
	functions $\psi_i \circ \psi_j^{-1}(x,v) = (x, \Phi_{i,x} \circ j_{i,x}^{-1}
		\circ j_{j,x} \circ \Phi_{j,x}^{-1}(v)) = (x, \Phi_{i,x} \circ g_{ij}(x,\cdot)
		\circ \Phi_{j,x}^{-1}(v))$ has to be orthogonal fiberwise and the formula shows
	that it continuously depends on $x$.
\end{proof}

As for finite-dimensional vector bundles (\cf \cite[Prop.~1.3]{HatcherKT}),
Hilbert subbundles have complements.

\begin{lemma}\label{lemSubbundle2}
	Let $E \to B$ be a Hilbert bundle with typical fiber $H$ over a connected base
	$B$ and $E' \subset E$ a Hilbert subbundle with typical fiber $H'$. Then
	$(E')^\perp \coloneqq \coprod_{x \in B} (E'_x)^\perp$ is a subbundle of $E$.
\end{lemma}
\begin{remark}
	Note that we need to assume here that the base $B$ is connected. This is to
	avoid trivial situations where the codimension of $E'_x$ changes with $x$. As
	an example, take the two-point space $B = \{\pm 1\}$ and let $H$ be an
	infinite-dimensional separable Hilbert space so that $H \cong H \oplus H$ (let
	$(e_0, e_1, e_2, \ldots)$ denote a Hilbert basis of $H$, then $H = H_0 \oplus
		H_1$, where $H_0$ denotes the closed subspace generated by $e_0, e_2, \ldots$
	and $H_1$ the one generated by $e_1, e_3, \ldots$, so both $H_0$ and $H_1$ are
	isometrically isomorphic to $H$). Let $E = \{\pm 1\} \times H \cong (\{-1\}
		\times H) \cup (\{1\} \times (H \oplus H))$ and $E' = \{\pm 1\} \times H$,
	where $E'$ is equal to $E$ over $\{-1\}$ and is given by one of the summands
	$H$ over $\{1\}$. Then $E' \subset E$ satisfies all the other assumptions of
	the lemma, but $(E')^\perp \cong (\{-1\} \times \{0\}) \cup (\{1\} \times H)$
	is not a Hilbert subbundle in the sense of \cref{defHilbertSubbundle}.
\end{remark}
\begin{proof}[Proof of \cref{lemSubbundle2}]
	Let $x_0 \in B$ be an arbitrary point.
	Let $U$ be an open neighborhood of $x_0$.
	Without loss of generality, we can assume that $U$ is so small that both $E$ and $E'$
	are trivial over $U$.
	With respect to these trivializations, the inclusion $E' \subset E$ is given by a
	family $(\iota_x)_{x \in U}$ of isometric embeddings $H' \hookrightarrow H$, continuous
	for the operator norm.

	Now note that $\iota_x^\star \circ \iota_x = \id$ since $\langle \iota_x^\star
		\circ \iota_x(v),w \rangle_{H'} = \langle \iota_x(v),\iota_x(w) \rangle_{H} =
		\langle v,w \rangle_{H'}$. From this it is not hard to see that $p_x \coloneqq
		\iota_x \circ \iota_x^\star \colon H \to H$ defines the orthogonal projection
	on the image of $\iota_x$. The map $x \mapsto p_x$ is continuous for the
	operator norm.

	The idea is to introduce the family of maps $f_x \colon H \to H$ defined as
	\[
		f_x \coloneqq p_0 \circ p_x + (\id - p_0) \circ (\id - p_x),
	\]
	where $p_0$ is a shorthand for $p_{x_0}$. Note that $f_x$ maps the image of
	$\iota_x$ into the image of $\iota_{x_0}$ and that $f_{x_0} = \id$. Since the
	set of invertible linear maps is an open subset of $\cL(H)$ when endowed with
	the operator norm, and the mapping $x \mapsto f_x$ is continuous for the
	operator norm, we can shrink $U$ to ensure that $f_x$ is invertible for all $x
		\in U$.

	Note that $f_x^\star \circ f_x$ is given by
	\[
		f_x^\star \circ f_x = p_x \circ p_0 \circ p_x + (\id - p_x) \circ (\id - p_0) \circ (\id - p_x)
	\]
	so it maps each $\Im(\iota_x)$ into itself. By standard results in spectral
	theory, $(f_x^\star \circ f_x)^{-1/2}$ also maps $\Im(\iota_x)$ into itself.

	As a consequence, the operator $u_x \coloneqq f_x \circ (f_x^\star \circ
		f_x)^{-1/2}$ is an orthogonal operator of $H$ mapping $\Im(\iota_x)$ onto
	$\Im(\iota_{x_0})$. Since $x \mapsto u_x$ is continuous for the operator norm,
	the map $u: U \times H \to U \times H$ defined by
	\[
		u(x, v) = (x, u_x(v)),
	\]
	provides us with a new trivialization $\Psi$ of $E$ over $U$ with the property
	that $E'$ gets mapped to $U \times H'_\Psi$ with $H'_\Psi \coloneqq
		\iota_{x_0}(H')$. For the purpose of this proof, we call trivializations with
	this property \emph{adapted trivializations}. This procedure can be repeated
	for other base points, so $B$ can be covered by adapted trivializations.

	We finally use the connectivity of $B$ to prove that the $H'_\Psi$ always ``sit
	in the same way'' inside $H$, \ie that given any adapted trivialization $\Psi$
	defined over an open set $U$ there is some $S_\Psi \in \mathrm{O}(H)$ such that
	$H'_0 = S_\Psi(H_\Psi')$ for the fixed subspace $H'_0 \coloneqq \iota_{x_0}(H')
		\subset H$. Then all the adapted trivializations can be modified so that $E'$
	gets mapped to $U \times H'_0$ and so they restrict to trivializations
	$(E')^\perp_{|U} \cong U \times (H'_0)^\perp$, showing that $(E')^\perp$ is a
	Hilbert bundle.

	So first of all, we note that, given a point $x \in B$, if there exists such an
	orthogonal operator $S_\Psi$ for one of the adapted trivializations $\Psi
		\colon E_{|U} \cong U \times H$ with $x \in U$, then it is also the case for
	any other one $\Phi \colon E_{|V} \cong V \times H$ with $x \in V$. Simply use
	$S_\Psi \circ g_x^{-1}$, where $\Phi \circ \Psi^{-1}(x,v) = (x,g_x(v))$. The
	subset of all the points where there are such orthogonal operators for one (and
	thus all) adapted trivializations will be called $\Omega$.

	By construction, $x_0 \in \Omega$, so $\Omega$ is non-empty. Moreover, $\Omega$
	is open, since if $x \in \Omega$ and $\Psi$ is an adapted trivialization
	defined on $U \ni x$, then existence of $S_\Psi$ shows that $U \subset \Omega$.
	But $\Omega$ is also closed, for if $x \not\in \Omega$, then there is an
	adapted trivialization $\Psi$ defined on $U \ni x$ for which an orthogonal
	operator with the required properties does not exist and then $U \cap \Omega =
		\emptyset$. Since $B$ is connected, it follows that $\Omega = B$, which is what
	we had to show.
\end{proof}

As a consequence of \cref{lemSubbundle2}, we obtain the following result.

\begin{corollary}\label{corSubbundle}
	Let $\pi_E \colon E \to B$ be a Hilbert bundle with typical fiber $H$ over a
	connected base $B$ and $E' \subset E$ a Hilbert subbundle with typical fiber
	$H'$. Then $H'$ can be identified with a closed subspace of $H$ such that the
	following hold:
	\begin{enumerate}
		\item There exists an open covering $(U_i)_{i \in I}$ of $B$ and local
		      trivializations $\phi_i \colon \pi_E^{-1}(U_i) \to U_i \times H$ such that
		      $\phi_i(E' \cap \pi_E^{-1}(U_i)) = U_i \times H'$.
		\item For each $i, j \in I$, the transition functions $\phi_i \circ \phi_j^{-1}$
		      mapping $(U_i \cap U_j) \times H$ onto itself can be written as
		      \[
			      \phi_i \circ \phi_j^{-1} (x, v) = (x, g_{ij}(x, v)),
		      \]
		      where, for each $x \in U_i \cap U_j$, $g_{ij}(x, \cdot) \in \mathrm{O}(H')
			      \times \mathrm{O}((H')^\perp) \subset \mathrm{O}(H)$ and the mapping $x \mapsto
			      g_{ij}(x, \cdot)$ is continuous for the operator norm topology.
	\end{enumerate}
\end{corollary}
\begin{proof}
	By \cref{lemSubbundle2}, $E'' \coloneqq (E')^\perp$ is a Hilbert subbundle with
	typical fiber $H''$. We may now choose trivializations $\phi'_i \colon
		\pi_{E'}^{-1}(U_i) \to U_i \times H'$ and $\phi''_i \colon \pi_{E''}^{-1}(U_i)
		\to U_i \times H''$, where we assumed without loss of generality that they are
	defined on the same open covering $(U_i)_{i \in I}$ of $B$. Then, using $E = E'
		\oplus E''$, we define $\phi_i \colon \pi_{E}^{-1}(U_i) \to U_i \times (H'
		\oplus H'')$ by adding up $\phi'_i$ and $\phi''_i$. For any point $x$, we have
	$H' \oplus H'' \cong E'_x \oplus E''_x = E_x \cong H$ and we may fix such an
	isomorphism to identify $H'$ with a closed subspace of $H$. Now the rest of the
	claim is readily checked.
\end{proof}

Before applying these results with base $B = \NCK(M)$, we record the two
topological properties of $\NCK(M)$ that they require. First, $\NCK(M)$ is
connected. Indeed, by \cref{corReductionToNCK} the inclusion $\NCK(M)
	\hookrightarrow \Met(M)$ is a weak homotopy equivalence, in particular a
bijection on path components, and $\Met(M)$ is convex. Second, $\NCK(M)$ is
paracompact. Indeed, it is a subspace of the space of smooth sections of $S_2M$
with its $C^\infty$-topology, which is a Fréchet space, hence metrizable, so
$\NCK(M)$ is metrizable, and every metrizable space is paracompact by Stone's
theorem, \cf \eg \cite[Thm.~41.4]{Munkres}.

We can now prove \cref{propTTBundle}. Note that, from now on, we will indicate
the metric used to define the functional norms. For example, given a metric $g
	\in \NCK(M)$, the norm $\|\cdot\|_{L^2, g}$ is defined as
\[
	\|v\|_{L^2, g}^2 \coloneqq \int_M |v|_g^2 \upd\mu^g.
\]

\begin{proof}[Proof of \cref{propTTBundle}]
	Due to \cref{lemUniformisation}, we may equally well endow each fiber of
	$\NCK(M) \times L^2(M, S_2M)$ with the scalar product inherited from that of
	the underlying metric $g$, which we will do in the following. We first note
	that the subbundle $\Sring_2(M)$ whose fiber at any $g \in \NCK(M)$ consists of
	(symmetric) tracefree 2-tensors is a subbundle of $\NCK(M) \times L^2(M,
		S_2M)$. This follows from \cref{lemSubbundle2} as it is the orthogonal
	complement of the subbundle of pure trace 2-tensors. To see in turn that these
	form a subbundle, we note that they form the image of the map
	\[
		\begin{array}{rccc}
			\iota_0 \colon & \NCK(M) \times L^2(M, \bR) & \to     & \NCK(M) \times L^2(M, S_2M) \\
			               & (g, f)                     & \mapsto & f\,g,
		\end{array}
	\]
	and that $\iota_0$ trivially fulfills the assumptions of \cref{corSubbundle1}.

	We now address the more complicated question of proving that the map $\bL
		\colon \NCK(M) \times H^1(M, TM)\to \Sring_2(M)$ sending a pair $(g, V)$ to
	$\bL^g V$ fulfills the assumptions of \cref{corSubbundle1}. The subbundle
	orthogonal to $\Im(\bL)$ in $\Sring_2(M)$, namely $\TT(M)$, will then itself be
	a subbundle of $\NCK(M) \times L^2(M, S_2M)$.

	We begin by proving the first assumption of \cref{corSubbundle1}. To this end,
	we use again Bochner's formula for the conformal Killing operator $\bL^g$ for
	any $V \in H^1(M, TM)$:
	\[
		\frac{1}{2}\int_M |\bL^g V|_g^2 \upd\mu^g = \int_M \left[|\nabla^g V|_g^2 + \frac{n-2}{n} (\divg^g V)^2 - \ric^g(V, V)\right] \upd\mu^g.
	\]
	As $\divg^g V = \tr^g (\nabla^g V)$, we have, from the Cauchy-Schwarz
	inequality, $|\divg^g V| \leq \sqrt{n} |\nabla^g V|_g$. As a consequence, we
	deduce that
	\begin{align*}
		\frac{1}{2}\int_M |\bL^g V|_g^2 \upd\mu^g
		 & \leq \int_M \left[|\nabla^g V|_g^2 + (n-2) |\nabla^g V|_g^2 + \|\ric^g\|_{L^\infty, g} |V|_g^2\right] \upd\mu^g \\
		 & = \int_M \left[(n-1)|\nabla^g V|_g^2 + \|\ric^g\|_{L^\infty, g} |V|_g^2\right] \upd\mu^g.
	\end{align*}

	As the mapping $g \mapsto \|\ric^g\|_{L^\infty, g} = \sup_M |\ric^g|_g$ is
	continuous due to the fact that $\NCK(M)$ is endowed with the $C^\infty$
	Fréchet topology, we immediately see that, in a neighborhood of any metric $g_0
		\in \NCK(M)$, there exists a constant $C > 0$ such that $\|\bL^g V\|_{L^2, g}
		\leq C \|V\|_{H^1, g}$ for any $V \in H^1(M,TM)$. The proof of the lower bound
	for $\|\bL^g V\|_{L^2, g}$ is more difficult.

	From the proof of \cref{propSplitting}, we know that the constant $\mu_V(g)$ is
	strictly positive for any $g \in \NCK(M)$. The existence of the constant $c$ as
	in \cref{corSubbundle1} is a consequence of the fact that the map $g \mapsto
		\mu_V(g)$ is lower semicontinuous as we prove now. Let $g_0 \in \NCK(M)$ be,
	once again, an arbitrary metric and let $(g_k)_{k \geq 1}$ be a sequence of
	metrics that converges in $C^1$-norm to $g_0$. We have to prove that
	$\liminf_{k \to \infty} \mu_V(g_k) \geq \mu_V(g_0)$.

	Extracting a subsequence $(g_{\phi(k)})_{k \geq 1}$ of the metrics $g_k$ by
	means of a strictly increasing function $\phi \colon \bN \to \bN$, we can
	assume that $\lim_{k \to \infty} \mu_V(g_{\phi(k)})$ actually exists and
	satisfies
	\[
		\lim_{k \to \infty} \mu_V(g_{\phi(k)}) = \liminf_{k \to \infty} \mu_V(g_k).
	\]
	We shall rename the sequence $(g_{\phi(k)})_{k \geq 1}$ as $(g_k)_{k \geq 1}$.
	By the definition of $\mu_V(g_k)$, we can choose a vector field $X_k$ such that
	$\|X_k\|_{H^1, g_k} = 1$ and
	\[
		\int_M \left|\bL^{g_k} X_k \right|_{g_k}^2 \upd\mu^{g_k} \leq \left(1 + \frac{1}{k}\right) \mu_V(g_k).
	\]
	So, in particular,
	\[
		\lim_{k \to \infty} \int_M \left|\bL^{g_k} X_k \right|_{g_k}^2 \upd\mu^{g_k} = \lim_{k \to \infty} \mu_V(g_k).
	\]
	Note that, even though we do not impose that the sequence $(X_k)_{k \geq 1}$
	converges, we have
	\[
		\lim_{k \to \infty} \|X_k\|_{H^1, g_0}^2 = \lim_{k \to \infty} \|X_k\|_{H^1, g_k}^2 = 1.
	\]
	Let us prove this fact as it is archetypal of the arguments used in the proof
	of the lower semicontinuity. We have, by definition,
	\[
		\|X_k\|_{H^1, g_0}^2 = \int_M \left[|\nabla^{g_0} X_k|_{g_0}^2 + |X_k|_{g_0}^2\right] \upd\mu^{g_0}.
	\]
	As the sequence $(g_k)_{k \geq 1}$ converges to $g_0$, there exist constants
	$\epsilon_k > 0$ tending to zero as $k$ goes to infinity such that, for all $k
		\geq 1$, we have
	\[
		(1 + \epsilon_k)^{-2} g_0 \leq g_k \leq (1 + \epsilon_k)^2 g_0.
	\]
	In particular, we have
	\[
		(1 + \epsilon_k)^{-n-2} \int_M |X_k|_{g_k}^2 \upd\mu^{g_k} \leq \int_M |X_k|_{g_0}^2 \upd\mu^{g_0} \leq (1 + \epsilon_k)^{n+2} \int_M |X_k|_{g_k}^2 \upd\mu^{g_k},
	\]
	and, similarly,
	\[
		(1 + \epsilon_k)^{-n-4} \int_M |\nabla^{g_0} X_k|_{g_k}^2 \upd\mu^{g_k} \leq \int_M |\nabla^{g_0} X_k|_{g_0}^2 \upd\mu^{g_0} \leq (1 + \epsilon_k)^{n+4} \int_M |\nabla^{g_0} X_k|_{g_k}^2 \upd\mu^{g_k}.
	\]
	Adding these two estimates and using the fact that $(1 + \epsilon_k)^{n+2} \leq
		(1 + \epsilon_k)^{n+4}$, we conclude that
	\begin{align*}
		\|X_k\|_{H^1, g_0}^2 & \geq (1 + \epsilon_k)^{-n-4} \int_M \left[|\nabla^{g_0} X_k|_{g_k}^2 + |X_k|_{g_k}^2\right] \upd\mu^{g_k}, \\
		\|X_k\|_{H^1, g_0}^2 & \leq (1 + \epsilon_k)^{n+4} \int_M \left[|\nabla^{g_0} X_k|_{g_k}^2 + |X_k|_{g_k}^2\right] \upd\mu^{g_k}.
	\end{align*}
	All we need to do now is to prove that
	\[
		\lim_{k \to \infty} \int_M \left[|\nabla^{g_0} X_k|_{g_k}^2 + |X_k|_{g_k}^2\right] \upd\mu^{g_k} = 1.
	\]
	The left-hand side is almost $\|X_k\|_{H^1, g_k}^2$, the only difference being
	that the connection is $\nabla^{g_0}$ instead of $\nabla^{g_k}$. Note that
	$\nabla^{g_k} X_k - \nabla^{g_0} X_k = \Gamma(g_0, g_k)(\cdot, X_k)$, where
	$\Gamma(g_0, g_k)$ is the relative Christoffel symbol defined, in abstract
	index notation, as
	\begin{align*}
		{\Gamma(g_0, g_k)_{ab}}^c
		 & \coloneqq \frac{1}{2} (g_k)^{cd}\left(\nabla^{g_0}_a (g_k)_{db} + \nabla^{g_0}_b (g_k)_{ad} - \nabla^{g_0}_d (g_k)_{ab}\right)            \\
		 & = \frac{1}{2} (g_k)^{cd}\left(\nabla^{g_0}_a (g_k - g_0)_{db} + \nabla^{g_0}_b (g_k - g_0)_{ad} - \nabla^{g_0}_d (g_k - g_0)_{ab}\right).
	\end{align*}
	As $g_k$ converges to $g_0$ in $C^1$-norm, we see that $\|\Gamma(g_0,
		g_k)\|_{L^\infty, g_0}$ converges to zero, and, hence, $\|\Gamma(g_0,
		g_k)\|_{L^\infty, g_k}$ also.

	From the $\epsilon$-Young inequality, we have, for any $\theta_k > 0$,
	\begin{align*}
		|\nabla^{g_0} X_k|_{g_k}^2
		 & = |\nabla^{g_k} X_k - \Gamma(g_0, g_k)(\cdot, X_k)|_{g_k}^2                                                                  \\
		 & \leq (1+\theta_k) |\nabla^{g_k} X_k|_{g_k}^2 + \left(1 + \frac{1}{\theta_k}\right) |\Gamma(g_0, g_k)|_{g_k}^2 |X_k|_{g_k}^2,
	\end{align*}
	and, similarly,
	\[
		|\nabla^{g_0} X_k|_{g_k}^2 \geq (1-\theta_k) |\nabla^{g_k} X_k|_{g_k}^2 + \left(1 - \frac{1}{\theta_k}\right) |\Gamma(g_0, g_k)|_{g_k}^2 |X_k|_{g_k}^2.
	\]

	We choose $\theta_k = \left(\|\Gamma(g_0, g_k)\|_{L^\infty, g_k}\right)^{1/2}$
	so that
	\[
		\theta_k \to_{k \to \infty} 0\quad\text{and}\quad \frac{1}{\theta_k} \|\Gamma(g_0, g_k)\|_{L^\infty, g_k}^2 \to_{k \to \infty} 0.
	\]

	Hence, integrating the inequalities over $M$, we get
	\begin{align}
		 & \int_M \left[|\nabla^{g_0} X_k|_{g_k}^2 + |X_k|_{g_k}^2\right] \upd\mu^{g_k} \label{eqUpperBound}                                                                                     \\
		 & \qquad \leq \int_M \left[|\nabla^{g_k} X_k|_{g_k}^2 + |X_k|_{g_k}^2\right] \upd\mu^{g_k}\nonumber                                                                                     \\
		 & \qquad\qquad + \theta_k \int_M |\nabla^{g_k} X_k|_{g_k}^2 \upd\mu^{g_k} + \left(1 + \frac{1}{\theta_k}\right) \int_M |\Gamma(g_0, g_k)|_{g_k}^2 |X_k|_{g_k}^2 \upd\mu^{g_k} \nonumber \\
		 & \qquad \leq \|X_k\|_{H^1, g_k}^2 + \theta_k \|X_k\|_{H^1, g_k}^2 + \left(1 + \frac{1}{\theta_k}\right) \|\Gamma(g_0, g_k)\|_{L^\infty, g_k}^2 \|X_k\|_{H^1, g_k}^2 \nonumber          \\
		 & \qquad = 1 + \theta_k + \left(1 + \frac{1}{\theta_k}\right) \|\Gamma(g_0, g_k)\|_{L^\infty, g_k}^2.\nonumber
	\end{align}

	And, similarly,
	\begin{equation}\label{eqLowerBound}
		\int_M \left[|\nabla^{g_0} X_k|_{g_k}^2 + |X_k|_{g_k}^2\right] \upd\mu^{g_k} \geq 1 - \theta_k - \left(1 + \frac{1}{\theta_k}\right) \|\Gamma(g_0, g_k)\|_{L^\infty, g_k}^2.
	\end{equation}
	Combining the inequalities~\eqref{eqUpperBound} and~\eqref{eqLowerBound}, we
	see that
	\[
		\int_M \left[|\nabla^{g_0} X_k|_{g_k}^2 + |X_k|_{g_k}^2\right] \upd\mu^{g_k} \to 1
	\]
	as claimed. The operator $\bL^{g_k} - \bL^{g_0}$ is a first-order operator
	whose coefficients are controlled by $\|g_k - g_0\|_{C^1, g_0}$, so there is a
	sequence $\delta_k \to 0$ with $\|(\bL^{g_k} - \bL^{g_0}) X\|_{L^2, g_k} \leq
		\delta_k \|X\|_{H^1, g_k}$ for all $X \in H^1(M, TM)$. Applying the
	$\epsilon$-Young inequality as before, comparing the norms $\|\cdot\|_{L^2,
		g_k}$ and $\|\cdot\|_{L^2, g_0}$ as in the first step, and using the
	normalization $\|X_k\|_{H^1, g_k} = 1$, we get, for some new sequence
	$\epsilon_k \to 0$,
	\[
		\int_M |\bL^{g_k} X_k|_{g_k}^2 \upd\mu^{g_k} \geq (1 - \epsilon_k) \int_M |\bL^{g_0} X_k|_{g_0}^2 \upd\mu^{g_0} - \epsilon_k.
	\]
	Hence, we infer that
	\begin{align*}
		\left(1 + \frac{1}{k}\right) \mu_V(g_k)
		 & \geq \int_M |\bL^{g_k} X_k|_{g_k}^2 \upd\mu^{g_k}                               \\
		 & \geq (1 - \epsilon_k) \int_M |\bL^{g_0} X_k|_{g_0}^2 \upd\mu^{g_0} - \epsilon_k \\
		 & \geq (1 - \epsilon_k) \mu_V(g_0) \|X_k\|_{H^1, g_0}^2 - \epsilon_k.
	\end{align*}
	Passing to the limit $k \to \infty$, we obtain
	\[
		\lim_{k \to \infty} \mu_V(g_k) \geq \mu_V(g_0).
	\]
	Continuity of $g \mapsto \bL^g$ for the operator norm follows by the same kind
	of arguments.
\end{proof}

We combine \cref{propTTBundle} with a corollary of Kuiper's theorem, namely
that infinite-dimensional Hilbert bundles over paracompact spaces, such as
$\NCK(M)$, are trivial bundles (see~\cite[Ch.~I.7, Rem.~1]{BoossBleecker}).
Indeed, from~\cite[Prop.~3]{BourguignonEbinMarsden}, the space of (smooth)
TT-tensors over any given Riemannian metric on $M$ is infinite dimensional,
since $n \geq 3$. As a consequence, we have proven the following theorem:

\begin{theorem}\label{thmTT}
	Let $M$ be a closed manifold of dimension $n \geq 3$. The space $\TT(M)$ is a
	trivial Hilbert bundle over $\NCK(M)$.
\end{theorem}

Although \cref{thmTT} is of independent interest, we will mostly use the
following consequence in the sequel:

\begin{corollary}\label{corTT}
	Let $M$ be a closed manifold of dimension $n \geq 3$. There exists a section
	$\sigma \colon \NCK(M) \to \TT(M)$ of the bundle of TT-tensors such that, for
	each $g \in \NCK(M)$, $\sigma(g)$ is a smooth TT-tensor with
	$\|\sigma(g)\|_{L^2, g} = 1$. Moreover, $\sigma$ is continuous as a map from
	$\NCK(M)$ to $C^\infty(M, S_2M)$.
\end{corollary}

\begin{proof}
	From \cref{thmTT}, there exists a nowhere-vanishing continuous section
	$\sigma_0$ of $\TT(M)$. Without loss of generality, we can assume that
	$\|\sigma_0(g)\|_{L^2, g} = 1$ for all $g \in \NCK(M)$. However, there is no
	reason for $\sigma_0(g)$ to be smooth, and $\sigma_0$ is a priori only
	continuous as a map from $\NCK(M)$ to $L^2(M, S_2M)$. We first approximate
	$\sigma_0$ by a family of smooth tensors depending continuously on $g$ in
	$C^\infty$.

	For a smooth symmetric $2$-tensor $T$ and a metric $g$, we denote by
	\[
		\mathring{T}^g \coloneqq T - \frac{1}{n} \tr^g(T)\, g
	\]
	its trace-free part with respect to $g$. The map $g \mapsto (T \mapsto
		\mathring{T}^g)$ is continuous from $\Met(M)$ to $C^\infty(M,
		\mathrm{End}(S_2M))$. Let $g_0 \in \NCK(M)$. Since smooth tensors are dense in
	$L^2(M, S_2M)$, there exists a smooth symmetric $2$-tensor $T$ such that $\|T -
		\sigma_0(g_0)\|_{L^2, g_0} < \frac{1}{2}$. As $\sigma_0(g_0)$ is trace-free
	with respect to $g_0$ and taking the trace-free part is a pointwise orthogonal
	projection, we also have $\|\mathring{T}^{g_0} - \sigma_0(g_0)\|_{L^2, g_0} <
		\frac{1}{2}$. The map $g \mapsto \|\mathring{T}^g - \sigma_0(g)\|_{L^2, g}$
	being continuous, the inequality $\|\mathring{T}^g - \sigma_0(g)\|_{L^2, g} <
		\frac{1}{2}$ holds for all $g$ in an open neighborhood of $g_0$.

	We obtain in this way an open cover $(\mathcal{O}_\alpha)_\alpha$ of $\NCK(M)$
	and smooth symmetric $2$-tensors $T_\alpha$ such that $\|\mathring{T}_\alpha^g
		- \sigma_0(g)\|_{L^2, g} < \frac{1}{2}$ for all $g \in \mathcal{O}_\alpha$. As
	$\NCK(M)$ is paracompact, there is a locally finite partition of unity
	$(\chi_\alpha)_\alpha$ subordinate to this cover, see~\cite[Theorem
		41.7]{Munkres}, and we set
	\[
		\sigma_1(g) \coloneqq \sum_\alpha \chi_\alpha(g)\, \mathring{T}_\alpha^g.
	\]
	The sum being locally finite, $\sigma_1$ is continuous from $\NCK(M)$ to
	$C^\infty(M, S_2M)$, and each $\sigma_1(g)$ is a smooth symmetric $2$-tensor
	that is trace-free with respect to $g$. Moreover, by the triangle inequality,
	\[
		\|\sigma_1(g) - \sigma_0(g)\|_{L^2, g}
		\leq \sum_\alpha \chi_\alpha(g) \|\mathring{T}_\alpha^g - \sigma_0(g)\|_{L^2, g}
		< \frac{1}{2}.
	\]

	Let $\sigmatil(g)$ denote the $L^2$-orthogonal projection of $\sigma_1(g)$ onto
	$\TT(M, g)$. Since $\sigma_0(g) \in \TT(M, g)$, we have $\|\sigmatil(g) -
		\sigma_0(g)\|_{L^2, g} < \frac{1}{2}$, hence $\|\sigmatil(g)\|_{L^2, g} >
		\frac{1}{2}$. Note that projection onto $\TT(M)$ is nothing but York's
	decomposition:
	\[
		\sigma_1(g) = \sigmatil(g) + \bL^g W_g,
	\]
	where $W_g \in H^1(M, TM)$ is unique since $g$ has no non-zero CKV. In
	particular, by \cref{propSplitting}, $W_g$ and $\sigmatil(g)$ are smooth.

	It remains to show that $\sigmatil$ is continuous as a map to $C^\infty(M,
		S_2M)$. We fix a background metric $\gbar$ on $M$ to define the Sobolev spaces
	$H^s$ and their norms. The vector field $W_g$ is the solution of
	\[
		(\bL^g)^* \bL^g W_g = (\bL^g)^* \sigma_1(g),
	\]
	where $(\bL^g)^*$ denotes the formal $L^2$-adjoint of $\bL^g$ with respect to
	$g$. For $s \geq 0$, we consider the operator
	\[
		\begin{array}{rccc}
			P_{g, s} \colon & H^{s+2}(M, TM) & \to     & H^s(M, TM)         \\
			                & W              & \mapsto & (\bL^g)^* \bL^g W.
		\end{array}
	\]
	Its coefficients depend smoothly on $g$ and its first two derivatives, so the
	map $g \mapsto P_{g, s}$ is continuous from $\NCK(M)$ to $\cL(H^{s+2}(M, TM),
		H^s(M, TM))$ for the operator norm. By standard elliptic theory, each $P_{g,
				s}$ is Fredholm with index zero, being elliptic and formally self-adjoint with
	respect to $g$, and its kernel consists of the conformal Killing vectors of
	$g$. For $g \in \NCK(M)$, $P_{g, s}$ is therefore invertible, and $g \mapsto
		P_{g, s}^{-1}$ is continuous for the operator norm, by continuity of the
	inverse. Since $g \mapsto (\bL^g)^* \sigma_1(g)$ is continuous from $\NCK(M)$
	to $H^s(M, TM)$, the map $g \mapsto W_g = P_{g, s}^{-1}\bigl((\bL^g)^*
		\sigma_1(g)\bigr)$ is continuous from $\NCK(M)$ to $H^{s+2}(M, TM)$, and
	$\sigmatil(g) = \sigma_1(g) - \bL^g W_g$ is continuous as a map to $H^{s+1}(M,
		S_2M)$. As $s$ is arbitrary, the Sobolev embedding theorem shows that
	$\sigmatil$ is continuous as a map to $C^\infty(M, S_2M)$. The required section
	$\sigma$ is now obtained by normalizing $\sigmatil$, since $g \mapsto
		\|\sigmatil(g)\|_{L^2, g}$ is continuous and does not vanish.
\end{proof}

\section{Constructing homotopy classes}\label{secHomotopy}
The goal of this section is to prove the following result for closed connected
manifolds~$M$ of dimension $n \geq 3$:

\begin{theorem}\label{thmHomotopy}
	There exists a map $\Psi \colon \sDEC_{0, \CMC}(M) \to \VacC_{0, \CMC}(M)
		\subset \DEC_{0, \CMC}(M)$ such that the inclusion map
	\[
		\mathrm{Incl} \colon \sDEC_{0, \CMC}(M) \hookrightarrow \DEC_{0, \CMC}(M)
	\]
	and $\Psi$ are homotopic as maps from $\sDEC_{0, \CMC}(M)$ to $\DEC_{0,
			\CMC}(M)$.
\end{theorem}

To this end, we are going to construct, given a pair $(g, k) \in \sDEC_{0,
		\CMC}(M)$, a $1$-parameter family of initial data sets $(g_t, k_t) \in \DEC_{0,
		\CMC}(M)$, $t \in [0, 1]$ so that $(g_1, k_1) = (g, k)$ and $(g_0, k_0) \in
	\VacC_{0, \CMC}(M)$. The construction will be continuous in $(g,k)$ and is an
adaptation of the well-established conformal method in general relativity, see
\eg\cite{CarlottoReview}.

Let $(g,k) \in \Ini_{\CMC}(M)$ for an $n$-dimensional manifold $M$. To motivate
our construction, we consider the constraint equations for the pair $(g_t,
	k_t)$:
\begin{equation}\label{eqConstraintsT}
	\left\lbrace
	\begin{aligned}
		\scal^{g_t} + \left(\tr^{g_t} k_t\right)^2 - |k_t|_{g_t}^2 & = 2 \rho_t, \\
		\divg^{g_t} (k_t) - \upd \left(\tr^{g_t} k_t\right)        & = j_t.
	\end{aligned}
	\right.
\end{equation}

We assume that $g_t$ is conformal to the metric $g$, namely
\[
	g_t = \phi_t^\kappa g,
\]
where $\kappa = \frac{4}{n-2}$. In this setting, the dominant energy condition
for $(g_t,k_t)$ reads
\[
	\rho_t \geq |j_t|_{g_t} = \phi_t^{-\kappa/2} |j_t|_g.
\]
This suggests looking for families $(g_t,k_t)$ such that
\[
	j_t = t \phi_t^\alpha j,
	\qquad
	\rho_t = t \phi_t^{\alpha - \kappa/2} \rho,
\]
for some exponent $\alpha$ to be chosen later. Here and in the following,
$\rho$ and $j$ denote the quantities defined by~\eqref{eqDensities} for the
initial data set~$(g,k)$. It is also natural to impose that the mean curvature
$\tau = \tr^{g_t} k_t \in \bR$ remains constant with respect to $t$, so that
\[
	k_t = \frac{\tau}{n} g_t + \kring_t,
\]
where $\kring_t$ denotes the trace-free part of $k_t$. With these choices, the
constraint equations~\eqref{eqConstraintsT} become
\begin{equation} \label{eqConstraintsRewritten}
	\left\lbrace
	\begin{aligned}
		\phi_t^{-1-\kappa}\left(\frac{4(n-1)}{n-2} \Delta \phi_t
		+ \scal^g\, \phi_t\right)
		+ \frac{n-1}{n} \tau^2
		- \phi_t^{-2\kappa}|\kring_t|_g^2
		 & = 2 t \phi_t^{\alpha - \kappa/2} \rho, \\
		\divg^{g_t} (\kring_t)
		 & = t \phi_t^\alpha j.
	\end{aligned}
	\right.
\end{equation}

We recall that for any symmetric trace-free $2$-tensor $T$, one has
\[
	\divg^{g_t} \left(\phi_t^{-2} T\right)
	= \phi_t^{-2-\kappa} \divg^g (T).
\]
Therefore, the momentum constraint (the second equation
of~\eqref{eqConstraintsRewritten}) can be rewritten as
\[
	\divg^g \left(\phi_t^2 \kring_t\right)
	= t \phi_t^{\alpha + 2 + \kappa} j.
\]
This naturally leads to the choice $\alpha = -2-\kappa$, so that
\[
	\divg^g \left(\phi_t^2 \kring_t\right) = t j.
\]
As a consequence, we must have
\[
	\phi_t^2 \kring_t = t \kring + \sigma_t,
\]
where $\kring = k - \frac{1}{n}\tau g$ and $\sigma_t$ is a TT-tensor with
respect to $g$, in order for the momentum constraint to be satisfied. Finally,
plugging this ansatz for $\kring_t$ into the Hamiltonian constraint (the first
equation of~\eqref{eqConstraintsRewritten}) yields
\[
	\phi_t^{-1-\kappa}\left(\frac{4(n-1)}{n-2} \Delta \phi_t
	+ \scal^g\, \phi_t\right)
	+ \frac{n-1}{n} \tau^2
	- \phi_t^{-2\kappa-4} |t \kring + \sigma_t|_g^2
	= 2 t \phi_t^{\alpha - \kappa/2} \rho.
\]
Multiplying both sides by $\phi_t^{1+\kappa}$, we obtain the following
Lichnerowicz-type equation:
\begin{equation}\label{eqLichnerowicz}
	\frac{4(n-1)}{n-2} \Delta \phi_t + \scal^g\, \phi_t
	+ \frac{n-1}{n} \tau^2 \phi_t^{\kappa+1}
	= \frac{|t \kring + \sigma_t|_g^2}{\phi_t^{\kappa+3}}
	+ \frac{2 t \rho}{\phi_t^{1 + \kappa/2}},
\end{equation}
which we will simply call the Lichnerowicz equation in what follows. Let us
summarize the construction described so far in the following lemma.

\begin{lemma}\label{lemIDSForPsi}
	Consider a CMC initial data set $(g,k) \in \Ini_{\CMC}(M)$ defined on an
	$n$-dimensional manifold $M$ and set $\kappa = \frac{4}{n-2}$. Let $\rho$ and
	$j$ be its energy and momentum density, defined by~\eqref{eqDensities}, and set
	\begin{align*}
		\tau   & \coloneqq \tr^g(k) \qquad \text{and} \\
		\kring & \coloneqq k - \frac{1}{n}\tau g.
	\end{align*}
	Let $t \in [0,1]$ and suppose that $\sigma_t$ is a TT-tensor with respect to
	$g$ and that $\phi_t$ is a positive solution to~\eqref{eqLichnerowicz}. Then
	$(g_t,k_t)$ defined by
	\begin{align}\label{eqDefIDSForPsi}
		g_t \coloneqq \phi_t^\kappa g, \qquad k_t \coloneqq \frac{\tau}{n} g_t +
		\phi_t^{-2}\bigl(t\kring + \sigma_t\bigr)
	\end{align}
	is a CMC initial data set with $\rho_t = t\phi_t^{-2-3\kappa/2}\rho$, $j_t =
		t\phi_t^{-2-\kappa}j$ and thus $|j_t|_{g_t} = t\phi_t^{-2-3\kappa/2}|j|_g$. In
	particular, for $t > 0$, $(g_t,k_t)$ satisfies the (strict) dominant energy
	condition if and only if $(g,k)$ does, and $(g_0,k_0)$ satisfies the vacuum
	constraints. Finally, for $t = 1$ and $\sigma_1 = 0$, the constant function
	$\phi_1 = 1$ solves~\eqref{eqLichnerowicz}, and the corresponding initial data
	set is $(g_1,k_1) = (g,k)$.
\end{lemma}

In order to conclude \cref{thmHomotopy} from this lemma, we have to do the
following for the initial data sets of interest: firstly, choose a continuous
family of TT-tensors -- this was done in \cref{secTTtensors} -- and secondly,
show that the Lichnerowicz equation~\eqref{eqLichnerowicz} admits a unique
positive solution, continuously depending on the parameters. We treat a
slightly more general version of this equation, using that $\rho \geq 0$ for
DEC initial data sets $(g,k)$. From now on, $M$ is always closed and connected.

\begin{proposition}\label{propLichnerowicz}
	The equation
	\begin{equation}\label{eqLichnerowicz2}
		\frac{4(n-1)}{n-2} \Delta \phi + \scal^g\, \phi
		+ \frac{n-1}{n} \tau^2 \phi^{\kappa+1}
		= \frac{A^2}{\phi^{\kappa+3}} + \frac{B^2}{\phi^{1 + \kappa/2}},
	\end{equation}
	with $A^2$ and $B^2$ two smooth non-negative functions and $\tau$ a constant,
	admits a unique positive solution $\phi \in C^\infty(M, \bR)$ unless
	\begin{itemize}
		\item $\cY(M, g) \geq 0$ and $A^2 \equiv B^2 \equiv 0$,
		\item $\cY(M, g) \leq 0$ and $\tau = 0$.
	\end{itemize}
	Further, there is no positive solution when only one of these two conditions
	holds, while there is a $1$-dimensional set of positive solutions
	to~\eqref{eqLichnerowicz2} when both of them hold.
\end{proposition}

Solving the Lichnerowicz equation with $B^2 \equiv 0$ is by now
well-established, going back to~\cite{Isenberg}, see
\eg\cite{GicquaudLichnerowicz}. However, the equation we are considering here
is more general as it corresponds to the non-vacuum case. Such an equation has
been studied \eg in~\cite{HNT1,HNT2} but their study was limited to the case of
dimension $n=3$. While it is well known among people working on initial data
sets that these results can be straightforwardly generalized to higher
dimensions and that the $B$-term does not modify the analysis of the equation,
we give a full treatment here for the sake of completeness. We will use the
fact that the sign of the Yamabe constant $\cY(M, g)$ is the same as that of
the first eigenvalue $\lambda_1(L_0)$ of the conformal Laplacian $L_0$ defined
as
\[
	L_0 u = \frac{4(n-1)}{n-2} \Delta u + \scal^g u.
\]
This fact is proven \eg in~\cite[Prop.~2.2]{GicquaudLichnerowicz}. We decompose
the proof of the proposition into several claims.

\begin{claim}
	The Lichnerowicz equation~\eqref{eqLichnerowicz2} admits no positive solution
	in the case where
	\begin{itemize}
		\item either $\cY(M, g) \geq 0$ and $A^2 \equiv B^2 \equiv 0$,
		\item or $\cY(M, g) \leq 0$ and $\tau = 0$,
	\end{itemize}
	unless $\cY(M, g) = 0$, $\tau = 0$ and $A^2 \equiv B^2 \equiv 0$. In that case,
	the set of positive solutions to~\eqref{eqLichnerowicz2} consists of the
	positive multiples of the first eigenfunction $u_0 > 0$ of the conformal
	Laplacian, $L_0 u_0 = 0$.
\end{claim}

\begin{proof}
	Suppose~\eqref{eqLichnerowicz2} admits a positive solution $\phi$ and consider
	the metric $\ghat \coloneqq \phi^\kappa g$. Its scalar curvature is given by
	\[
		\scal^{\ghat} = \phi^{-\kappa-1} \left(\frac{4(n-1)}{n-2} \Delta \phi + \scal^g\, \phi\right) = - \frac{n-1}{n} \tau^2 + \frac{A^2}{\phi^{2\kappa+4}} + \frac{B^2}{\phi^{2 + 3\kappa/2}}.
	\]

	\begin{itemize}
		\item \textbf{Case 1:} If $A^2 \equiv B^2 \equiv 0$, then $\scal^{\ghat} = -\frac{n-1}{n} \tau^2$.
		      \begin{itemize}
			      \item If $\tau \neq 0$, then $\scal^{\ghat} < 0$, implying $\cY(M, g) = \cY(M, \ghat)
				            < 0$.
			      \item If $\tau = 0$, then $\scal^{\ghat} = 0$, implying $\cY(M, g) = \cY(M, \ghat) =
				            0$.
		      \end{itemize}

		\item \textbf{Case 2:} If $\tau = 0$, then
		      $\scal^{\ghat} = \frac{A^2}{\phi^{2\kappa+4}} + \frac{B^2}{\phi^{2 + 3\kappa/2}} \geq 0$,
		      forcing $\cY(M, g) \geq 0$ with equality if and only if $A^2 \equiv B^2 \equiv 0$.
	\end{itemize}

	In the case where $\cY(M, g) = 0$ and $A^2 \equiv B^2 \equiv 0$, the equation
	reduces to $L_0 \phi = 0$. The positive solutions are the positive multiples of
	the first eigenfunction $u_0 > 0$ of the conformal Laplacian, $L_0 u_0 = 0$.
\end{proof}

We now address the cases of existence of a unique solution to the Lichnerowicz
equation. So we assume that we are in none of the exceptional cases. If $A^2 +
	B^2 \equiv 0$, this means that $\cY(M,g) < 0$ and further $\tau \neq 0$. Then,
since $\tau$ is a constant, the unique solution is provided by the resolution
of the Yamabe problem. If $A^2 + B^2 \not\equiv 0$, this means that $\tau \neq
	0$ or $\cY(M,g) > 0$. The proof of existence of a unique solution in these
cases is based on the sub- and super-solution method (see
\eg\cite[Ch.~14]{Taylor3}).

One key ingredient in the sequel is the following positivity lemma.

\begin{lemma}\label{lemPositivity}
	Let $a > 0$ be a constant, let $f \in C^\infty(M, \bR)$, and consider the
	formally self-adjoint second-order elliptic operator $L = a\Delta + f$ on $M$,
	assumed to have positive first eigenvalue $\lambda_1(L) > 0$. For any $h \in
		C^\infty(M, \bR)$ with $h \geq 0$ and $h \not\equiv 0$, the unique solution $u
		\in C^\infty(M, \bR)$ to $Lu = h$ satisfies $u > 0$.
\end{lemma}

\begin{proof}
	We first claim that the bounded symmetric bilinear form $B_L$ defined over
	$H^1(M, \bR)$ by
	\[
		B_L(u,v) = \int_M \left[a \langle \upd u, \upd v\rangle + f uv\right]\, \upd\mu^g
	\]
	is coercive. First, since $\lambda_1(L) > 0$, we have, by Rayleigh's principle,
	that for any $v \in H^1(M, \bR)$,
	\[
		B_L(v,v) \geq \lambda_1(L) \|v\|_{L^2}^2.
	\]
	We now promote this inequality to the coercivity of $B_L$. Set $\Lambda
		\coloneqq \|f\|_{L^\infty} + 1$. Then
	\[
		B_L(v,v) + \Lambda \|v\|_{L^2}^2 = \int_M \left[a |\upd v|^2 + (f + \|f\|_{L^\infty} + 1) v^2\right]\, \upd\mu^g \geq \min(a, 1) \|v\|_{H^1}^2.
	\]
	Thus
	\[
		\left(1+ \frac{\Lambda}{\lambda_1(L)}\right)B_L(v,v)
		\geq B_L(v,v)+ \Lambda \|v\|_{L^2}^2
		\geq \min(a, 1) \|v\|_{H^1}^2.
	\]

	By the Lax-Milgram theorem, applied to the coercive bilinear form $B_L$ on the
	real Hilbert space $H^1(M, \bR)$, there is a unique solution $u \in H^1(M,\bR)$
	of the equation
	\begin{gather*}
		B_L(u,v) = \int_M hv \,\upd \mu^g \qquad \text{ for all } v \in H^1(M,\bR),
	\end{gather*}
	\ie a unique weak solution of $Lu = h$.
	It may be characterized as the unique minimizer
	of the functional
	\[
		J_h(v) = \frac{1}{2} \int_M \left[a |\upd v|^2 + f v^2\right]\, \upd\mu^g - \int_M h v\, \upd\mu^g.
	\]
	By elliptic regularity, $u \in C^\infty(M)$ and solves $Lu = h$ classically.

	Writing $u_\pm = \max(\pm u,\,0)$, we have $u_\pm \geq 0$ and $u = u_+ - u_-$.
	By~\cite[Lem.~7.6]{GilbargTrudinger}, applied in local charts, we have $u_\pm
		\in H^1(M,\bR)$, with $\upd u_+ = \upd u$ on $\{u > 0\}$, $\upd u_- = -\upd u$
	on $\{u < 0\}$, and $\upd u_\pm = 0$ elsewhere. In particular, $u_+u_- = 0$
	everywhere and $\langle \upd u_+, \upd u_- \rangle = 0$ almost everywhere, so
	that $B_L(u,u) = B_L(u_+,u_+) + B_L(u_-,u_-)$. Using $\int_M hu \,\upd \mu^g =
		\int_M hu_+ \,\upd \mu^g - \int_M hu_- \,\upd \mu^g$ we obtain the
	decomposition
	\[
		J_h(u) = J_h(u_+) + \frac{1}{2} B_L(u_-,u_-) + \int_M h u_-\, \upd\mu^g.
	\]
	Since $B_L(u_-,u_-) \geq 0$ and $\int_M hu_-\,\upd \mu^g \geq 0$ (as $h \geq
		0$), the right-hand side is bounded below by $J_h(u_+)$. But $u$ is the unique
	minimizer of $J_h$, so $u = u_+ \geq 0$.

	It remains to show $u > 0$. Since $h \not\equiv 0$, the constant function $u
		\equiv 0$ is not a solution. Assume, by contradiction, that there exists a
	point $x_0 \in M$ such that $u(x_0) = 0$. Writing $f = f^+ - f^-$ with $f^\pm =
		\max(\pm f, 0)$, we have
	\[
		a \Delta u + f^+ u = h + f^- u \geq 0,
	\]
	since $u \geq 0$. Let $r > 0$ be such that $r < \mathrm{inj}_g(x_0)$. The
	function $u$ attains its minimum value at $x_0$, which is an interior point of
	$B(x_0, r)$. By~\cite[Thm.~8.19]{GilbargTrudinger}, $u$ is constant on $B(x_0,
		r)$. Hence $u^{-1}(0)$ is open. As it is also closed and $M$ is connected, $u$
	vanishes on $M$, a contradiction.
\end{proof}

We now continue to show that there is a solution in the case $A^2 + B^2
	\not\equiv 0$ with $\tau \neq 0$ or $\cY(M,g) > 0$. Since $\tau$ is constant,
the operator
\[
	L_\Lambda \coloneqq L_0 + \frac{n-1}{n}\tau^2\Lambda
\]
has first eigenvalue $\lambda_1(L_0) + \frac{n-1}{n}\tau^2\Lambda$. When $\tau
	\neq 0$, choose
\[
	\Lambda \geq \frac{n}{(n-1)\tau^2}\max\bigl\{0,\,-\lambda_1(L_0)\bigr\} + 1
\]
so that $\lambda_1(L_\Lambda) > 0$. When $\tau = 0$ set $\Lambda = 0$ (so
$L_\Lambda = L_0$, which has positive first eigenvalue since $\cY(M,g) > 0$).
We let $u \in C^\infty(M)$ be the unique solution to
\begin{equation}\label{eqDefU}
	L_\Lambda u = A^2 + B^2.
\end{equation}
By \cref{lemPositivity}, $u > 0$.

\begin{claim}\label{clSuper}
	The function $\phi_+ \coloneqq \lambda u$ (resp. $\phi_- = \mu u$) is a
	super-solution (resp.\ a sub-solution) to~\eqref{eqLichnerowicz2} for $\lambda
		> 0$ sufficiently large (resp.\ $\mu > 0$ sufficiently small):
	\begin{align*}
		 & \frac{4(n-1)}{n-2} \Delta \phi_+ + \scal^g\, \phi_+
		+ \frac{n-1}{n} \tau^2 \phi_+^{\kappa+1}
		\geq \frac{A^2}{\phi_+^{\kappa+3}} + \frac{B^2}{\phi_+^{1 + \kappa/2}}  \\
		 & \text{(resp.\ }\ \frac{4(n-1)}{n-2} \Delta \phi_- + \scal^g\, \phi_-
		+ \frac{n-1}{n} \tau^2 \phi_-^{\kappa+1}
		\leq \frac{A^2}{\phi_-^{\kappa+3}} + \frac{B^2}{\phi_-^{1 + \kappa/2}}\ \text{)}.
	\end{align*}
\end{claim}

\begin{proof}
	Using~\eqref{eqDefU}, the left-hand side of the super-solution condition
	\[
		L_0(\lambda u) + \frac{n-1}{n}\tau^2(\lambda u)^{\kappa+1}
		\geq \frac{A^2}{(\lambda u)^{\kappa+3}} + \frac{B^2}{(\lambda u)^{1+\kappa/2}}
	\]
	equals
	\[
		\lambda(A^2+B^2) + \tfrac{n-1}{n}\tau^2\lambda u
		\bigl(\lambda^\kappa u^\kappa - \Lambda\bigr).
	\]
	So, after rearranging, $\phi_+$ is a super-solution to the Lichnerowicz
	equation if and only if
	\[
		\frac{n-1}{n}\tau^2 \lambda u\left(\lambda^\kappa u^\kappa - \Lambda \right)
		\geq A^2 \left(\frac{1}{\lambda^{\kappa+3} u^{\kappa+3}} - \lambda\right) + B^2 \left(\frac{1}{\lambda^{1+\kappa/2} u^{1+\kappa/2}} - \lambda\right).
	\]
	This inequality is fulfilled as soon as the following three conditions are met:
	\[
		\lambda^\kappa u^\kappa \geq \Lambda,
		\
		\lambda^{\kappa+4}u^{\kappa+3} \geq 1,
		\ \text{and}\
		\lambda^{2+\kappa/2}u^{1+\kappa/2} \geq 1.
	\]
	As $u$ is bounded from below by a positive constant, this is the case if
	$\lambda > 0$ is large enough. A similar analysis can be done for the
	sub-solution $\phi_-$ upon reversing all the inequalities. As $u$ is bounded
	from above, we see that, if $\mu > 0$ is small enough, $\phi_- = \mu u$ is a
	sub-solution to the Lichnerowicz equation.
\end{proof}

Since $\phi_- \leq \phi_+$, the sub- and super-solution method (see
\eg\cite[Ch.~14]{Taylor3}) yields a smooth positive solution $\phi$
to~\eqref{eqLichnerowicz2} with $\phi_- \leq \phi \leq \phi_+$. The only
remaining question is that of the uniqueness of the solution $\phi$
to~\eqref{eqLichnerowicz2}. We follow here the approach used
in~\cite{GicquaudLichnerowicz}. Let $\phi_1, \phi_2$ be two positive smooth
solutions to~\eqref{eqLichnerowicz2}. Setting
\[
	f(\phi) \coloneqq \frac{n-2}{4(n-1)}\!\left[
		\frac{A^2}{\phi^{\kappa+4}} + \frac{B^2}{\phi^{2+\kappa/2}}
		- \scal^g - \frac{n-1}{n}\tau^2\phi^{\kappa}
		\right],
\]
each solution satisfies $\Delta\phi_i/\phi_i = f(\phi_i)$. By integration by
parts~\cite{BrezisOswald},
\begin{align}\label{eqBrezisOswald}
	 & \int_M\!(f(\phi_1) - f(\phi_2))(\phi_1^2 - \phi_2^2)\,\upd\mu^g \nonumber \\
	 & = \int_M\!\left(
	\frac{\Delta\phi_1}{\phi_1} - \frac{\Delta\phi_2}{\phi_2}
	\right)\!(\phi_1^2 - \phi_2^2)\,\upd\mu^g \nonumber                          \\
	 & \qquad =
	\int_M\!(\phi_1^2 + \phi_2^2) \left|\frac{\upd\phi_1}{\phi_1} - \frac{\upd\phi_2}{\phi_2}\right|^2\!\upd\mu^g.
\end{align}
Viewing $f$ as a family of smooth functions $f_x \colon \bR_{>0} \to \bR$
parametrized over $x \in M$, we see that all the functions $f_x$ are
non-increasing:
\[
	f_x'(\phi) = \frac{n-2}{4(n-1)}\!\left[
		-(\kappa{+}4)\frac{A(x)^2}{\phi^{\kappa+5}}
		- \!\left(2{+}\tfrac{\kappa}{2}\right)\!\frac{B(x)^2}{\phi^{3+\kappa/2}}
		- \frac{(n-1)\kappa}{n}\tau^2\phi^{\kappa-1}
		\right] \leq 0.
\]
Hence $[f(\phi_1)-f(\phi_2)](\phi_1^2-\phi_2^2) \leq 0$ pointwise,
and~\eqref{eqBrezisOswald} forces $\upd \log(\phi_1/\phi_2) \equiv 0$, so the
ratio $\phi_1/\phi_2$ equals a positive constant $c$.

Substituting $\phi_1 = c\phi_2$ into~\eqref{eqLichnerowicz2} and subtracting
$c$ times the equation for $\phi_2$ gives
\begin{equation}\label{eqConstancy}
	\frac{n-1}{n}\tau^2\phi_2^{\kappa+1}(c^{\kappa+1}-c)
	+ \frac{A^2}{\phi_2^{\kappa+3}}\bigl(c - c^{-\kappa-3}\bigr)
	+ \frac{B^2}{\phi_2^{1+\kappa/2}}\bigl(c - c^{-1-\kappa/2}\bigr)
	= 0.
\end{equation}
Since $\kappa > 0$ and $c > 0$, the three factors $(c^{\kappa+1}-c)$, $(c -
	c^{-\kappa-3})$, $(c - c^{-1-\kappa/2})$ all share the sign of $c-1$. As
$\tau^2, A^2, B^2 \geq 0$ and $\phi_2 > 0$, equation~\eqref{eqConstancy} forces
$c = 1$, unless $\tau = 0$ and $A \equiv B \equiv 0$, in which case any $c > 0$
is possible. This concludes the proof of \cref{propLichnerowicz}.

We now apply \cref{propLichnerowicz} to show that for any $(g, k) \in \sDEC_{0,
		\CMC}(M)$, any $t \in [0, 1]$ and a suitable choice of TT-tensor $\sigma_t$ the
Lichnerowicz equation~\eqref{eqLichnerowicz} admits a unique positive solution
$\phi \in C^\infty(M,\bR)$. Here, the choice of TT-tensors is motivated by the
requirement that $\sigma_1 = 0$ for $t=1$ in the application later, while we
generally need $\sigma_0 \not\equiv 0$ to ensure existence of a positive
solution for $t=0$, as the proof shows.

\begin{corollary}\label{corExistence}
	For any $(g, k) \in \sDEC_{\CMC}(M)$, any TT-tensor $\sigma \not\equiv 0$ with
	respect to $g$ and any $t \in [0, 1]$, there exists a unique positive solution
	$\phi_t \in C^\infty(M, \bR)$ to the equation~\eqref{eqLichnerowicz} with
	$\sigma_t = (1-t)\sigma$.
\end{corollary}

\begin{proof}
	Since $(g, k) \in \sDEC_{\CMC}(M)$, the strict DEC gives
	\[
		\scal^g + \frac{n-1}{n} \tau^2 \geq \scal^g + \frac{n-1}{n} \tau^2 - |\kring|^2 = 2 \rho > 0.
	\]
	In particular, we either have $\scal^g > 0$ everywhere or, at a point $x_0$
	where $\scal^g(x_0) \leq 0$, we have
	\[
		\frac{n-1}{n} \tau^2 \geq 2 \rho(x_0) > 0.
	\]
	This means that at least one of the following two assumptions is fulfilled:
	\[
		\cY(M, g) > 0 \ \text{or}\ \tau \neq 0.
	\]
	Further, if $t > 0$, comparing~\eqref{eqLichnerowicz}
	and~\eqref{eqLichnerowicz2}, we have
	\[
		B^2 = 2 t \rho > 0.
	\]
	When $t = 0$, comparing once again~\eqref{eqLichnerowicz}
	and~\eqref{eqLichnerowicz2}, we have
	\[
		A^2 = |t \kring + (1-t) \sigma|^2 = |\sigma|^2 \not\equiv 0.
	\]
	All in all, we can never have $A^2 \equiv B^2 \equiv 0$. Now the statement
	follows from \cref{propLichnerowicz}.
\end{proof}

We finally need to see that this solution depends continuously on $(g, k, t)$.
This requires a continuous choice of TT-tensors. In the following, we will use
the continuous choice $\sigma(g) \in \TT(M,g) \setminus \{0\}$ of a non-zero
TT-tensor with respect to $g \in \NCK(M)$ provided by \cref{corTT} and then set
$\sigma_t=(1-t)\sigma(g)$ as above.

\begin{lemma}\label{lemContinuity}
	The positive solution $\phi \in C^\infty(M, \bR)$ to the Lichnerowicz
	equation~\eqref{eqLichnerowicz} with $\sigma_t = (1-t)\sigma(g)$ depends
	continuously on $(g, k, t) \in \sDEC_{0, \CMC}(M) \times [0,1]$.
\end{lemma}

\begin{proof}
	We are going to show that, for each $s > \max\{n/2, 2\}$, the map $\Theta$
	taking a triple $(g, k, t)$ to the solution $\phi_t$ to the Lichnerowicz
	equation~\eqref{eqLichnerowicz} with $\sigma_t = (1-t)\sigma(g)$ is continuous
	as a map
	\[
		\Theta \colon \sDEC_{0,\CMC}(M) \times [0, 1] \to H^s(M, \bR_{>0}).
	\]
	We do this by means of the implicit function theorem. We consider the map
	\[
		\mathrm{Lich} \colon \sDEC_{0, \CMC}(M) \times [0, 1] \times H^s(M, \bR_{> 0}) \to H^{s-2}(M, \bR)
	\]
	defined by
	\[
		\mathrm{Lich}(g, k, t, \phi) \coloneqq \frac{4(n-1)}{n-2} \Delta \phi
		+ \scal^g \phi + \frac{n-1}{n} \tau^2 \phi^{\kappa+1}
		- \frac{|t \kring + (1-t)\sigma(g)|_g^2}{\phi^{\kappa+3}}
		- \frac{2 t \rho}{\phi^{1 + \kappa/2}},
	\]
	where $\tau \coloneqq \tr^g(k)$, $\kring \coloneqq k - \frac{1}{n}\tau g$ and
	$\rho \coloneqq \frac{1}{2}\bigl(\scal^g + (\tr^g k)^2 - |k|_g^2\bigr)$.
	Because of our assumption on $s$, we have the Sobolev embedding $H^s(M, \bR)
		\hookrightarrow L^\infty(M, \bR)$, $H^s(M,\bR)$ is a Banach algebra, and
	composition with smooth functions on the open set $\{\phi > 0\}$ is
	continuously differentiable on it (\cf\cite[App.~III]{ChoquetBruhat}). Since
	moreover $\sigma$ is continuous as a map from $\NCK(M)$ to $C^\infty(M, S_2M)$
	by \cref{corTT}, the map $\mathrm{Lich}$ is well-defined and continuous, it is
	continuously differentiable with respect to $\phi$, and its partial
	differential $D\mathrm{Lich}_\phi$ depends continuously on $(g, k, t, \phi)$
	for the operator norm. These are the assumptions of the implicit function
	theorem when the parameter $(g, k, t)$ lies in a topological space, since its
	proof by the contraction mapping principle only uses the continuity of the
	dependence on the parameter. Moreover, by the uniqueness in
	\cref{corExistence}, the positive solution provided by the implicit function
	theorem is $\Theta(g, k, t)$. All we need to see is therefore that
	$D\mathrm{Lich}_\phi$ is invertible at a solution $\phi$, \ie at a point where
	$\mathrm{Lich}(g, k, t, \phi) = 0$.

	To keep the formulas readable, we write, as in~\eqref{eqLichnerowicz2},
	\[
		A^2 \coloneqq |t \kring + (1-t)\sigma(g)|_g^2
		\qquad \text{and} \qquad
		B^2 \coloneqq 2 t \rho.
	\]
	The differential of $\mathrm{Lich}$ with respect to $\phi$ in the direction
	$\psi$ is given by
	\begin{align*}
		D\mathrm{Lich}_\phi(\psi) & = \frac{4(n-1)}{n-2} \Delta \psi + \scal^g \psi             \\
		                          & \qquad + \left[\frac{n-1}{n}(\kappa+1) \tau^2 \phi^\kappa +
			(\kappa+3)\frac{A^2}{\phi^{\kappa+4}} + \left(1 +
			\frac{\kappa}{2}\right)\frac{B^2}{\phi^{2 + \kappa/2}}\right] \psi.
	\end{align*}
	For simplicity, we denote by $h \geq 0$ the function appearing inside the
	brackets. So all we need to do is to show that $D\mathrm{Lich}_\phi \colon
		H^s(M, \bR) \to H^{s-2}(M, \bR)$ is an isomorphism, \ie that for any $v \in
		H^{s-2}(M, \bR)$, the equation
	\[
		D\mathrm{Lich}_\phi(\psi) = \frac{4(n-1)}{n-2} \Delta \psi + \scal^g \psi + h \psi = v
	\]
	admits a unique solution. By standard elliptic theory, the operator
	$D\mathrm{Lich}_\phi$ is Fredholm with zero index so all we need to do is to
	prove that the equation $D\mathrm{Lich}_\phi(\psi) = 0$ only admits zero as a
	solution. We use again the strategy from the proof of \cref{propLichnerowicz}.
	We rewrite it as follows:
	\begin{align*}
		0
		 & =
		\frac{4(n-1)}{n-2} \Delta \left(\phi \frac{\psi}{\phi}\right) + \scal^g \psi + h \psi \\
		 & = \frac{4(n-1)}{n-2} \left[\phi \Delta \left(\frac{\psi}{\phi}\right)
			- 2 \left\<\upd\phi, \upd\left(\frac{\psi}{\phi}\right)\right\>\right]
		+ \left(\frac{4(n-1)}{n-2} \Delta \phi + \scal^g \phi + h \phi\right) \left(\frac{\psi}{\phi}\right),
	\end{align*}
	where we used the product rule $\Delta(uv) = u \Delta v - 2\<\upd u, \upd v\> +
		v \Delta u$. Using $\mathrm{Lich}(g, k, t, \phi) = 0$, the first factor of the
	second summand evaluates as
	\begin{align*}
		\frac{4(n-1)}{n-2} \Delta \phi + \scal^g \phi + h \phi
		 & = - \frac{n-1}{n} \tau^2 \phi^{\kappa+1}
		+ \frac{A^2}{\phi^{\kappa+3}} + \frac{B^2}{\phi^{1 + \kappa/2}}    \\
		 & \quad + \frac{n-1}{n}(\kappa+1) \tau^2 \phi^{\kappa+1}
		+ (\kappa+3)\frac{A^2}{\phi^{\kappa+3}}
		+ \left(1 + \frac{\kappa}{2}\right)\frac{B^2}{\phi^{1 + \kappa/2}} \\
		 & = \frac{n-1}{n}\kappa \tau^2 \phi^{\kappa+1}
		+ (\kappa+4) \frac{A^2}{\phi^{\kappa+3}}
		+ \left(2 + \frac{\kappa}{2}\right) \frac{B^2}{\phi^{1 + \kappa/2}}.
	\end{align*}
	From the maximum principle~\cite[Thm.~8.19]{GilbargTrudinger}, we see that
	$\psi/\phi$ vanishes unless
	\[
		\frac{n-1}{n}\kappa \tau^2 \phi^{\kappa+1}
		+ (\kappa+4) \frac{A^2}{\phi^{\kappa+3}}
		+ \left(2 + \frac{\kappa}{2}\right) \frac{B^2}{\phi^{1 + \kappa/2}} \equiv 0.
	\]
	But if $t > 0$, this is impossible because $\rho > 0$ and, if $t = 0$, we have
	$\sigma(g) \not\equiv 0$. All in all, we conclude that $D\mathrm{Lich}_\phi$ is
	an isomorphism. The implicit function theorem then allows us to conclude that
	$\phi \in H^s(M, \bR_{> 0})$ depends continuously on $(g, k, t)$. As $s$ is
	arbitrary, we have that $\phi$ depends continuously on $(g, k, t)$ in
	$C^\infty(M, \bR_{>0})$.
\end{proof}

\begin{proof}[Proof of \cref{thmHomotopy}]
	Let $g \mapsto \sigma(g) \in \TT(M,g) \setminus \{0\}$ be a continuous choice
	of non-zero TT-tensors for metrics $g \in \NCK(M)$, which exists by
	\cref{corTT}.
	Given $(g, k) \in \sDEC_{0, \CMC}(M)$, we set $\kappa$, $\rho$, $\tau$ and
	$\kring$ as in \cref{lemIDSForPsi} and consider the Lichnerowicz
	equation~\eqref{eqLichnerowicz} for $t \in [0,1]$ and $\sigma_t =
		(1-t)\sigma(g)$.

	By \cref{corExistence}, there is a unique positive solution $\phi_t \in
		C^\infty(M, \bR)$ to this equation. The formula~\eqref{eqDefIDSForPsi} then
	defines an initial data set $(g_t,k_t) \in \DEC_{0,\CMC}(M)$. Here, the part of
	the statement that it is a CMC initial data set satisfying DEC is contained in
	\cref{lemIDSForPsi}, while the part that $g_t \in \NCK(M)$ follows directly
	from the conformal invariance of the existence of CKVs. \Cref{lemIDSForPsi}
	implies further that $(g_0,k_0) \in \VacC_{0,\CMC}(M)$ and that $(g_1,k_1) =
		(g,k)$ since $\sigma_1 = 0$, so that $\phi_1 = 1$ solves~\eqref{eqLichnerowicz}
	and is, by uniqueness, the solution.

	Since $\sigma$ is continuous with values in $C^\infty(M, S_2M)$ (\cref{corTT}),
	\cref{lemContinuity} ensures that $\phi_t$, and hence $(g_t,k_t)$, depends
	continuously on $(g, k, t)$, yielding a continuous map $H \colon \sDEC_{0,
			\CMC}(M) \times [0,1] \to \DEC_{0, \CMC}(M)$. This map is the required homotopy
	from the inclusion $\mathrm{Incl} = H(\blank,1)$ to the map $\Psi \coloneqq
		H(\blank,0) \colon \sDEC_{0, \CMC}(M) \to \VacC_{0, \CMC}(M) \subset \DEC_{0,
			\CMC}(M)$, which takes its values in $\VacC_{0, \CMC}(M)$.
\end{proof}

\appendix
\section{More details on the index difference}\label{secOladiff}
In this section, we briefly recall the definition of the index difference
\begin{align*}
	\oladiff \colon \pi_\ell(\sDEC(M), (g_0,k_0))	\lto \KO^{-n-\ell}(\pt),
\end{align*}
defined for a base point $(g_0,k_0) \in \sDEC(M)$, from~\cite{Gloeckle2024a}.
Throughout, $M$ is a closed spin manifold of dimension $n$ with chosen spin
structure.

Associated to the spin structure of $M$, there is a $\Cl_n$-linear spinor
bundle $\Sigma_{\Cl} M \to M$. It is equipped with a scalar product, a
$\Ztwo$-grading, a (left) Clifford multiplication by $TM$, a (right)
$\Cl_n$-action and a connection $\nabla$, all of which are compatible. The
$\Cl_{n,1}$-linear \emph{hypersurface spinor bundle} $\overline\Sigma_{\Cl} M
	\to M$ is defined similarly by associating $\Cl_{n,1}$ to the spin structure
via left multiplication. Since $\overline\Sigma_{\Cl} M \cong \Sigma_{\Cl} M
	\otimes_{\Cl_n} \Cl_{n,1}$, it may just be viewed as two copies of
$\Sigma_{\Cl} M$, but its important property is that the Clifford actions
extend. Most notably, this amounts to the existence of an involution, which we
denote by $e_0 \cdot$, that anti-commutes with the left $TM$-action. This can
be used to define, in the spirit of Witten~\cite{Witten}, a modified connection
$\overline{\nabla}$ by
\begin{gather} \label{eqDefWittenConn}
	\overline{\nabla}_X \psi \coloneqq \nabla_X \psi
	- \frac{1}{2} e_0 \cdot k(X,\blank)^\sharp \cdot \psi
\end{gather}
for $X \in TM$ and $\psi \in \Gamma(\overline{\Sigma}_{\Cl}M)$.

The ($\Cl_{n,1}$-linear) \emph{Dirac-Witten operator} is the self-adjoint
elliptic first-order differential operator on $\overline{\Sigma}_{\Cl} M$
defined by the local formula
\begin{gather} \label{eqDefDiracWitten}
	\overline{D} \psi
	\coloneqq \sum_{i = 1}^n e_i \cdot \overline{\nabla}_{e_i} \psi,
\end{gather}
where $\psi \in \Gamma(\overline{\Sigma}_{\Cl}M)$ and $e_1, \ldots, e_n$ is a
local orthonormal basis of $TM$. By virtue of~\eqref{eqDefWittenConn}, it
satisfies the following comparison formula to the Dirac operator $D$ on
$\overline{\Sigma}_{\Cl} M$ (defined by~\eqref{eqDefDiracWitten} with
$\overline{\nabla}$ replaced by $\nabla$):
\begin{gather} \label{eqCompDW}
	\overline{D} = D - \frac{1}{2} \tr^g(k) e_0 \cdot
\end{gather}
Moreover, it satisfies the Schrödinger-Lichnerowicz type formula
\begin{gather*}
	\overline{D}^2 = \overline{\nabla}^* \overline{\nabla}
	+ \frac{1}{2} (\rho - e_0 \cdot j^\sharp \cdot),
\end{gather*}
which in particular implies that for initial data sets with $\rho > |j|_g$ the
associated Dirac-Witten operator has zero kernel and is hence invertible. This
motivates the following definition:

\begin{definition}\label{defInvI}
	For a spin manifold $M$ with fixed spin structure, we define the subspace of
	\emph{initial data sets with invertible Dirac-Witten operator} by $\InvI(M)
		\coloneqq \{(g,k) \in \Ini(M) \mid \ker(\overline{D}) = 0\}$.
\end{definition}

\begin{remark}
	Note that $\InvI(M)$ depends on the chosen spin structure. For example, the
	Dirac-Witten operator of the flat metric on the torus $T^2$ and $k = 0$ is
	invertible for all but one of the four spin structures on $T^2$.
\end{remark}

\begin{remark}
	Since the second fundamental form $k$ only enters the Dirac-Witten operator
	through its trace, \cf\eqref{eqCompDW}, whether $(g,k)$ lies in $\InvI(M)$
	depends on $g$ and $\tr^g(k)$ alone.
\end{remark}

Because $\overline{D}$ is a self-adjoint elliptic differential operator on a
closed manifold, standard analytic results
(\cf\eg\cite[Thm.~III.5.8]{LawsonMichelsohn}) imply that its bounded transform
$\frac{\overline{D}}{\sqrt{1+\overline{D}^2}}$ defines a (bounded) Fredholm
operator on the Hilbert space of $L^2$-sections of $\overline{\Sigma}_{\Cl}M
	\to M$. Since the right $\Cl_{n,1}$-action on $\overline{\Sigma}_{\Cl}M$
commutes with $\overline{D}$, this Fredholm operator is $\Cl_{n,1}$-linear.
Moreover, it is invertible if and only if $(g,k) \in \InvI(M)$.

Now, we recall that the space of Fredholm operators on an infinite-dimensional
separable Hilbert space $H$ is a classifying space for
$\KO$-theory~\cite{AtiyahSinger}. More precisely, assume that $H$ is equipped
with an (ample) $\Cl_{r,s}$-action and let $\Fred^{r,s}(H)$ denote the (``ample
component'' of the) space of Fredholm operators commuting with this action and
$\Fred^{*,r,s}(H) \subset \Fred^{r,s}(H)$ its subspace of invertible operators.
Then, for any compact relative CW-complex $(X,Y)$ there is a bijective index
map (\cf\cite[Thm.~2.7 and Lem.~2.8]{Ebert2017}
or~\cite[Thm.~2.5]{Gloeckle2024a})
\begin{align*}
	\ind \colon [(X,Y), (\Fred^{r,s}(H), \Fred^{*,r,s}(H))] \overset{\cong}{\lto} \KO^{s-r}(X,Y)
\end{align*}
which is natural in $(X,Y)$ with respect to continuous maps of pairs. Here, the
left hand side denotes the set of homotopy classes of continuous maps of pairs
$(X,Y) \to (\Fred^{r,s}(H), \Fred^{*,r,s}(H))$.

The idea is to associate to each initial data set $(g,k)$ the bounded transform
of its Dirac-Witten operator, which yields a continuous map of pairs
\begin{gather*}
	\overline{F} \colon (\Ini(M), \InvI(M)) \lto (\Fred^{n,1}(H), \Fred^{*,n,1}(H)),
\end{gather*}
where $H = L^2(M, \overline{\Sigma}_{\Cl}M)$. The technical problem arising
here is that the spinor bundle (and thus $H$) depends on the metric. One way to
overcome this is to fix some $g_0 \in \Met(M)$ and to suitably identify the
spinor bundles for other metrics with the one for $g_0$. A suitable way of
identification is provided by the method of generalized cylinders by Bär,
Gauduchon and Moroianu~\cite{BaerGauduchonMoroianu} and this gives rise to a
well-defined and continuous map~$\overline{F}$
(\cf\cite[Thm.~3.14]{Gloeckle2024a}).

We need two more small ingredients. First, fixing some $(g_0,k_0) \in \Ini(M)$,
any map $S^\ell \to \InvI(M)$ extends by convex combination to a map of pairs
$(D^{\ell + 1}, S^\ell) \to (\Ini(M), \InvI(M))$. On the level of homotopy
classes, this construction yields a bijection
\begin{gather*}
	[S^\ell, \InvI(M)] \overset{\cong}{\lto} [(D^{\ell + 1}, S^\ell),(\Ini(M), \InvI(M))],
\end{gather*}
which is independent of the choice of cone point $(g_0,k_0)$. Second, since
$\KO$ is a cohomology theory, there are canonical isomorphisms
$\KO^{*}(D^{\ell+1}, S^\ell) \cong \widetilde{\KO}^{*}(S^{\ell+1}) \cong
	\widetilde{\KO}^{*-\ell-1}(S^0) \cong \KO^{*-\ell-1}(\pt)$.

\begin{definition}\label{defOladiff}
	The index difference for initial data sets is the map defined by the
	composition
	\begin{align*}
		\oladiff \colon [S^\ell, \InvI(M)] & \overset{\cong}{\lto} [(D^{\ell + 1}, S^\ell),(\Ini(M), \InvI(M))]                            \\
		                                   & \overset{\overline{F}_*}{\lto} [(D^{\ell + 1}, S^\ell),(\Fred^{n,1}(H), \Fred^{*,n,1}(H))]    \\
		                                   & \overset{\ind}{\underset{\cong}{\lto}} \KO^{1-n}(D^{\ell+1}, S^\ell) \cong \KO^{-n-\ell}(\pt)
	\end{align*}
	of the maps discussed above, where $H = L^2(M, \overline{\Sigma}_{\Cl}M)$ for
	some arbitrarily fixed metric $g_0 \in \Met(M)$.
\end{definition}

Since $\sDEC(M) \subseteq \InvI(M)$ and a fixed base point can be easily
forgotten, the index difference defined above immediately gives rise to the map
$\oladiff \colon \pi_\ell(\sDEC(M),(g_0,k_0)) \to \KO^{-n-\ell}(\pt)$
considered in~\cite{Gloeckle2024a}, which we denote by the same symbol.

\printbibliography

@incollection{AmmannGloeckle,
  author    = {Ammann, B. and Glöckle, J.},
  booktitle = {Perspectives in Scalar Curvature},
  editor    = {Gromov, M. L. and Lawson, Jr., H. B.},
  month     = {2},
  publisher = {World Scientific},
  title     = {Dominant energy condition and spinors on {L}orentzian manifolds},
  year      = {2023}
}

@article{ArmsMarsdenMoncrief,
  author  = {Arms, J. M. and Marsden, J. E. and Moncrief, V.},
  journal = {Ann. Physics},
  number  = {1},
  pages   = {81--106},
  title   = {The structure of the space of solutions of {E}instein's equations. {II}. {S}everal {K}illing fields and the {E}instein--{Y}ang--{M}ills equations},
  volume  = {144},
  year    = {1982}
}

@article{AtiyahSinger,
  author  = {Atiyah, M. F. and Singer, I. M.},
  journal = {Publ. math. IHES},
  pages   = {5-26},
  title   = {Index theory for skew-adjoint Fredholm operators},
  volume  = {37},
  year    = {1969}
}

@article{BaerGauduchonMoroianu,
  author  = {Bär, C. and Gauduchon, P. and Moroianu, A.},
  journal = {Math. Z.},
  number  = {3},
  pages   = {545--580},
  title   = {Generalized cylinders in semi-Riemannian and spin geometry},
  volume  = {249},
  year    = {2005}
}

@article{Bartnik,
  author   = {Bartnik, R.},
  fjournal = {Communications in Analysis and Geometry},
  journal  = {Commun. Anal. Geom.},
  number   = {5},
  pages    = {845--885},
  title    = {Phase space for the {Einstein} equations},
  volume   = {13},
  year     = {2005}
}

@incollection{BartnikIsenberg,
  author    = {Bartnik, R. and Isenberg, J.},
  booktitle = {The {E}instein equations and the large scale behavior of gravitational fields},
  editor    = {Chruściel, P. T. and Friedrich, H.},
  pages     = {1--38},
  publisher = {Birkh\"auser, Basel},
  title     = {The constraint equations},
  year      = {2004}
}

@article{BeigChruscielSchoen,
  author  = {Beig, R. and Chruściel, P. T. and Schoen, R.},
  journal = {Ann. Henri Poincar\'e},
  number  = {1},
  pages   = {155--194},
  title   = {{KID}s are non-generic},
  volume  = {6},
  year    = {2005}
}

@book{Besse,
  author    = {Besse, A. L.},
  number    = {10},
  publisher = {Springer},
  series    = {Ergebnisse der Mathematik und ihrer Grenzgebiete, 3.~Folge},
  title     = {Einstein Manifolds},
  year      = {1987}
}

@book{BoossBleecker,
  author    = {Booss, B. and Bleecker, D. D.},
  publisher = {Springer, Cham},
  series    = {Universitext},
  title     = {Topology and analysis. {The} {Atiyah}-{Singer} index formula and gauge- theoretic physics. {Transl}. from the {German} by {D}. {D}. {Bleecker} and {A}. {Mader}},
  year      = {1985}
}

@misc{BamlerKleiner,
  author = {Bamler, R. H. and Kleiner, B.},
  note   = {\arxiv{1909.08710}},
  title  = {Ricci flow and contractibility of spaces of metrics},
  year   = {2019}
}

@article{BotvinnikGilkey,
  author  = {Botvinnik, B. and Gilkey, P. B.},
  journal = {Math. Ann.},
  number  = {3},
  pages   = {507--517},
  title   = {The eta invariant and metrics of positive scalar curvature},
  volume  = {302},
  year    = {1995}
}

@article{BotvinnikEbertRandalWilliams,
  author  = {Botvinnik, B. and Ebert, J. and Randal-Williams, O.},
  journal = {Invent. math.},
  pages   = {749--835},
  title   = {Infinite loop spaces and positive scalar curvature},
  year    = {2014}
}

@article{Bouldin,
  author  = {Bouldin, R.},
  journal = {SIAM Journal on Mathematical Analysis},
  number  = {2},
  pages   = {206-210},
  title   = {The Norm Continuity Properties of Square Roots},
  volume  = {3},
  year    = {1972}
}

@article{BourguignonEbinMarsden,
  author  = {Bourguignon, J.-P. and Ebin, D. G. and Marsden, J. E.},
  journal = {C. R. Acad. Sci., Paris, S{\'e}r. A},
  pages   = {867--870},
  title   = {Sur le noyau des op{\'e}rateurs pseudo-diff{\'e}rentiels {\`a} symbole surjectif et non injectif},
  volume  = {282},
  year    = {1976}
}

@article{BrezisOswald,
  author  = {Br{\'e}zis, H. and Oswald, L.},
  journal = {Nonlinear Anal.},
  number  = {1},
  pages   = {55--64},
  title   = {Remarks on sublinear elliptic equations},
  volume  = {10},
  year    = {1986}
}

@article{CarlottoReview,
  author  = {Carlotto, A.},
  journal = {Living Rev. Relativ.},
  note    = {Id/No 2},
  pages   = {170},
  title   = {The general relativistic constraint equations},
  volume  = {24},
  year    = {2021}
}

@article{CheegerGromov,
  author  = {Cheeger, J. and Gromov, M.},
  journal = {Topology},
  number  = {2},
  pages   = {189--215},
  title   = {$L_2$-Cohomology and group cohomology},
  volume  = {25},
  year    = {1986}
}

@book{ChoquetBruhat,
  author    = {Choquet-Bruhat, Y.},
  publisher = {Oxford University Press},
  title     = {General relativity and the Einstein equations},
  year      = {2009}
}

@incollection{ChoquetBruhatYork,
  address   = {New York},
  author    = {Choquet-Bruhat, Y. and York, Jr., J. W.},
  booktitle = {General relativity and gravitation, Vol.~1},
  editor    = {Held, A.},
  pages     = {99--172},
  publisher = {Plenum Press},
  title     = {The {C}auchy problem},
  year      = {1980}
}

@book{Conway,
  author    = {Conway, J. B.},
  edition   = {2nd ed.},
  publisher = {New York etc.: Springer-Verlag},
  series    = {Grad. Texts Math.},
  title     = {A course in functional analysis.},
  volume    = {96},
  year      = {1990}
}

@article{CrowleySchickSteimle,
  author  = {Crowley, D. and Schick, T. and Steimle, W.},
  journal = {Journal of Topology},
  number  = {4},
  pages   = {1077--1099},
  title   = {Harmonic spinors and metrics of positive curvature via the Gromoll filtration and Toda brackets},
  volume  = {11},
  year    = {2018}
}

@incollection{CurryGover,
  address   = {Cambridge},
  author    = {Curry, S. N. and Gover, A. R.},
  booktitle = {Asymptotic analysis in general relativity},
  editor    = {Daud{\'e}, T. and H{\"a}fner, D. and Nicolas, J.-P.},
  pages     = {86--170},
  publisher = {Cambridge University Press},
  series    = {London Math. Soc. Lecture Note Ser.},
  title     = {An introduction to conformal geometry and tractor calculus, with a view to applications in general relativity},
  volume    = {443},
  year      = {2018}
}

@article{DahlGicquaudHumbert,
  author  = {Dahl, M. and Gicquaud, R. and Humbert, E.},
  journal = {Duke Math. J.},
  number  = {14},
  pages   = {2669--2697},
  title   = {{A limit equation associated to the solvability of the vacuum {E}instein constraint equations by using the conformal method}},
  volume  = {161},
  year    = {2012}
}

@article{DixmierDouady,
  author    = {Dixmier, J. and Douady, A.},
  journal   = {Bulletin de la Soci\'et\'e Math\'ematique de France},
  pages     = {227--284},
  publisher = {Soci\'et\'e math\'ematique de France},
  title     = {Champs continus d{\textquoteright}espaces hilbertiens et de {$C^\ast $}-alg\`ebres},
  volume    = {91},
  year      = {1963}
}

@article{Ebert2017,
  author  = {Ebert, J.},
  journal = {Trans. Amer. Math. Soc.},
  pages   = {7469--7507},
  title   = {The two definitions of the index difference},
  volume  = {369},
  year    = {2017}
}

@article{FischerMarsden,
  author  = {Fischer, A. E. and Marsden, J. E.},
  journal = {Bull. Amer. Math. Soc.},
  number  = {5},
  pages   = {997--1003},
  title   = {Linearization stability of the {E}instein equations},
  volume  = {79},
  year    = {1973}
}

@article{FischerMoncrief,
  author  = {Fischer, A. E. and Moncrief, V.},
  journal = {Classical Quantum Gravity},
  number  = {21},
  pages   = {4493--4515},
  title   = {The reduced {E}instein equations and the conformal volume collapse of $3$-manifolds},
  volume  = {18},
  year    = {2001}
}

@article{GarciaHehlHeinickeMacias,
  author   = {Garc{\'{\i}}a, A. A. and Hehl, F. W. and Heinicke, C. and Mac{\'{\i}}as, A.},
  doi      = {10.1088/0264-9381/21/4/024},
  fjournal = {Classical and Quantum Gravity},
  issn     = {0264-9381},
  journal  = {Classical Quantum Gravity},
  language = {English},
  number   = {4},
  pages    = {1099--1118},
  title    = {The {Cotton} tensor in {Riemannian} spacetimes},
  volume   = {21},
  year     = {2004},
  zbl      = {1045.83051},
  zbmath   = {2070692}
}

@article{GicquaudLichnerowicz,
  author  = {Gicquaud, R.},
  journal = {J. Math. Phys.},
  number  = {2},
  pages   = {Paper No. 022501, 12},
  title   = {Existence of solutions to the {L}ichnerowicz equation: a new proof},
  volume  = {63},
  year    = {2022}
}

@article{GicquaudSmallTT,
  author   = {Gicquaud, Romain},
  title    = {Solutions to the {Einstein} constraint equations with a small {TT}-tensor and vanishing {Yamabe} invariant},
  fjournal = {Annales Henri Poincar{\'e}},
  journal  = {Ann. Henri Poincar{\'e}},
  issn     = {1424-0637},
  volume   = {22},
  number   = {7},
  pages    = {2407--2435},
  year     = {2021},
  language = {English},
  doi      = {10.1007/s00023-021-01036-1},
  zbmath   = {7373884},
  zbl      = {1471.53036}
}

@unpublished{GicquaudConformal,
  author = {Gicquaud, R.},
  note   = {in preparation},
  title  = {{The conformal method is not conformal}}
}

@book{GilbargTrudinger,
  address   = {Berlin},
  author    = {Gilbarg, D. and Trudinger, N. S.},
  note      = {Reprint of the 1998 edition},
  publisher = {Springer-Verlag},
  series    = {Classics in Mathematics},
  title     = {Elliptic partial differential equations of second order},
  year      = {2001}
}

@article{Gloeckle2021,
  title    = {An Enlargeability Obstruction for Spacetimes with both Big Bang and Big Crunch},
  author   = {Glöckle, J.},
  journal  = {Math. Res. Lett.},
  fjournal = {Mathematical Research Letters},
  year     = {2024},
  volume   = {31},
  number   = {5},
  pages    = {1435--1469},
  doi      = {10.4310/MRL.241211042922}
}

@article{Gloeckle2024a,
  author  = {Glöckle, J.},
  journal = {Math. Ann.},
  month   = {02},
  pages   = {1323--1355},
  title   = {On the space of initial values strictly satisfying the dominant energy condition},
  volume  = {388},
  year    = {2024}
}

@article{GromovLawson1980,
  author  = {Gromov, M. and Lawson, Jr., H. B.},
  journal = {Ann. of Math. (2)},
  number  = {2},
  pages   = {209--230},
  title   = {Spin and scalar curvature in the presence of a fundamental group. {I}},
  volume  = {111},
  year    = {1980}
}

@article{GromovLawson1980b,
  author  = {Gromov, M. and Lawson, Jr., H. B.},
  journal = {Ann. of Math. (2)},
  number  = {3},
  pages   = {423--434},
  title   = {The classification of simply connected manifolds of positive scalar curvature},
  volume  = {111},
  year    = {1980}
}

@article{HankeSchickSteimle,
  author  = {Hanke, B. and Schick, T. and Steimle, W.},
  journal = {Publ. math. IHES},
  number  = {1},
  pages   = {335--367},
  title   = {The space of metrics of positive scalar curvature},
  volume  = {120},
  year    = {2014}
}

@book{HatcherKT,
  author = {Hatcher, A.},
  note   = {Available at \url{https://pi.math.cornell.edu/~hatcher/VBKT/VBpage.html}},
  title  = {Vector Bundles and $K$-Theory},
  year   = {2017}
}

@article{Hitchin,
  author  = {Hitchin, N.},
  journal = {Advances in Mathematics},
  number  = {1},
  pages   = {1-55},
  title   = {Harmonic Spinors},
  volume  = {14},
  year    = {1974}
}

@article{HNT1,
  author  = {Holst, M. and Nagy, G. and Tsogtgerel, G.},
  journal = {Phys. Rev. Lett.},
  number  = {16},
  pages   = {161101, 4},
  title   = {{Far-from-constant mean curvature solutions of {E}instein's constraint equations with positive {Y}amabe metrics}},
  volume  = {100},
  year    = {2008}
}

@article{HNT2,
  author  = {Holst, M. and Nagy, G. and Tsogtgerel, G.},
  journal = {Comm. Math. Phys.},
  number  = {2},
  pages   = {547--613},
  title   = {{Rough solutions of the {E}instein constraints on closed manifolds without near-{CMC} conditions}},
  volume  = {288},
  year    = {2009}
}

@book{HurewiczWallman,
  author    = {Hurewicz, W. and Wallman, H.},
  address   = {Princeton, NJ},
  booktitle = {Dimension theory},
  edition   = {Rev. ed.},
  publisher = {Princeton University Press},
  series    = {Princeton mathematical series},
  title     = {Dimension theory},
  year      = {1948}
}

@article{Isenberg,
  author  = {Isenberg, J.},
  journal = {Classical and Quantum Gravity},
  number  = {9},
  pages   = {2249--2274},
  title   = {Constant mean curvature solutions of the {E}instein constraint equations on closed manifolds},
  volume  = {12},
  year    = {1995}
}

@book{Kammeyer,
  author    = {Kammeyer, H.},
  publisher = {Springer Nature Switzerland},
  number    = {2247},
  series    = {Lecture Notes in Mathematics},
  title     = {Introduction to $\ell^2$-invariants},
  year      = {2019}
}

@book{LawsonMichelsohn,
  address   = {Princeton},
  author    = {Lawson, Jr., H. B. and Michelsohn, M.-L.},
  publisher = {Princeton University Press},
  title     = {Spin Geometry},
  year      = {1989}
}

@article{Lichnerowicz1944,
  author  = {Lichnerowicz, A.},
  journal = {J. Math. Pures Appl. (9)},
  pages   = {37--63},
  title   = {L'int\'egration des \'equations de la gravitation relativiste et le probl\`eme des $n$ corps},
  volume  = {23},
  year    = {1944}
}

@article{Marques,
  author  = {Marques, F. C.},
  journal = {Ann. of Math. (2)},
  number  = {2},
  pages   = {815--863},
  title   = {Deforming three-manifolds with positive scalar curvature},
  volume  = {176},
  year    = {2012}
}

@article{MaxwellDrift,
  author  = {Maxwell, D.},
  journal = {Comm. Anal. Geom.},
  note    = {\arxiv{1407.1467}},
  number  = {1},
  pages   = {207--281},
  title   = {Initial data in general relativity described by expansion, conformal deformation and drift},
  volume  = {29},
  year    = {2021}
}

@article{MaxwellModel,
  author  = {Maxwell, D.},
  journal = {Comm. Math. Phys.},
  note    = {\arxiv{0909.5674}},
  number  = {3},
  pages   = {697--736},
  title   = {A model problem for conformal parameterizations of the {E}instein constraint equations},
  volume  = {302},
  year    = {2011}
}

@article{MaxwellNonCMC,
  author  = {Maxwell, D.},
  journal = {Math. Res. Lett.},
  number  = {4},
  pages   = {627--645},
  title   = {{A class of solutions of the vacuum {E}instein constraint equations with freely specified mean curvature}},
  volume  = {16},
  year    = {2009}
}

@article{Moncrief,
  author  = {Moncrief, V.},
  journal = {J. Math. Phys.},
  number  = {3},
  pages   = {493--498},
  title   = {Spacetime symmetries and linearization stability of the {E}instein equations. {I}},
  volume  = {16},
  year    = {1975}
}

@book{MunkresDiffTop,
  author    = {Munkres, J. R.},
  address   = {Princeton, NJ},
  booktitle = {Elementary differential topology},
  publisher = {Princeton University Press},
  series    = {Annals of mathematics studies, 54},
  title     = {Elementary differential topology},
  year      = {1963}
}

@book{Munkres,
  address   = {Upper Saddle River, NJ},
  author    = {Munkres, J. R.},
  edition   = {2},
  publisher = {Prentice Hall},
  title     = {Topology},
  year      = {2000}
}

@article{ParkerTaubes,
  author  = {Parker, T. and Taubes, C. H.},
  journal = {Comm. Math. Phys. },
  pages   = {223--238},
  title   = {{On Witten's proof of the positive energy theorem}},
  volume  = {84},
  year    = {1982}
}

@book{Ringstroem,
  author    = {Ringström, H.},
  month     = {05},
  publisher = {Oxford University Press},
  title     = {{On the Topology and Future Stability of the Universe}},
  year      = {2013}
}

@incollection{Rosenberg,
  author    = {Rosenberg, J.},
  booktitle = {Surveys in differential geometry. {V}ol. {XI}},
  editor    = {Cheeger, J. and Grove, K.},
  pages     = {259--294},
  series    = {Surv. Differ. Geom.},
  title     = {Manifolds of positive scalar curvature: a progress report},
  volume    = {11},
  year      = {2007}
}

@article{Ruberman,
  author  = {Ruberman, D.},
  journal = {Geom. Topol.},
  pages   = {895--924},
  title   = {Positive scalar curvature, diffeomorphisms and the {S}eiberg-{W}itten invariants},
  volume  = {5},
  year    = {2001}
}

@article{SchoenYau,
  author  = {Schoen, R. and Yau, S.-T.},
  journal = {manuscripta math.},
  pages   = {159--183},
  title   = {On the structure of manifolds with positive scalar curvature},
  volume  = {28},
  year    = {1979}
}

@article{Stolz,
  author  = {Stolz, S.},
  journal = {Ann. of Math. (2)},
  number  = {3},
  pages   = {511--540},
  title   = {Simply connected manifolds of positive scalar curvature},
  volume  = {136},
  year    = {1992}
}

@book{Taylor3,
  address   = {New York},
  author    = {Taylor, M. E.},
  edition   = {Second},
  pages     = {xxii+715},
  publisher = {Springer},
  series    = {Applied Mathematical Sciences},
  title     = {Partial differential equations {III}. {N}onlinear equations},
  volume    = {117},
  year      = {2011}
}

@book{Weinberg,
  address   = {New York},
  author    = {Weinberg, S.},
  publisher = {John Wiley \& Sons},
  title     = {Gravitation and cosmology: principles and applications of the general theory of relativity},
  year      = {1972}
}

@book{Wald,
  author    = {Wald, R. M.},
  title     = {General relativity},
  publisher = {University of Chicago Press},
  address   = {Chicago},
  year      = {1984}
}

@article{Witten,
  author  = {Witten, E.},
  journal = {Comm. Math. Phys.},
  number  = {3},
  pages   = {381--402},
  title   = {A new proof of the positive energy theorem},
  volume  = {80},
  year    = {1981}
}

@book{Yano,
  address   = {New York},
  author    = {Yano, K.},
  pages     = {ix+156},
  publisher = {Marcel Dekker Inc.},
  series    = {Pure and Applied Mathematics, No. 1},
  title     = {Integral formulas in {R}iemannian geometry},
  year      = {1970}
}
\end{document}